%% file: main.tex
\documentclass[11pt]{article}
\usepackage[letterpaper,hmargin=0.95in,vmargin=0.95in]{geometry}

\usepackage[ruled]{algorithm2e} 

\SetAlFnt{\small}
\SetAlCapFnt{\small}
\SetAlCapNameFnt{\small}
\SetAlCapHSkip{0pt}
\IncMargin{-\parindent}

\usepackage{amsmath,amsthm,amsfonts,latexsym,verbatim,threeparttable,amssymb,amsfonts,rotating,pifont}

\usepackage[numbers]{natbib}
\usepackage{caption}
\usepackage{adjustbox,multirow}

\usepackage{float}

\input{packages_article}
\input{figure_env}
\input{symbols_env}
\usepackage{verbatim,color}
\input{symbols}

 \newtheorem{theorem}{Theorem}[section]

\newtheorem{lemma}[theorem]{Lemma}

 \theoremstyle{definition}
 \newtheorem{definition}[theorem]{Definition}

 \newtheorem{notation}[theorem]{Notation}
 
\usepackage{color}
\usepackage{framed}

\begin{document}
\title{\bf A Survey on Perfectly-Secure Verifiable Secret-Sharing}

\author{Anirudh Chandramouli\footnote{International Institute of Information Technology, Bangalore India.
 Email: {\tt{anirudh.c@iiitb.ac.in}}.}   \and  Ashish Choudhury\footnote{International Institute of Information Technology, Bangalore India.
 Email: {\tt{ashish.choudhury@iiitb.ac.in}}. This research is an outcome of the R \& D work undertaken in the project under the Visvesvaraya PhD  Scheme of  
 Ministry of Electronics \& Information Technology, Government of India, being implemented by Digital India Corporation (formerly Media Lab Asia).} 
  \and Arpita Patra\footnote{Indian Institute of Science, Bangalore India. Email: {\tt arpita@iisc.ac.in}.  Arpita Patra would like to acknowledge financial support from SERB MATRICS (Theoretical Sciences) Grant 2020 and Google India AI/ML Research Award 2020.
}
 }
\maketitle
\begin{abstract}
 {\it Verifiable Secret-Sharing} (VSS) is a fundamental primitive in secure distributed computing. It is used as a building block in several distributed computing tasks, such
  as Byzantine agreement and secure multi-party computation.  In this article, we consider VSS schemes with {\it perfect} security,
  tolerating {\it computationally unbounded} adversaries. 
  We comprehensively survey the existing perfectly-secure VSS schemes in three different communication settings, namely synchronous, asynchronous and hybrid setting
  and provide full details of the existing schemes in these settings. The aim of this survey is to provide a clear knowledge and foundation
  to researchers who are interested in knowing and extending the state-of-the-art perfectly-secure VSS schemes.
\end{abstract}

\input{intro}

\begin{center}
{\bf Part I :  Synchronous Communication Setting}
\end{center}

\input{SPrelim}

\input{SThresholdBounds}

\input{SThreshold}

\begin{center}
{\bf Part III :  Asynchronous Communication Setting}
\end{center}

\input{APrelim}

\input{AThreshold}

\begin{center}
{\bf Part III :  Hybrid Communication Setting}
\end{center}

\input{HPrelim}
\input{HThreshold}

\section{Open Problems}
 We identify the following open problems in the domain of {\it perfectly-secure} VSS.
 \begin{myitemize}
 \item The communication complexity of existing {\it efficient} VSS schemes with a $3$-round sharing phase is 
  $n$ times more compared to the VSS schemes which allow four or more rounds in the sharing phase
  (see Table \ref{tab:VSSSummary}). It is interesting to see if one can close this gap. 
  \item We are unaware of any non-trivial lower bound on the communication complexity of perfectly-secure VSS schemes. One could explore
   deriving any non-trivial lower bound.
  \item The {\it broadcast complexity} (namely the number of bits broadcasted) of all VSS schemes  with $n > 3t$
   is proportional to the number of values $L$ shared by the scheme (namely $\Order(L \cdot n^3 \log{|\F|})$ bits).
    This is unlike the AVSS scheme of Patra et al.~\cite{CHP13}, where the broadcast complexity is $\Order(n^2 \log n)$ for sharing $L$ values.
    Given that each instance of broadcast needs to be emulated by running a costly {\it reliable broadcast} (RB) protocol,
    it is interesting to see if one can design a 
    VSS scheme with $n > 3t$ where the broadcast complexity is {\it independent} of $L$.
  \item The $\AKP$ VSS scheme \cite{AKP20} uses the broadcast channel during two rounds, while the optimal usage is {\it one} round.  
  One could explore to achieve the same
  properties as the $\AKP$ scheme, with the broadcast channel being used only during one round.
 \end{myitemize}

\bibliographystyle{plain}
\bibliography{main}


\end{document}

%% file: packages_article.tex
\usepackage{enumitem}
\usepackage{latexsym}
\usepackage{mathtools}
\usepackage{multirow}
\usepackage{multicol}
\usepackage{mwe}
\usepackage{makecell}
\usepackage{pgfplots}
\usepackage[section]{placeins}
\usepackage{pifont}
\usepackage{ragged2e}
\usepackage{rotating}
\usepackage{subcaption}
\usepackage{tabu}
\usepackage{tabularx}
\usepackage{threeparttable}
\usepackage{tikz}
\usepackage[normalem]{ulem}
\usepackage{url}
\usepackage{verbatim}
\usepackage{wrapfig}
\usepackage{xcolor}
\usepackage{xspace}
\usepackage[framemethod=tikz]{mdframed}
\usepackage{algpseudocode}

%% file: figure_env.tex
\definecolor{DarkRed}{rgb}{0.5,0.1,0.1}
\definecolor{ForestGreen}{rgb}{0.1333,0.5451,0.1333}
\definecolor{BabyBlue}{rgb}{0.54, 0.81, 0.94}
\definecolor{Cobalt}{rgb}{0.0, 0.28, 0.67}
\definecolor{Blond}{rgb}{0.98, 0.94, 0.75}
\usepackage{tikz}
\usetikzlibrary{decorations.pathreplacing}
\makeatletter
\DeclareRobustCommand{\rvdots}{%
  \vbox{
    \baselineskip4\p@\lineskiplimit\z@
    \kern-\p@
    \hbox{.}\hbox{.}\hbox{.}
  }}
\makeatother

%% file: symbols.tex
\newcommand{\F}{\ensuremath{\mathbb{F}}}
\newcommand{\defined}{\ensuremath{\stackrel{def}{=}}}
\newcommand{\view}{\ensuremath{\mathsf{view}}}

\newcommand{\Order}{\ensuremath{\mathcal{O}}}
\newcommand{\starF}{\ensuremath{{F}^{\star}}}
\newcommand{\starq}{\ensuremath{{q}^{\star}}}
\newcommand{\stars}{\ensuremath{{s}^{\star}}}
\newcommand{\starf}{\ensuremath{{f}^{\star}}}
\newcommand{\Gr}{\ensuremath{\mathbb{G}}}
\newcommand{\STAR}{\ensuremath{\mathsf{star}}}
\newcommand{\param}{\ensuremath{L}}
\newcommand{\default}{\ensuremath{\perp}}
\newcommand{\intersection}{\ensuremath{\cap}}

\newcommand{\D}{\ensuremath{\mathsf{D}}}
\newcommand{\Adv}{\ensuremath{\mathsf{Adv}}}
\newcommand{\Sender}{\ensuremath{\mathsf{S}}}

\newcommand{\Partyset}{\ensuremath{\mathcal{P}}}
\newcommand{\Honest}{\ensuremath{\mathcal{H}}}
\newcommand{\Happy}{\ensuremath{\mathcal{H}}}
\newcommand{\Bad}{\ensuremath{\mathcal{C}}}

\newcommand{\VCORE}{\ensuremath{\mathcal{V}}}
\newcommand{\WCORE}{\ensuremath{\mathcal{W}}}

\newcommand{\BC}{\ensuremath{\mathcal{BC}}}

\newcommand{\SecretSpace}{\ensuremath{\mathcal{S}}}
\newcommand{\R}{\ensuremath{\mathcal{R}}}
\newcommand{\Unhappy}{\ensuremath{\mathcal{UH}}}
\newcommand{\Group}{\ensuremath{\mathcal{G}}}
\newcommand{\CSet}{\ensuremath{\mathcal{C}}}
\newcommand{\DSet}{\ensuremath{\mathcal{D}}}
\newcommand{\ESet}{\ensuremath{\mathcal{E}}}
\newcommand{\FSet}{\ensuremath{\mathcal{F}}}
\newcommand{\RCORE}{\ensuremath{\mathsf{CORE}}}
\newcommand{\RC}{\ensuremath{\mathcal{RC}}}
\newcommand{\CC}{\ensuremath{\mathcal{CC}}}
\newcommand{\AList}{\ensuremath{\mathcal{A}}}
\newcommand{\BList}{\ensuremath{\mathcal{B}}}
\newcommand{\CList}{\ensuremath{\mathcal{C}}}
\newcommand{\DList}{\ensuremath{\mathcal{D}}}
\newcommand{\Support}{\ensuremath{\mathcal{SS}}}

\newcommand{\OK}{\ensuremath{\texttt{OK}}}
\newcommand{\Accuse}{\ensuremath{\texttt{accuse}}}
\newcommand{\complaint}{\ensuremath{\texttt{complaint}}}
\newcommand{\NEQ}{\ensuremath{\texttt{NEQ}}}
\newcommand{\EQ}{\ensuremath{\texttt{EQ}}}
\newcommand{\disagreef}{\ensuremath{\texttt{disagree}\mbox{-}\texttt{row}}}
\newcommand{\disagreeg}{\ensuremath{\texttt{disagree}\mbox{-}\texttt{column}}}
\newcommand{\agreef}{\ensuremath{\texttt{agree}\mbox{-}\texttt{row}}}
\newcommand{\agreeg}{\ensuremath{\texttt{agree}\mbox{-}\texttt{column}}}
\newcommand{\agree}{\ensuremath{\texttt{agree}}}
\newcommand{\disagree}{\ensuremath{\texttt{disagree}}}

\newcommand{\Sh}{\ensuremath{\mathsf{Sh}}}
\newcommand{\Rec}{\ensuremath{\mathsf{Rec}}}

\newcommand{\Gen}{\ensuremath{\mathsf{G}}}
\newcommand{\Recover}{\ensuremath{\mathsf{R}}}

\newcommand{\Shamir}{\ensuremath{\mathsf{Sha}}}

\newcommand{\OEC}{\ensuremath{\mathsf{OEC}}}
\newcommand{\RSDec}{\ensuremath{\mathsf{RS}\mbox{-}\mathsf{Dec}}}

\newcommand{\BGW}{\ensuremath{\mathsf{7BGW\mbox{-}VSS}}}
\newcommand{\ModBGW}{\ensuremath{\mathsf{5BGW\mbox{-}VSS}}}
\newcommand{\GIKR}{\ensuremath{\mathsf{GIKR}}}
\newcommand{\GIKRI}{\ensuremath{\mathsf{4GIKR\mbox{-}VSS}}}
\newcommand{\GIKRII}{\ensuremath{\mathsf{3GIKR\mbox{-}VSS}}}
\newcommand{\GIKRIII}{\ensuremath{\mathsf{2GIKR\mbox{-}VSS}}}
\newcommand{\FGGRS}{\ensuremath{\mathsf{3FGGRS\mbox{-}VSS}}}
\newcommand{\WSSSh}{\ensuremath{\mathsf{WSS}{-}\mathsf{Sh}}}
\newcommand{\WSSRec}{\ensuremath{\mathsf{WSS}{-}\mathsf{Rec}}}
\newcommand{\WSS}{\ensuremath{\mathsf{WSS}}}
\newcommand{\FGGRSWSS}{\ensuremath{\mathsf{3FGGRS\mbox{-}WSS}}}
\newcommand{\KKK}{\ensuremath{\mathsf{3KKK\mbox{-}VSS}}}
\newcommand{\KKKWSS}{\ensuremath{\mathsf{3KKK\mbox{-}WSS}}}
\newcommand{\WC}{\ensuremath{\mathsf{WC}}}
\newcommand{\AKP}{\ensuremath{\mathsf{3AKP\mbox{-}VSS}}}
\newcommand{\AKPWC}{\ensuremath{\mathsf{3AKP\mbox{-}WC}}}
\newcommand{\OneVSS}{\ensuremath{\mathsf{1GIKR\mbox{-}VSS}}}

\newcommand{\BCG}{\ensuremath{\mathsf{BCG\mbox{-}AVSS}}}
\newcommand{\PCR}{\ensuremath{\mathsf{PCR\mbox{-}AVSS}}}
\newcommand{\CHP}{\ensuremath{\mathsf{CHP\mbox{-}AVSS}}}

\newcommand{\AWPS}{\ensuremath{\mathsf{WPS}}}
\newcommand{\PR}{\ensuremath{\mathsf{PR}}}

\newenvironment{myitemize}
{\begin{list}{$\bullet$}{ 
\itemindent=-0.1in
\itemsep=0.0in
\parsep=0.0in
\topsep=0.0in
\partopsep=0.0in}}{\end{list}}
\newcounter{itemcount}

\newenvironment{myenumerate}
{\setcounter{itemcount}{0}\begin{list}
{\arabic{itemcount}.}{\usecounter{itemcount} \itemindent=-0.2cm
\itemsep=0.0in
\parsep=0.0in
\topsep=5pt
\partopsep=0.0in}}{\end{list}}

\newenvironment{mydescription}
{\setcounter{itemcount}{0}\begin{list}
{\arabic{itemcount}.}{\usecounter{itemcount} \itemindent=-0.5cm
\itemsep=0.0in
\parsep=0.0in
\topsep=5pt
\partopsep=0.0in}}{\end{list}}

%% file: intro.tex
\section{Introduction}
\label{sec:intro}
A central concept in cryptographic protocols is that of {\it Secret Sharing} (SS) \cite{Sha79,Bla79}. Consider a set $\Partyset = \{P_1, \ldots, P_n \}$
 of mutually distrusting parties, where the distrust is modeled by a centralized {\it adversary}, who can control up to $t$ parties. Then 
  a SS scheme allows a designated {\it dealer} $\D \in \Partyset$  to share a secret
  $s$ among $\Partyset$, by providing each $P_i$ a {\it share} of the secret $s$. The sharing is done in such a way that the adversary
   controlling any subset of at most $t$ share-holders fails
   to learn anything about $s$, while any subset of at least $t + 1$ share-holders can jointly recover $s$. In a SS scheme, it is assumed that {\it all} the parties
   including the ones under the adversary's control follow the protocol instructions correctly (thus, the adversary is assumed to only {\it eavesdrop}
   the computation and communication of the parties under its control). VSS \cite{CGMA85} extends the notion of SS to the more 
   powerful {\it malicious/active}
   adversarial model, where the adversary can completely dictate the behaviour of the parties under its control during a protocol execution. Moreover, 
 $\D$ is allowed to be potentially corrupted. A VSS scheme consists of a sharing phase and a reconstruction phase, each implemented by a publicly-known protocol.
   During the sharing phase, $\D$ shares its secret in a {\it verifiable} fashion, which is later reconstructed during the reconstruction phase. If $\D$ is {\it honest}, then the privacy
   of its secret is maintained during the sharing phase and the shared secret is later robustly reconstructed, irrespective of the behaviour of the corrupt parties. 
   The interesting property of VSS is the {\it verifiability} property, which guarantees that even if $\D$ is {\it corrupt},  it has ``consistently/correctly" shared some value among the
   parties and the same value is later reconstructed. One can interpret VSS as a distributed commitment, where, during the sharing phase, $\D$ publicly commits to a private input
   (known only to $\D$) and later during the reconstruction phase, the committed value is reconstructed publicly (even if $\D$ does not cooperate). 
    Due to its central importance in secure distributed-computing tasks, such as secure {\it multi-party computation} (MPC) \cite{BGW88,RB89}
    and Byzantine agreement (BA) \cite{FM97}, VSS has been studied in various settings, based on the following categorizations. 
    \begin{myitemize}
    \item {\bf Conditional vs Unconditional Security}: If the adversary is {\it computationally bounded} (where it is allowed to perform only polynomial amount of computations), then 
    the notion of security achieved is conditional/cryptographic \cite{Ped91,BKP11,BCG19},
     whereas unconditionally-secure VSS provides security even against a {\it computationally unbounded}
    adversary. Unconditionally-secure VSS can be further categorized as {\it perfectly-secure} where all security guarantees are achieved in an error-free fashion \cite{BGW88},
    and {\it statistically-secure} where a negligible error is allowed  \cite{RB89,CDDHR99,KPR10}.
    \item {\bf Type of Communication}: Here we have three categories. The {\it synchronous} model assumes that the parties
     are synchronized through a global clock and there are strict (publicly-known) upper bounds on the message delays \cite{BGW88,RB89,GIKR01,FGGRS06,KKK09,AKP20}. 
     The second category is the 
     {\it asynchronous}
     model \cite{BCG93,BKR94,BH07,PCR14,PCR15,CP17},
      where the parties are not synchronized and where the messages can be arbitrarily (but finitely)
     delayed. A major challenge in the asynchronous model is that a {\it slow} sender party (whose messages are delayed) cannot be distinguished from a 
     {\it corrupt} sender who does not send messages at all. Due to this inherent phenomenon, asynchronous VSS (AVSS) 
      protocols are more complicated than their synchronous
     counter-parts. The third category is the {\it hybrid} communication setting \cite{PR18}, which is a mix of the synchronous and asynchronous models. Namely, the
     protocol starts with a few initial synchronous rounds, followed by a completely asynchronous execution. The main motivation for 
      considering a hybrid setting is to ``bridge" the feasibility and
     efficiency gaps between completely synchronous and completely asynchronous protocols.
    \item {\bf Corruption Capacity}: Most of the works on VSS assume a {\it threshold} adversary which
     can corrupt any subset of
    $t$ parties. A {\it non-threshold} adversary \cite{CDM00,Mau06,CKP11,CP20} is a more generalized adversary, where the corruption capacity of
    adversary is specified by a publicly-known {\it adversary structure}, which is a collection of potentially corrupt subsets of parties. During the protocol execution,
   the  adversary can choose any subset from the collection for corruption. 
    \end{myitemize}
\noindent \paragraph{\bf Our Contributions and Paper Organization} We provide a comprehensive survey of all the existing {\it perfectly-secure} VSS schemes tolerating a 
 {\it threshold} adversary.
 We cover three communication settings, namely synchronous, asynchronous and hybrid. These schemes are designed over a period of three decades. The nuances, subtleties and foundational ideas involved  in these works need a holistic and unified treatment, which is the focus of this work.  
 This survey is structured to provide an easy digest of the perfectly-secure VSS schemes. 
   Through this survey, we hope to provide a clear knowledge and foundation
   to researchers who are interested in knowing and extending the state-of-the-art perfectly-secure VSS schemes.
  The survey is divided into three parts, each dealing with a separate communication model. 
  We do {\it  not} survey SS schemes and their share sizes, for which the interested readers are referred to the survey by Beimel \cite{Bei11}.

%% file: SPrelim.tex
\section{Preliminaries and Definitions}
\label{sec:SPrelim}
Throughout part I, we consider a {\it synchronous} communication setting, where the parties in $\Partyset$ are
 connected by pair-wise private and authentic channels.
  The distrust in the system is modelled by a {\it computationally unbounded} adversary $\Adv$, who 
  can corrupt at most $t$ parties during the execution of a protocol in a malicious/Byzantine fashion 
    and force them to behave in any arbitrary manner. 
  The parties
 under the control of $\Adv$ are called {\it corrupt/malicious}, while the parties not under $\Adv$'s control are called {\it honest}.
  We assume a
  {\it static} adversary, who decides the set of corrupt parties at the beginning of a protocol. However, following \cite{KLR10} the protocols discussed can be proved to be secure
   even against an {\it adaptive} adversary, who can corrupt parties as the protocol proceeds.

   We also assume
   the presence of a {\it broadcast} channel, which allows any designated party to send some message identically to all the parties.
   A protocol in the synchronous setting operates as a sequence of {\it rounds}. In each round, a party can (privately) send messages to other parties  and broadcast a message.
   Moreover in a given round, each party can simultaneously use the broadcast channel.
   The messages sent or broadcast by a party is determined by its input, random coins and the messages received from other parties in previous rounds. 
     The {\it view} of a party during a protocol execution consists of its inputs, random coins and all the messages received 
     by the party throughout the protocol execution. The {\it view} of $\Adv$ is the collection of the views of the corrupt parties.
\paragraph{\bf Structure of a VSS Scheme.}
Following \cite{GIKR01}, VSS schemes can be structured into two phases. A {\it sharing phase} executed by a protocol $\Sh$, followed by a 
 {\it reconstruction phase} executed by a protocol $\Rec$. While the goal of $\Sh$ is to share a secret held by a designated {\it dealer} $\D \in \Partyset$,
  the aim of $\Rec$ is to reconstruct back the shared secret. Specifically, during $\Sh$, the input of $\D$ is some secret 
   $s \in \SecretSpace$, where $\SecretSpace$ is some publicly-known {\it secret-space} which is the set of all possible $\D$'s secrets. Additionally, the parties may have random inputs
   for the protocol. Let $\view_i$ denote the view of $P_i$ at the end of $\Sh$. Based on $\view_i$, each $P_i$ outputs a {\it share} $s_i$, as determined by $\Sh$.
   
   During $\Rec$, each $P_i$ reveals a subset of $\view_i$, as per $\Rec$. 
    The parties then apply a reconstruction
   function on the revealed views, as per $\Rec$ and reconstruct some output. 
    Following \cite{KKK09}, we say that {\it round-complexity} of $\Sh$ (resp. $\Rec$) is $(R, R')$, if $\Sh$ (resp. $\Rec$) involves total $R$ rounds 
    and among these $R$ rounds, the broadcast channel is used for $R'$ rounds. 
   By {\it communication complexity} of a protocol, we mean the total number of bits communicated by the honest parties
   in the protocol. 
\subsection{Definitions}
A {\it $t\mbox{-out-of-}n$ secret-sharing} (SS) scheme is a pair of functions $(\Gen, \Recover)$. 
  While $\Gen$ is probabilistic, $\Recover$ is deterministic. 
    Function $\Gen$ generates shares for the input secret, while $\Recover$ maps the shares back to the secret.
    The shares are generated in such a way that the probability distribution of any set of $t$ shares is independent of the secret, while any set of
    $t + 1$ shares uniquely determines the secret.
\begin{definition}[\bf $t\mbox{-out-of-}n$ secret-sharing \cite{Gol04}]
\label{def:SS}
It is a pair of algorithms $(\Gen, \Recover)$, such that:
\begin{myitemize}
\item[--] {\bf Syntax}: The {\it share-generation function} $\Gen$ takes input
 a secret $s$ and some randomness $q$ and outputs a vector of $n$ shares $(s_1, \ldots, s_n)$.
   The {\it recovery function} $\Recover$ takes input a set of $t + 1$ shares corresponding to $t+1$ indices
  $\{i_1, \ldots, i_{t + 1} \} \subset \{1, \ldots, n \}$
  and outputs a value.  
\item[--]{\bf Correctness}: For any $s \in \SecretSpace$ and any vector $(s_1, \ldots, s_n)$ where 
 $(s_1, \ldots, s_n) = \Gen(s, q)$ for some randomness $q$, the condition $\Recover(s_{i_1}, \ldots, s_{i_{t+1}}) = s$ holds for any subset
  $\{i_1, \ldots, i_{t + 1} \} \subset \{1, \ldots, n \}$.
\item[--] {\bf Privacy}: For any subset of $t$ indices, the probability distribution of the shares corresponding to these indices is independent of
 the underlying secret. That is, for any $I = \{i_1, \ldots, i_t \} \subset \{1, \ldots, n \}$, let
 $g_I(s) \defined (s_{i_1}, \ldots, s_{i_t})$, where $(s_1, \ldots, s_n) = \Gen(s, q)$ for some randomness $q$. Then we require that
 for every index-set $I$ where $|I| = t$, the random variables $g_I(s)$ and $g_I(s')$ are identically distributed, for every $s, s' \in \SecretSpace$, where
 $s \neq s'$.
\end{myitemize}
\end{definition}    
\begin{definition}
\label{def:SSShared}
Let $\Pi = (\Pi_{\Gen}, \Pi_{\Recover})$ be a $t\mbox{-out-of-}n$ secret-sharing scheme. Then we say that {\it a value $s$ is secret-shared among $\Partyset$ as per
 $\Pi$}, if there exists some randomness $q$ such that $(s_1, \ldots, s_n) = \Pi_{\Gen}(s, q)$ and each {\it honest} party $P_i \in \Partyset$ has the share $s_i$.
\end{definition}
 In the literature, two types of VSS schemes have been considered. 
 The type-I VSS schemes are ``weaker" compared 
 to the type-II VSS schemes.
  Namely, in type-II VSS, it is {\it guaranteed} that the dealer's secret is secret-shared
 as per the semantics of some {\it specified} secret-sharing scheme\footnote{In Type-I VSS, the underlying secret {\it need not} be secret-shared
                         as per the semantics of any $t\mbox{-out-of-}n$ SS scheme.} (for instance, say Shamir's SS \cite{Sha79}).
  While type-I VSS is sufficient to study VSS as a stand-alone primitive (for instance, to study the round complexity of VSS  \cite{FGGRS06} or to
  design a BA protocol \cite{CanettiThesis}), 
 type-II VSS schemes are desirable when VSS is used as a primitive in secure MPC protocols \cite{KKK09,AL17,AKP20}.
   \begin{definition}[\bf Type-I VSS \cite{GIKR01}]
   \label{def:VSSWeak}
   Let $(\Sh, \Rec)$ be a pair of protocols for the parties in $\Partyset$, where a designated {\it dealer}
    $\D \in \Partyset$ has some private input $s \in \SecretSpace$ for the protocol $\Sh$.
   Then $(\Sh, \Rec)$ is called a {\it Type-I perfectly-secure VSS scheme}, if the following requirements hold.
     \begin{myitemize}
        \item[--] {\bf Privacy}: If $\D$ is {\it honest}, then the view of $\Adv$ during $\Sh$ is distributed independent of $s$.
         \item[--] {\bf Correctness}: If $\D$ is {\it honest}, then all honest parties output $s$ at the end of $\Rec$.
         \item[--] {\bf Strong Commitment}: Even if $\D$ is {\it corrupt}, in any execution of $\Sh$, the joint view of the honest parties defines a unique value
                    $s^{\star} \in \SecretSpace$ (which could be different from $s$), such that all honest parties
                 output $s^{\star}$ at the end of $\Rec$, irrespective of the behaviour of $\Adv$.
      \end{myitemize}
   \end{definition}
    \begin{definition}[\bf Type-II VSS \cite{Gol04}]
   \label{def:VSSStrong}
   Let $\Pi = (\Pi_{\Gen}, \Pi_{\Recover})$ be a $t\mbox{-out-of-}n$ SS scheme. 
   Then $(\Sh, \Rec)$ is called a {\it Type-II perfectly-secure VSS scheme with respect to $\Pi$}, if the following requirements hold.
     \begin{myitemize}
        \item[--] {\bf Privacy}: If $\D$ is {\it honest}, then the view of $\Adv$ during $\Sh$ is distributed independent of $s$.
         \item[--] {\bf Correctness}: If $\D$ is {\it honest}, then at the end of $\Sh$, the value $s$ is secret-shared among $\Partyset$ as per
                         $\Pi$ (see Definition \ref{def:SSShared}). Moreover, all honest parties output $s$ at the end of $\Rec$.
         \item[--] {\bf Strong Commitment}: Even if $\D$ is {\it corrupt}, in any execution of $\Sh$ the joint view of the honest parties defines some value
                    $s^{\star} \in \SecretSpace$, such that $s^{\star}$ is secret-shared among $\Partyset$ as per
                         $\Pi$ (see Definition \ref{def:SSShared}). Moreover, all honest parties output $s^{\star}$ at the end of $\Rec$.
      \end{myitemize}
   \end{definition}   
\paragraph{\bf Alternative Definition of VSS}
Definition \ref{def:VSSWeak}-\ref{def:VSSStrong} are called ``property-based" definition, where
  we enumerate a list of desired security goals.
  One can instead follow other definitional frameworks such as the the {\it ideal-world/real-world}
    paradigm of Canetti \cite{Can20} or the {\it constructive-cryptography} paradigm of Liu-Zhang and Maurer \cite{LM20}.
   Proving the security of VSS schemes as per these paradigm brings in additional technicalities in the proofs.
   Since our main goal  is to survey the existing VSS protocols, we will stick to the property-based definitions, which are easy to follow. 
  \subsection{Properties of Polynomials Over a Finite Field}
  Let $\F$ be a finite field where $|\F| > n$ with $\alpha_1, \ldots, \alpha_n$ be distinct non-zero elements of $\F$.
  A degree-$d$ {\it univariate polynomial} over $\F$ is of the form
 $f(x) = a_0  + \ldots + a_d x^d$, where $a_i \in \F$. A degree-$(\ell, m)$ {\it bivariate polynomial} over $\F$
   is of the form
    $F(x, y) = \displaystyle \sum_{i, j = 0}^{i = \ell, j = m}r_{ij}x^i y^j$, where $r_{ij} \in \F$.
    The polynomials $f_i(x) \defined F(x, \alpha_i)$ and $g_i(y) \defined F(\alpha_i, y)$
     are called the $i^{th}$ {\it row}
   and {\it column-polynomial} respectively of  $F(x, y)$
    as evaluating $f_i(x)$
   and $g_i(y)$ at $x = \alpha_1, \ldots, \alpha_n$
   and at $y = \alpha_1, \ldots, \alpha_n$ respectively results in an $n \times n$ matrix
   of points on $F(x, y)$ (see Fig \ref{fig:Bivariate}).
   Note  that $f_i(\alpha_j) = g_j(\alpha_i) = F(\alpha_j, \alpha_i)$ holds for all $\alpha_i, \alpha_j$. 
       We say a degree-$m$ polynomial $F_i(x)$ (resp.~a degree-$\ell$ polynomial $G_i(y)$), where $i \in \{1, \ldots, n \}$, {\it lies} on 
   a degree-$(\ell, m)$ bivariate polynomial $F(x, y)$, if $F(x, \alpha_i) = F_i(x)$ (resp.~$F(\alpha_i, y) = G_i(y)$) holds.
     $F(x, y)$ is called {\it symmetric}, if
    $r_{ij} = r_{ji}$ holds, implying $F(\alpha_j, \alpha_i) = F(\alpha_i, \alpha_j)$ 
    and $F(x, \alpha_i) = F(\alpha_i, x)$.
\begin{definition}[{\bf $d$-sharing} \cite{DN07,BH08}]
\label{def:sharing}
  A value $s \in \F$ is said to be $d$-shared, if there exists a 
   degree-$d$  polynomial, say $q(\cdot)$, with $q(0) = s$, such that each (honest) $P_i \in \Partyset$ holds
  its {\it share} $s_i \defined q(\alpha_i)$ (we interchangeably use the term {\it shares of $s$} and {\it shares of the polynomial $q(\cdot)$} to denote
   the values $q(\alpha_i)$).
  The vector of shares of $s$ corresponding to the (honest) parties in $\Partyset$ is denoted as $[s]_d$.
   A set of values $S = (s^{(1)}, \ldots, s^{(L)}) \in \F^{L}$ is said to be $d$-shared,
    if each $s^{(i)} \in S$ is $d$-shared.
\end{definition}
 \subsubsection{\bf Properties of Univariate Polynomials Over $\F$}
   Most of the type-II VSS schemes 
   are with respect to the Shamir's $t\mbox{-out-of-}n$ SS scheme $(\Shamir_{\Gen}, \Shamir_{\Recover})$ \cite{Sha79}.
    Algorithm $\Shamir_{\Gen}$ takes input a secret $s \in \F$. To compute the shares, it picks
   a {\it Shamir-sharing} polynomial $q(\cdot)$ randomly from the set $\Partyset^{s, t}$ of all degree-$t$ univariate polynomials
   over $\F$ whose constant term is $s$.
   The output of $\Shamir_{\Gen}$ is $(s_1, \ldots, s_n)$, where $s_i = q(\alpha_i)$. 
   Since $q(\cdot)$
    is chosen randomly, the probability distribution of the $t$ shares learnt by $\Adv$ will be independent of the underlying secret.
   Formally:
    \begin{lemma}[{\bf \cite{AL17}}]
   \label{lemma:Shamir}
   For any set of distinct non-zero elements  $\alpha_1, \ldots, \alpha_n \in \F$, any pair of values $s, s' \in \F$, any subset $I \subset \{1, \ldots, n \}$
   where $|I| = \ell \leq t$, and every $\vec{y} \in \F^{\ell}$, it holds that:
   \[\underset {f(x) \in_r \Partyset^{s, t}}{\mbox{Pr}} \Big [ \vec{y} = (\{f(\alpha_i) \}_{i \in I}   )     \Big ] =
     \underset {g(x) \in_r \Partyset^{s', t}}{\mbox{Pr}} \Big [ \vec{y} = (\{g(\alpha_i) \}_{I \in I}   )     \Big ],\] 
   where $f(x)$ and $g(x)$ are chosen randomly (denoted by the notation $\in_r$)
    from  $\Partyset^{s, t}$ and $\Partyset^{s', t}$, respectively.
   \end{lemma}
    Let $(s_1, \ldots, s_n)$ be a vector of Shamir-shares for $s$, generated by $\Shamir_{\Gen}$. Moreover, let $I \subset \{1, \ldots, n \}$, where
    $|I| = t + 1$. Then $\Shamir_{\Recover}$ takes input the shares $\{s_i \}_{i \in I}$ and outputs $s$ by interpolating the unique degree-$t$ Shamir-sharing polynomial passing
    through the points $\{(\alpha_i, s_i) \}_{i \in I}$. 
    \paragraph{\bf Relationship Between $d$-sharing and Reed-Solomon (RS) Codes}
    Let $s$ be $d$-shared through a polynomial $q(\cdot)$ and let $(s_1, \ldots, s_n)$ be the vector of shares. 
    Moreover, let $W$ be a subset of these shares, such that it is ensured that at most $r$ shares in $W$ are incorrect (the exact identity of the incorrect shares are
    not known). The goal is to error-correct the incorrect shares in $W$ and correctly reconstruct back the polynomial $q(\cdot)$. 
    Coding-theory \cite{MS81} says that this is possible if and only if $|W| \geq d + 2r + 1$ and the corresponding algorithm
    is denoted by $\RSDec(d, r, W)$. There are several well-known efficient instantiations of $\RSDec$, such as  
    the  Berlekamp-Welch algorithm  \cite{CodingTheory}.
\subsubsection{\bf Properties of Bivariate Polynomials Over $\F$}
 There always exists a unique degree-$d$ univariate polynomial, passing through
   $d+1$ distinct points. A generalization of this result for bivariate polynomials is that
    if there are ``sufficiently many" univariate polynomials which are ``pair-wise consistent",
   then together they lie on a unique bivariate polynomial. Formally:
   \begin{lemma}[\bf Pair-wise Consistency Lemma \cite{CanettiThesis,PCR15,AL17}]
  \label{lemma:bivariate}
Let $\{f_{i_1}(x), \ldots, f_{i_{q}}(x)\}$ and $\{g_{j_1}(y), \ldots, \allowbreak g_{j_r}(y) \}$
  be degree-$\ell$ and degree-$m$ polynomials respectively where
 $q \geq m + 1, r \geq \ell + 1$ and where 
 $i_1, \ldots, i_{q}, j_1, \ldots, j_r \in \{1, \ldots, n \}$.
 Moreover, let for every $i \in  \{i_1, \ldots, i_{q} \}$ and every $j \in \{j_1, \ldots, j_r \}$, 
 the condition
  $f_i(\alpha_j) = g_j(\alpha_i)$ holds.  Then there exists a unique
  degree-$(\ell, m)$ bivariate polynomial $\starF(x, y)$, such that
  the polynomials $f_{i_1}(x), \ldots, f_{i_q}(x)$ and 
  $g_{j_1}(y), \ldots, g_{j_r}(y)$ lie on $\starF(x, y)$.
\end{lemma}
In type-II VSS schemes based on Shamir's SS scheme, $\D$ on having input $s$ 
  first picks a random
  degree-$t$ Shamir-sharing polynomial $q(\cdot) \in \Partyset^{s, t}$ and then embeds $q(\cdot)$
 into a random degree-$(t, t)$ bivariate polynomial $F(x, y)$ at $x = 0$. 
  Each $P_i$ then receives  $f_i(x) = F(x, \alpha_i)$ and $g_i(y) = F(\alpha_i, y)$ from $\D$. Similar to 
 Shamir SS, $\Adv$ by learning at most $t$ row and column-polynomials, does not learn 
 $s$. Intuitively, this is because $(t + 1)^2$ distinct values  are required to uniquely determine
  $F(x, y)$, but $\Adv$ learns at most $t^2 + 2t$ distinct values.
  In fact, it can be shown that for every two degree-$t$ polynomials $q_1(\cdot), q_2(\cdot)$ 
 such that $q_1(\alpha_i) = q_2(\alpha_i) = f_i(0)$ holds for every $P_i \in \Bad$ (where $\Bad$ is the set of corrupt parties), the distribution of the polynomials
 $\{f_i(x), g_i(y) \}_{P_i \in \Bad}$ when $F(x, y)$ is chosen based on $q_1(\cdot)$, 
 is identical to the distribution when $F(x, y)$ is chosen based on $q_2(\cdot)$. Formally:
    \begin{lemma}[\cite{AL17}]
  \label{lemma:bivariateprivacy}
  Let $\Bad \subset \Partyset$ where $|\Bad| \leq t$, and let $q_1(\cdot) \neq q_2(\cdot)$ be degree-$t$ polynomials where $q_1(\alpha_i) = q_2(\alpha_i)$ holds
     for all $P_i \in \Bad$. Then the probability distributions 
     $\Big \{ \{F(x, \alpha_i), F(\alpha_i, y) \}_{P_i \in \Bad} \Big \}$  and  $\Big \{ \{F'(x, \alpha_i), F'(\alpha_i, y) \}_{P_i \in \Bad} \Big \}$ are 
     identical,  where $F(x, y) \neq F'(x, y)$ are different degree-$(t, t)$ bivariate polynomials, chosen at random, under the constraints that
  $F(0, y) = q_1(\cdot)$ and $F'(0, y) = q_2(\cdot)$ holds.  
  \end{lemma}


%% file: SThresholdBounds.tex
\section{Lower Bounds}
\label{sec:SThresholdBounds}
In {\it any} perfectly-secure VSS scheme, 
 the joint view of the {\it honest} parties should uniquely determine the dealer's secret. Otherwise the {\it correctness} property 
 will be violated if the corrupt parties produce incorrect view during the
  reconstruction phase. Since there can be only $n - t$ {\it honest} parties, to satisfy the {\it privacy} property, the condition
  $n - t > t$ should necessarily hold, as otherwise the view of the adversary will not be independent of the dealer's secret. 
    We actually need a stricter necessary condition of
  $n > 3t$ to hold for {\it any} perfectly-secure VSS scheme, as stated in the following theorem.
  \begin{theorem}[\cite{DDWY93}]
  \label{SThresholdVSS:Necessaity}
  Let $\Pi = (\Sh, \Rec)$ be a perfectly-secure VSS scheme. Then $n > 3t$ holds.
  \end{theorem}
   Theorem \ref{SThresholdVSS:Necessaity} is first proved formally by Dolev, Dwork, Waarts and Yung in \cite{DDWY93} by relating VSS with the problem of
  {\it $1$-way perfectly-secure message transmission} (1-way PSMT) \cite{DDWY93}.  
  In the 1-way PSMT problem, there is a sender {\bf S} and a receiver
   {\bf R}, such that there are $n$ {\it disjoint uni-directional}
  communication channels $Ch_1, \ldots, Ch_n$
   from {\bf S} to {\bf R} (i.e. only {\bf S} can send messages to {\bf R} along these channels).
  At most $t$ out of these channels can be controlled by a {\it computationally unbounded} malicious/Byzantine adversary in any arbitrary fashion.
    The goal is to design a protocol, which allows {\bf S} to send some input message $m$ {\it reliably} (i.e. {\bf R} should be able to receive $m$ without any error)
    and {\it privately} (i.e. view of the adversary should be independent of $m$) to {\bf R}.
    In \cite{DDWY93}, it is shown that a 1-way PSMT protocol exists only if $n > 3t$. Moreover,
     if there exists a perfectly-secure VSS scheme $(\Sh, \Rec)$ with $n \leq 3t$, then one can design a $1$-way PSMT protocol with $n \leq 3t$, which is a contradiction.
     On a very high level, the reduction from $1$-way PSMT to VSS can be shown as follows: 
     {\bf S} on having a message $m$, acts as a dealer and runs an instance of $\Sh$ with input $m$ by playing the role of the parties $P_1, \ldots, P_n$ as per the protocol
     $\Sh$. Let $\view_i$ be the view generated for $P_i$, which {\bf S} communicates to {\bf R} over $Ch_i$.
     Let {\bf R} receives $\view'_i$ over $Ch_i$, where $\view'_i = \view_i$ if $Ch_i$ is {\it not} under adversary's control. 
     To recover $m$, {\bf R} applies the reconstruction function as per $\Rec$ on $(\view'_1, \ldots, \view'_n)$. 
     It is easy to see that the {\it correctness} of the VSS scheme implies that {\bf R} correctly recovers $m$, while the {\it privacy}
     of the VSS scheme guarantees that the view of any adversary controlling at most $t$ channels remains independent of $m$.   
     
 An alternative argument for the requirement of  $n > 3t$ for perfect VSS follows from its reduction to a perfectly-secure {\it reliable-broadcast} (RB) protocol  over the pair-wise channels, for which   $n > 3t$ is required \cite{PSL80}. This is elaborated further  later in the context of usage of a broadcast channel for designing VSS protocols (the paragraph entitled "On the Usage of Broadcast Channel" after Lemma~\ref{lemma:SMImpossibility} and Footnote \ref{fn:VSStoBA}).
    \paragraph{\bf The Round Complexity of VSS}
    In Genarro, Ishai, Kushilevitz and Rabin \cite{GIKR01}, 
    the round-complexity of a VSS scheme is defined to be the number of rounds in the sharing phase, as all (perfectly-secure) VSS schemes adhere to  a single-round reconstruction.
      The interplay between the round-complexity of perfectly-secure VSS and resilience bounds is stated below. 
  \begin{theorem}[\cite{GIKR01}]
  \label{SThresholdVSS:RoundComplexity}
  Let $R \geq 1$ be a positive integer and let $R' \leq R$. Then:
    \begin{myitemize}
      \item If $R = 1$, then there exists no perfectly-secure VSS scheme with $(R, R')$ rounds in the sharing phase 
      if either $t > 1$ (irrespective of the value of $n$) or if ($t = 1$ and $n \leq 4$).      
         \item If $R = 2$, then perfectly-secure VSS with $(R, R')$ rounds in  sharing phase  is possible only if $n > 4t$.
         \item If $R \geq 3$, then perfectly-secure VSS with $(R, R')$ rounds in sharing phase is possible only if $n > 3t$.
    \end{myitemize}
  \end{theorem}
     We give a very high level overview of the proof of Theorem \ref{SThresholdVSS:RoundComplexity}. We focus only on  $2$ and $R$-round sharing phase VSS schemes where
     $R \geq 3$, as
  one can use standard hybrid arguments to derive the bounds related to 
   VSS schemes with $1$-round sharing phase 
      (see for instance \cite{PCRR09}).  
   The lower bound for $2$-round VSS schemes (namely $n > 4t$) is derived by relating VSS to the 
  {\it secure multi-cast} (SM) problem. In the SM problem, there exists a designated {\it sender} $\Sender \in \Partyset$ with some private message $m$
  and a designated {\it receiving set} $\R \subseteq \Partyset$, where $\Sender \in \R$ and where $|\R| > 2$.
   The goal is to design a protocol which allows $\Sender$ to send its message identically to all 
   the parties {\it only} in $\R$,
    even in the presence of an adversary who can control any $t$ parties, possibly including $\Sender$.
    Moreover, if all the parties in $\R$ are {\it honest}, then the view of the adversary should be independent of $m$.
    Genarro et al.~\cite{GIKR01} establishes the following relationship between VSS and SM. 
    \begin{lemma}[\cite{GIKR01}]
    \label{lemma:VSStoSM}
    Let $(\Sh, \Rec)$ be a VSS scheme with a $k$-round sharing phase where $k \geq 2$. Then there exists a $k$-round SM protocol.
    \end{lemma}   
    The proof of Lemma \ref{lemma:VSStoSM} proceeds in two steps. 
    \begin{myitemize}
    \item[--] It is first shown that for   any $R$-round protocol $\Pi$, there 
    exists an $R$-round protocol $\Pi'$ with the same security guarantees as $\Pi$, 
    such that all the messages in rounds $2, \ldots, R$ of $\Pi'$ are broadcast messages. 
    The idea is to let each
    $P_i$ exchange a ``sufficiently-large" random pad with every $P_j$ apriori during the first round. Then in any subsequent round of $\Pi$, if
    $P_i$ is supposed to {\it privately} send $v_{ij}$ to $P_j$, then in $\Pi'$, party $P_i$ instead broadcasts $v_{ij}$ being masked with appropriate pad, exchanged
    with $P_j$. Since $P_j$ is supposed to hold the same pad, it can unmask the broadcasted value and recover $v_{ij}$ and process it as per $\Pi$.
    \item[--] Let $(\Sh, \Rec)$ be a VSS scheme where $\Sh$ requires $k$ rounds such that $k \geq 2$. Using the previous implication, we can assume that all the messages during the
     {\it final}
    round of $\Sh$ are broadcast messages. Using $(\Sh, \Rec)$, one can get a $k$-round SM protocol as follows: the parties invoke an instance of $\Sh$, with 
    $\Sender$ playing the role of $\D$ with input $m$. In the final round, apart from the messages broadcasted by the parties as part of $\Sh$, every party
    $P_i$ {\it privately} sends its entire view of the first $k - 1$ rounds during $\Sh$ to every party in $\R$. Based on this, every party in
    $\R$ will get the view of all the parties in $\Partyset$ for all the $k$ rounds of $\Sh$. 
    Intuitively, this also gives every (honest) party in $\R$ the share of every (honest) 
    party from $\Sh$.
    Every party in $\R$ then applies the reconstruction function on the $n$ extrapolated views as per $\Rec$ and outputs the result. 
    It is easy to see that the {\it correctness} of the VSS implies that if $\Sender$ is {\it honest}, then all the honest parties in $\R$ obtains $m$.
    Moreover, the {\it privacy} of the VSS scheme implies that if $\R$ contains only honest parties, then the view of the adversary remains independent of $m$.
    Finally, the {\it strong commitment} of the VSS guarantees that even if $\Sender$ is {\it corrupt}, all {\it honest} parties in $\R$ obtain the same output.              
    \end{myitemize}
    Next, Genarro et al.~\cite{GIKR01} shows the impossibility of any $2$-round SM protocol with $n \leq 4t$.
    \begin{lemma}
    \label{lemma:SMImpossibility}
    There exists no $2$-round SM protocol where $n \leq 4t$.
    \end{lemma}   
    By a standard {\it player-partitioning} argument \cite{Lyn96}, Lemma \ref{lemma:SMImpossibility} reduces to showing the impossibility of any $2$-round SM protocol with
    $n = 4$ and $t = 1$. At a high-level, the player-partitioning argument goes as follows: if there exists a $2$-round SM protocol $\Pi$ with $n = 4t$, then one can also design a $2$-round SM protocol $\Pi'$ for
    $4$ parties, where each of the $4$ parties in $\Pi'$ plays the role of a disjoint set of $t$ parties as per $\Pi$. However, one can show that there does not exist any $2$-round SM protocol with $n = 4$ and $t = 1$, thus ruling out the existence of any $2$-round SM protocol with $n \leq 4t$. We refer the interested readers to \cite{GIKR01} for the proof of {\it non-existence} of any $2$-round SM protocol with
    $n = 4$ and $t = 1$.
    Now combining Lemma \ref{lemma:VSStoSM} with
     Lemma \ref{lemma:SMImpossibility}, we get that there exists no VSS scheme with $n \leq 4t$ and  a $2$-round sharing phase. Since $n > 3t$ is necessary for {\it any} VSS scheme, 
     this automatically implies that a VSS scheme with $3$ or more rounds of sharing phase will necessarily require $n > 3t$.
 \paragraph{\bf On the Usage of Broadcast Channel}
 All {\it synchronous} VSS schemes assume the existence of  a broadcast channel. However, this is just a
   {\it simplifying 
  abstraction}, as the parties can ``emulate" the effect of a broadcast channel by executing a 
  perfectly-secure {\it reliable-broadcast} (RB) protocol 
   over the pair-wise channels, provided $n > 3t$ holds \cite {PSL80}.
    RB protocols with {\it guaranteed termination} in the presence of malicious/Byzantine adversaries
   require $\Omega(t)$ rounds of communication \cite{FL82}, while the RB protocols with {\it probabilistic termination} guarantees require 
   $\Order(1)$ expected number of 
  rounds \cite{FM97,KK06} where the constants are high. 
    Given the fact that the usage of broadcast channel is an ``expensive resource", a
  natural question is whether one can design a VSS scheme with a {\it constant} number of rounds in the sharing phase and which {\it does not} require the usage
  of broadcast channel in any of these rounds. Unfortunately, the answer is no. This is because such a VSS scheme will imply the existence of a 
   strict constant round
   RB protocol with {\it guaranteed termination} in the presence of a {\it malicious} adversary
   (the message to be broadcast by the sender can be shared using the VSS scheme with sender playing the role of the dealer,
    followed by reconstructing the shared message),
   which is {\it impossible} as per the result of\footnote{\label{fn:VSStoBA}This reduction from RB to VSS is another way to argue the necessity of the $n > 3t$ condition for VSS, since the necessary condition
   for RB is $n > 3t$ \cite{PSL80}.} \cite{FL82}. 
   Hence the best that one can hope for is to design VSS schemes which invoke the broadcast channel only in a fewer rounds.

%% file: SThreshold.tex
\section{Upper Bounds}
\label{sec:SThreshold}
We now discuss the optimality of the bounds in Theorem \ref{SThresholdVSS:RoundComplexity} by presenting
 VSS schemes with various round-complexities. The sharing phase of these schemes are summarized in Table \ref{tab:VSSSummary} and their  
 reconstruction phase require one round. The prefix in the
  names of the schemes denotes the number of rounds 
  in the sharing phase.  In the table, $\Gr$ denotes a finite group and RSS stands for replicated secret-sharing \cite{ISN87} (see Section \ref{sec:GIKRII}).  
  The $\AKP$ scheme has some special properties, compared to $\FGGRS, \KKK$ schemes, which are useful for designing round-optimal 
  perfectly-secure MPC protocols (see Section \ref{sec:AKP}).

   \begin{table}[H]
    \centering
    \resizebox{\textwidth}{!}{  
    \begin{tabular}{|c|c|c|c|c|c|c|}
        \hline
 Scheme & $n$ & Round Complexity & Type & Sharing Semantic & Algebraic Structure & Communication Complexity
  \\ \hline
  $\BGW$ \cite{BGW88} & $n > 3t$ & $(7, 5)$ & Type-II & Shamir &  $\F$ & $\Order(n^2 \log{|\F|} + \BC(n^2 \log{|\F|}))$ \\ \hline
  $\ModBGW$ \cite{GIKR01} & $n > 3t$ & $(5, 3)$ & Type-II & Shamir & $\F$ & $\Order(n^2 \log{|\F|} + \BC(n^2 \log{|\F|}))$ \\ \hline
   $\GIKRI$ \cite{GIKR01} & $n > 3t$ & $(4, 3)$ & Type-II & Shamir & $\F$ &  $\Order(n^2 \log{|\F|} + \BC(n^2 \log{|\F|}))$ \\ \hline
  $\GIKRII$  \cite{GIKR01} & $n > 3t$ & $(3, 2)$ & Type-II & RSS & $\Gr$ & $\Order(n \cdot {n \choose t} \log{|\Gr|} + \BC(n \cdot {n \choose t} \log{|\Gr|}))$ \\ \hline
$\FGGRS$  \cite{FGGRS06} & $n > 3t$ & $(3, 2)$ & Type-I & Not Applicable & $\F$ & $\Order(n^3 \log{|\F|} + \BC(n^3 \log{|\F|}))$ \\ \hline
  $\KKK$  \cite{KKK09} & $n > 3t$ & $(3, 1)$ & Type-II & Shamir & $\F$ & $\Order(n^3 \log{|\F|} + \BC(n^3 \log{|\F|}))$ \\ \hline
$\AKP$  \cite{AKP20} & $n > 3t$ & $(3, 2)$ & Type-II & Shamir & $\F$ & $\Order(n^3 \log{|\F|} + \BC(n^3 \log{|\F|}))$ \\ \hline
  $\GIKRIII$  \cite{GIKR01} & $n > 4t$ & $(2, 1)$ & Type-II & Shamir & $\F$ & $\Order(n^2 \log{|\F|} + \BC(n^2 \log{|\F|}))$\\ \hline
   $\OneVSS$  \cite{GIKR01} & $n= 5, t = 1$ & $(1, 0)$ & Type-I & Not Applicable  & $\F$ & $\Order(n \log{|\F|})$ \\ \hline
    \end{tabular}
    }
      \caption{\label{tab:VSSSummary}Summary of the sharing phase of the perfectly-secure VSS schemes, with $\BC$ denoting communication over the broadcast channel.}
\end{table}
\vspace*{-0.6cm}
 \subsection{VSS Schemes with $n > 3t$}
 \input{SThreshold3t}

  \subsection{VSS Scheme with $n > 4t$}
 \input{SThreshold4t}

  \subsection{VSS Scheme with a Single Round}
 \input{SThreshold1}

%% file: SThreshold3t.tex
We start with perfectly-secure VSS schemes with $n > 3t$. 
    While presenting these schemes, we use the following simplifying 
    conventions. If in a protocol a party is expecting some message from a sender party and if it either receives no message
     or
    semantically/syntactically incorrect message, then the receiving party substitutes some default value 
    and proceeds with the steps of the protocol. Similarly, if the dealer is {\it publicly} identified to be cheating then the parties discard the dealer
    and terminate the protocol execution with a default sharing of $0$.
    \subsubsection{\bf The $\BGW$ Scheme}
    The scheme (Fig \ref{fig:BGW}) consists of the protocols $\BGW\mbox{-}\Sh$ and $\BGW\mbox{-}\Rec$.
    Protocol $\BGW\mbox{-}\Sh$ is actually a simplified version of the original protocol, taken from \cite{Hir01}.
    To share $s$, the dealer $\D$ picks a random degree-$t$ {\it Shamir-sharing polynomial} $q(\cdot)$ and the goal is to ensure that
    $\D$ {\it verifiably} distributes the Shamir-shares of $s$ as per $q(\cdot)$. Later, during $\BGW\mbox{-}\Rec$,
    the parties exchange these Shamir-shares and reconstruct $q(\cdot)$ by error-correcting up to $t$ incorrect shares using the algorithm $\RSDec$.
    The verifiability in $\BGW\mbox{-}\Sh$ ensures that even if $\D$ is {\it corrupt}, the shares distributed by $\D$ to the (honest) parties  are as per some
    degree-$t$ Shamir-sharing polynomial, say $q^{\star}(\cdot)$, thus ensuring that $s^{\star} \defined q^{\star}(0)$ is $t$-shared. Moreover, the same $s^{\star}$ is
    reconstructed during $\BGW\mbox{-}\Rec$. This ensures that the scheme is of type-II.
    
    To prove that $\D$ is sharing its secret using a degree-$t$ polynomial $q(\cdot)$, the dealer $\D$ embeds $q(\cdot)$
    in a random degree-$(t, t)$ bivariate polynomial $F(x, y)$. As shown in Fig \ref{fig:Bivariate}, there are two approaches to do this embedding.
    The polynomial $q(\cdot)$ could be either embedded at $x = 0$ (the approach shown in part (a)) or it could be embedded at $y = 0$ 
    (the approach shown in part (b)). For our description, we follow the first approach. The dealer then 
    distributes distinct row and column-polynomials to respective parties. If $\D$ is {\it honest}, then
    this distribution of information maintains the privacy of dealer's secret (follows from Lemma \ref{lemma:bivariateprivacy}). 
    Also, if $\D$ is {\it honest}, then constant term of the individual row-polynomials 
    actually constitute the Shamir-shares of $s$, as they constitute distinct points on $q(\cdot)$.
     However, a potentially {\it corrupt}
    $\D$ may distribute polynomials, which {\it may not} be derived from a single degree-$(t, t)$ bivariate polynomial. Hence, the parties interact to
    verify if $\D$ has distributed ``consistent" row and column-polynomials, {\it without} revealing any additional information about $s$.
    \input{Figures/Bivariate.tex}
    % ...
    
    Every pair of parties $P_i, P_j$ upon receiving the polynomials $f_i(x), g_i(y)$ and $f_j(x), g_j(y)$ respectively, interact and check if 
     $f_i(\alpha_j) = g_j(\alpha_i)$ and $f_j(\alpha_i) = g_i(\alpha_j)$ holds.  
     If the checks pass for all the pairs of (honest) parties, then 
    from Lemma \ref{lemma:bivariate}, it follows that $\D$ has distributed consistent row (and column-polynomials) to the parties. However, if the checks do not pass,
     then either $\D$ has distributed inconsistent polynomials or the parties have not exchanged the correct common values. In this case,
    the parties interact publicly with $\D$ to resolve these inconsistencies. The details follow.
    
    Every $P_i$ upon receiving the supposedly common values on its column polynomial prepares a {\it complaint-list}
    $L_i$, which includes all the parties $P_j$ whose received value is inconsistent with $P_i$'s column-polynomial (this is interpreted as if there is a dispute
    between $P_i$ and $P_j$). If there is a dispute between $P_i$ and $P_j$, then at least one of the three parties $\D, P_i, P_j$ is {\it corrupt}.
    Each party then broadcasts its complaint-list. 
     In response, for every dispute reported by a party, the dealer 
        $\D$ makes public its version of the disputed value, namely the corresponding value on its bivariate polynomial. 
        This is followed by the first-stage accusations against the dealer. Namely, a party publicly ``accuses" $\D$, if  the party is in dispute with more than $t$ parties or if it finds
        $\D$ making public a value, which is not consistent with its column-polynomial. In response, $\D$ makes public the row and column-polynomials of such accusing parties. However, care has to taken to ensure that these broadcasted polynomials are consistent with the polynomials of the parties who have not yet accused $\D$.
        This is done through the second-stage of public accusations against the dealer, where a party (who has not yet accused $\D$) publicly accuses $\D$, if it finds any inconsistency
        between the row and column-polynomials held by the party and the polynomials which were made public by $\D$.

        An {\it honest} $\D$ always respond correctly against any accusation or dispute. Moreover, there will be at most
        $t$ accusations against $\D$. Consequently, if the 
        parties find $\D$ not responding to any accusation/dispute or if more than $t$ parties accuse $\D$, then 
        $\D$ is {\it corrupt} and hence the parties discard $\D$.
        If $\D$ is {\it honest}, then all the values which are made public correspond to corrupt parties (already known to $\Adv$)
        and hence {\it does not} violate {\it privacy}.         
        On the other hand, if $\D$ is {\it corrupt} but not discarded, then it ensures that the polynomials of
        all honest parties are consistent.    
 \begin{schemesplitbox}{$\BGW$}{The perfectly-secure VSS scheme of Ben-Or, Goldwasser and Wigderson \cite{BGW88}.}{fig:BGW}
	\justify
\centerline{\algoHead{Sharing Phase: Protocol $\BGW\mbox{-}\Sh$}}
\begin{myitemize}
\item {\bf Round I (sending polynomials) --- the dealer does the following}:
  \begin{myitemize}
    \item[--] On having the input $s \in \F$, pick a random degree-$t$ {\it Shamir-sharing polynomial} $q(\cdot)$, such that $q(0) = s$ holds.
    Then pick a random degree-$(t, t)$ bivariate polynomial $F(x, y)$, such that $F(0, y) = q(\cdot)$ holds.
     \item[--] For $i = 1, \ldots, n$, send the polynomials $f_i(x) \defined F(x, \alpha_i)$ and $g_i(y) \defined F(\alpha_i, y)$ to party $P_i$. 
  \end{myitemize}
\item {\bf Round II (pair-wise consistency checks) --- each party $P_i$ does the following}: 
   \begin{myitemize}
      \item[--] On receiving degree-$t$ polynomials $f_i(x), g_i(y)$ from $\D$, send 
      $f_{ij} \defined f_i(\alpha_j)$ to $P_j$, for  $j = 1, \ldots, n$.
   \end{myitemize}
\item {\bf Round III (broadcast complaints) --- each party $P_i$ does the following}: 
    \begin{myitemize}
    \item[--] Initialize a {\it complaint-list} $L_i$ to $\emptyset$. For $j = 1, \ldots, n$, include $P_j$ to $L_i$, if $f_{ji} \neq g_i(\alpha_j)$.
     Broadcast $L_i$.
    \end{myitemize} 
\item {\bf Round IV (resolving complaints) --- the dealer does the following}:
     \begin{myitemize}
     \item[--] For $i = 1, \ldots, n$, if $P_i$ has broadcast $L_i \neq \emptyset$, then for every $P_j \in L_i$, broadcast the 
     value $F(\alpha_i, \alpha_j)$.     
     \end{myitemize}  
\item {\bf Round V (first-stage accusations) --- each party $P_i$ does the following}: 
      \begin{myitemize}
       \item[--] Broadcast the message $(P_i, \Accuse, \D)$, if any of the following conditions hold.
          \begin{mydescription}
          \item $|L_i| > t$ or if $P_i \in L_i$;
          \item If $\exists k \in \{1, \ldots, n \}$, such that $P_i \in L_k$ and $F(\alpha_k, \alpha_i) \neq f_i(\alpha_k)$;           
          \item If for any $P_j \in L_i$, the condition $F(\alpha_i, \alpha_j) \neq g_i(\alpha_j)$ holds.
          \end{mydescription}
        \end{myitemize}       
\item {\bf Round VI (resolving first-stage accusations) --- the dealer does the following}:
      \begin{myitemize}
      \item[--] For every $P_i$ who has broadcast $(P_i, \Accuse, \D)$, broadcast the degree-$t$ polynomials
             $f_i(x)$ and $g_i(y)$.
      \end{myitemize}       
\item {\bf Round VII (second-stage accusations) --- each party $P_i$ does the following}:  
     \begin{myitemize}
     \item[--]  Broadcast the message $(P_i, \Accuse, \D)$ if there exists any $j \in \{1, \ldots, n \}$ such that
     $\D$ has broadcast degree-$t$ polynomials $f_j(x), g_j(y)$ and either $f_j(\alpha_i) \neq g_i(\alpha_j)$ or $g_j(\alpha_i) \neq f_i(\alpha_j)$ holds.      
     \end{myitemize}
\item {\bf Output decision --- each party $P_i$ does the following}:
    \begin{myitemize}
    \item[--] If more than $t$ parties $P_j$  broadcast $(P_j, \Accuse, \D)$ throughout the protocol, then discard $\D$.
    \item[--] Else output the {\it share}\footnote{If $\D$ has broadcast new polynomials $f_i(x), g_i(y)$ for $P_i$ during Round VI, then consider
     these new polynomials.} $s_i = f_i(0)$. \\[.1cm]
    \end{myitemize}
\end{myitemize}
\justify
\centerline{\algoHead{Reconstruction Phase: Protocol $\BGW\mbox{-}\Rec$}}
Each party $P_i$ sends the share $s_i$ to every party $P_j \in \Partyset$.
  Let $W_i$ be the set of shares received by $P_i$ from the parties. Party $P_i$ executes
     $\RSDec(t, t, W_i)$ to reconstruct $s$.
\end{schemesplitbox}  
Since the idea of bivariate polynomials has been used in all the followup works on perfectly-secure VSS, we give a very
 high level overview of the proof of  the properties of $\BGW$ scheme.
 \begin{theorem}
 \label{thm:BGW}
 Protocols $(\BGW\mbox{-}\Sh, \BGW\mbox{-}\Rec)$ constitute a Type-II perfectly-secure VSS scheme with respect to
  Shamir's $t\mbox{-out-of-}n$ secret-sharing scheme. The protocol incurs a communication of
  $\Order(n^2 \log{|\F|})$ bits over the point-to-point channels and broadcast of $\Order(n^2 \log{|\F|})$ bits.
 \end{theorem}
 \begin{proof}
  If $\D$ is {\it honest}, then $f_i(\alpha_j) = g_j(\alpha_i)$ and $f_j(\alpha_i) = g_i(\alpha_j)$ holds
  for every pair of parties $(P_i, P_j)$. Consequently, no honest $P_j$ will be present in
  the list $L_i$ of any honest $P_i$. Moreover, $\D$ honestly resolves the first-stage accusations as well as second-stage accusations and consequently, no honest party
  accuses and discards $\D$. Hence $s$ will be $t$-shared through the degree-$t$
   polynomial $q(\cdot)$. Moreover, during $\BGW\mbox{-}\Rec$, the honest parties correctly reconstruct $q(\cdot)$ and hence $s$.
  This follows from the properties of $\RSDec$ and the fact that $q(\cdot)$ is a degree-$t$ polynomial
   and at most $t$ corrupt parties can send incorrect shares. Hence, the {\it correctness}
  property is guaranteed.
  
  Let $\Bad$ be the set of corrupt parties. If $\D$ is {\it honest}, then throughout $\BGW\mbox{-}\Sh$, the view of $\Adv$ consists of 
  $\{f_i(x), g_i(y) \}_{P_i \in \Bad}$. Moreover, no honest party accuses $\D$ and hence all the information which $\D$ makes public can be derived
  from $\{f_i(x), g_i(y) \}_{P_i \in \Bad}$.  Now since these polynomials are derived from $F(x, y)$, which is a randomly chosen polynomial
  embedding $q(\cdot)$, it follows from Lemma \ref{lemma:bivariateprivacy} that the view of
  $\Adv$ is distributed independently of $s$, thus guaranteeing {\it privacy}.
  
  For {\it strong commitment} we have to consider a {\it corrupt} $\D$. If the honest parties discard
  $\D$, then clearly the value $s^{\star} \defined 0$ will be $t$-shared 
   and the same value $0$ gets reconstructed during $\BGW\mbox{-}\Rec$. On the other hand, consider the case
  when the honest parties do not discard $\D$ during $\BGW\mbox{-}\Sh$. In this case we claim that the row and column-polynomials of all the {\it honest}
  parties are derived from a single degree-$(t, t)$ bivariate polynomial, say $\starF(x, y)$, which we call
  as $\D$'s committed bivariate polynomial. Consequently, $s^{\star} \defined \starF(0, 0)$ will be
  $t$-shared through the Shamir-sharing polynomial $q^{\star}(\cdot) \defined \starF(0, y)$ and 
  $s^{\star}$ gets reconstructed during $\BGW\mbox{-}\Rec$.
  
  To prove the claim, we first note that there are at least $n - t \geq 2t + 1$ parties who do not broadcast an $\Accuse$ message against
   $\D$ (as otherwise $\D$ is discarded).
  Let $\Honest$ be the set of {\it honest} parties among these $n - t$ parties. It is easy to see that $|\Honest| \geq n - 2t \geq t + 1$. The parties in $\Honest$
  receive degree-$t$ row and column-polynomials from $\D$. Moreover, for every $P_i ,P_j \in \Honest$, their row and column-polynomials are
  pair-wise consistent (as otherwise either $P_i$ or $P_j$ would broadcast an $\Accuse$ message against $\D$). It then follows from Lemma \ref{lemma:bivariate}
  that the row and column-polynomials of all the parties in $\Honest$ lie on a single degree-$(t, t)$ bivariate polynomial, say $\starF(x, y)$. 
  Next consider any {\it honest} party $P_j \not \in \Honest$, who broadcasts an  $\Accuse$ message against $\D$ and corresponding to which
  $\D$ makes public the row and column-polynomials of $P_j$. To complete the proof of the claim, we need to show that these polynomials also lie on 
   $\starF(x, y)$. We show it for the row-polynomial of $P_j$ and a similar argument can be used for the column-polynomial as well.
   So let $f_j(x)$ be the degree-$t$ row-polynomial broadcast by $\D$ for $P_j$. It follows that $f_j(\alpha_i) = g_i(\alpha_j)$ holds for {\it every}
   $P_i \in \Honest$, where $g_i(y)$ is the degree-$t$ column polynomial held by $P_i$ (otherwise $P_i$ would have broadcast
   an $\Accuse$ message against $\D$). Now $g_i(y) = \starF(\alpha_i, y)$ holds. Moreover, since $|\Honest| \geq t + 1$,
    the distinct points $\{(\alpha_j, g_i(\alpha_j)) \}_{P_i \in \Honest}$
   uniquely determine the degree-$t$ polynomial $\starF(x, \alpha_j)$. This implies that $f_j(x) = \starF(x, \alpha_j)$, as two different degree-$t$
   polynomials can have at most $t$ common points.
   
   In the protocol, $\D$ sends two degree-$t$ polynomials to each party and every pair of parties exchange $2$ common values,
      which requires a communication of $\Order(n^2)$ field elements. There could be at most $t$ parties corresponding to which 
      $\D$ makes public their polynomials and this requires a broadcast of $\Order(n^2)$ field elements.
 \end{proof}
 \subsubsection{\bf A $5$-round Version of $\BGW\mbox{-}\Sh$}
 Protocol $\BGW\mbox{-}\Sh$ follows the ``share-complaint-resolve" paradigm, where
  $\D$ first distributes the information on its bivariate polynomial, followed by parties complaining about any ``inconsistency", which is followed by $\D$
   resolving
  these inconsistencies.
   At the end, either all (honest) parties held consistent polynomials derived from a single degree-$(t, t)$ bivariate polynomial
    or $\D$ is discarded.
  The ``complaint" and ``resolve" phases of $\BGW\mbox{-}\Sh$ occupied {\it five} rounds.
  In Gennaro, Ishai, Kushilevitz and Rabin \cite{GIKR01}, the authors proposed a {\it round-reducing} technique, which collapses these phases to {\it three}
  rounds, thus reducing the overall number of rounds to five. The modified protocol $\ModBGW\mbox{-}\Sh$ is presented in Fig \ref{fig:ModBGW}.
  
  The high level idea of  $\ModBGW\mbox{-}\Sh$ is as follows. In $\BGW\mbox{-}\Sh$, during the third round, $P_i$ broadcasts {\it only} the identity of the parties
  with which it has a dispute (through $L_i$), followed by $\D$ making the corresponding disputed values public during Round IV, which is further
  followed by $P_i$ accusing $\D$ during Round V, if $P_i$ finds $\D$'s version  to mis-match with $P_i$'s version. 
  Let us call $P_i$ to be {\it unhappy}, if it accuses $\D$ during Round V. The round-reducing technique of \cite{GIKR01} enables to identify the set of
  unhappy parties $\Unhappy$ by the end of Round IV as follows. During Round III, apart from broadcasting the list of disputed parties,
  party $P_i$ also makes public its version of the corresponding disputed values.
   In response, both $\D$ and the corresponding complainee party
  makes public their respective version of the disputed value. Now based on whether $\D$'s version matches the complainant's version or complainee's version,
  the parties can identify the set $\Unhappy$. 
  
  In $\BGW\mbox{-}\Sh$, once $\Unhappy$ is identified, $\D$ makes public the polynomials of the parties in $\Unhappy$ during Round VI. And to verify if $\D$
  made public the correct polynomials, during Round VII,
  the parties {\it not in} $\Unhappy$ raise accusations against $\D$, if they find any inconsistency between the polynomials held by them
  and the polynomials made public by $\D$. 
    The round-reducing technique of \cite{GIKR01} collapses these two rounds into a single round.
  Namely, once $\Unhappy$ is decided, 
  $\D$ makes public the row-polynomials of these parties. In {\it parallel}, the parties not in $\Unhappy$ make public
  the corresponding supposedly common values on these row-polynomials. Now to check if $\D$ broadcasted correct row-polynomials, 
  one just has to verify whether each broadcasted row-polynomial is pair-wise consistent with at least $2t + 1$ corresponding values,
  broadcasted by the parties {\it not} in $\Unhappy$. 
 \begin{protocolsplitbox}{$\ModBGW\mbox{-}\Sh$}{A simplified $5$-round version of $\BGW\mbox{-}\Sh$ due to Genarro, Ishai, Kushilevitz and Rabin \cite{GIKR01}.}{fig:ModBGW}
	\justify
\begin{myitemize}
\item {\bf Round I} --- $\D$ picks $F(x, y)$ as in $\BGW\mbox{-}\Sh$ and distributes $f_i(x), g_i(y)$ lying on $F(x, y)$ to each $P_i$.
\item {\bf Round II} --- As in $\BGW\mbox{-}\Sh$, each $P_i$ upon receiving $f_i(x), g_i(y)$ from $\D$, sends $f_{ij} = f_i(\alpha_j)$ to $P_j$.
\item {\bf Round III --- each party $P_i$ does the following}:
   \begin{myitemize}
   \item[--] For every $P_j \in \Partyset$ where $f_{ji} \neq g_i(\alpha_j)$, broadcast
   $(\complaint, i, j, g_i(\alpha_j))$.
   \end{myitemize}
\item {\bf Round IV (making disputed values public) --- the dealer and each party $P_j$ does the following}: 
    \begin{myitemize}
     \item[--] If $P_i$ broadcasts $(\complaint, i, j, g_i(\alpha_j))$, then $\D$ and $P_j$  broadcasts
      $F(\alpha_i, \alpha_j)$ and $f_j(\alpha_i)$ respectively.
    \end{myitemize}
\item {\bf Local computation (deciding unhappy parties) --- each party $P_k$ does the following}:
      \begin{myitemize}
      \item[--] Initialize a set of {\it unhappy parties} $\Unhappy$ to $\emptyset$.
      \item[--] For every pair of parties $(P_i, P_j)$, such that
       $P_i$ has broadcast $(\complaint, i, j, g_i(\alpha_j))$, $\D$ has broadcast $F(\alpha_i, \alpha_j)$ and $P_j$ has broadcast $f_j(\alpha_i)$, do the following.
         \begin{myitemize}
          \item[--] If $g_i(\alpha_j) \neq F(\alpha_i, \alpha_j)$, then include $P_i$ to $\Unhappy$.
           \item[--]  If $f_j(\alpha_i) \neq F(\alpha_i, \alpha_j)$, then include $P_j$ to $\Unhappy$.
        \end{myitemize}
      \item[--] If $|\Unhappy| > t$, then discard $\D$.
      \end{myitemize}    
\item {\bf Round V (resolving unhappy parties) --- the dealer and each $P_j \not \in \Unhappy$ does the following}:
       \begin{myitemize}
       \item[--] For every $P_i \in \Unhappy$, the dealer $\D$ broadcasts degree-$t$ polynomial $f_i(x)$. 
       \item[--] For every $P_i \in \Unhappy$, party $P_j$ broadcasts $g_j(\alpha_i)$.       
       \end{myitemize}
\item {\bf Output decision --- each party $P_k$ does the following}:
    \begin{myitemize}
    \item[--] If  there exists any $P_i \in \Unhappy$ for which $\D$ broadcasts $f_i(x)$ and at most $2t$ parties 
    $P_j  \not \in \Unhappy$ broadcast $g_j(\alpha_i)$ values where $f_i(\alpha_j) = g_j(\alpha_i)$ holds, then 
       discard $\D$.
    \item[--] Else output the {\it share} $f_k(0)$.
    \end{myitemize}   
\end{myitemize}
\end{protocolsplitbox}

  \subsubsection{\bf The $4$-round $\GIKRI$ Scheme}
 Genarro et al \cite{GIKR01} proposed another round-reducing technique to reduce the number of rounds of $\ModBGW\mbox{-}\Sh$ by one.
  The modified protocol $\GIKRI\mbox{-}\Sh$ is presented in Fig \ref{fig:GIKRI}.
  The idea is to  ensure that the set $\Unhappy$ is decided by the end of Round III, even though this might look like
  an impossible task. This is because $\Unhappy$ can be decided during Round IV, only after 
   the results of pair-wise consistency checks are available during Round III. The key-observation of \cite{GIKR01} is 
   that the parties can ``initiate" the pair-wise consistency checks from Round I itself. More specifically, every $P_i, P_j$ exchange random pads {\it privately} during the
  Round I, independently of $\D$'s distribution of the polynomials.
  During the second round, $P_i, P_j$  can then broadcast a masked version of the supposedly common values
  on their polynomials, using the exchanged pads as the masks. If $\D, P_i$ and $P_j$ are {\it honest}, then the masked version of the
  common values will be the same and nothing about the common values will be learnt, as the corresponding
  masks will be private. By comparing the masked versions of the common values, the parties publicly learn about the results of pair-wise
  consistency checks by the end of the Round II. 
  \begin{protocolsplitbox}{$\GIKRI\mbox{-}\Sh$}{A $4$-round sharing phase protocol due to Genarro, Ishai, Kushilevitz and Rabin \cite{GIKR01}.}{fig:GIKRI}
	\justify
\begin{myitemize}
\item {\bf Round I (sending polynomials and exchanging random pads)}
   \begin{myitemize}
	  \item[--] $\D$ picks $F(x, y)$ as in $\BGW\mbox{-}\Sh$ and distributes $f_i(x), g_i(y)$ lying on $F(x, y)$ to each $P_i$.
	  \item[--] Each $P_i \in \Partyset$ picks a random pad $r_{ij} \in \F$ corresponding to every $P_j \in \Partyset$ and sends $r_{ij}$ to $P_j$.
    \end{myitemize}
\item {\bf Round II (broadcasting common values in a masked fashion) --- each $P_i$ does the following}:
       \begin{myitemize}
       \item[--] Broadcast $a_{ij} \defined f_i(\alpha_j) + r_{ij}$ and
        $b_{ij} \defined g_i(\alpha_j) + r'_{ji}$, where $r'_{ji}$ is the pad, received from $P_j$.
       \end{myitemize} 
\item {\bf Round III (making disputed values public) --- for all $P_i, P_j$ where $a_{ij} \neq b_{ji}$, party $P_i, P_j$ and $\D$ does the following}:
     \begin{myitemize}
       \item[--] $\D$ broadcasts $F(\alpha_j, \alpha_i)$; 
       $P_i$ broadcasts $f_i(\alpha_j)$ and 
       $P_j$ broadcasts $g_j(\alpha_i)$.
     \end{myitemize}       
\item {\bf Local computation (deciding unhappy parties) --- each party $P_k$ does the following}:
      \begin{myitemize}
      \item[--] Initialize a set of {\it unhappy parties} $\Unhappy$ to $\emptyset$.
       For every $(P_i, P_j)$ where $a_{ij} \neq b_{ji}$ and where 
              $P_i$ has broadcast $f_i(\alpha_j)$, $\D$ has broadcast $F(\alpha_j, \alpha_i)$ and $P_j$ has broadcast $g_j(\alpha_i)$, do the following.
         \begin{myitemize}
          \item[--] If $f_i(\alpha_j) \neq F(\alpha_j, \alpha_i)$, then include $P_i$ to $\Unhappy$.
           \item[--]  If $g_j(\alpha_i) \neq F(\alpha_j, \alpha_i)$, then include $P_j$ to $\Unhappy$.
        \end{myitemize}
      \item[--] If $|\Unhappy| > t$, then discard $\D$.
      \end{myitemize}    
\item {\bf Round IV (resolving unhappy parties) --- the dealer and each $P_j \not \in \Unhappy$ does the following}:
       \begin{myitemize}
       \item[--] For every $P_i \in \Unhappy$, $\D$ broadcasts degree-$t$ polynomial $f_i(x)$ and 
        $P_j$ broadcasts $g_j(\alpha_i)$.       
       \end{myitemize}
\item {\bf Output decision --- each party $P_k$ does the following}:
    \begin{myitemize}
    \item[--] If  there exists any $P_i \in \Unhappy$ for which $\D$ broadcasts $f_i(x)$ and at most $2t$ parties 
    $P_j  \not \in \Unhappy$ broadcast $g_j(\alpha_i)$ values where $f_i(\alpha_j) = g_j(\alpha_i)$ holds, then 
       discard $\D$.
    \item[--] Else output the {\it share} $f_k(0)$.
    \end{myitemize}     
\end{myitemize}
 \end{protocolsplitbox}
 
 \subsubsection{\bf The $3$-round $\GIKRII$ Scheme}
 \label{sec:GIKRII}
  We now present the 
  $\GIKRII$ scheme from \cite{GIKR01}, which has a {\it round-optimal} sharing phase, namely a 3-round sharing phase.
  However, the protocol is {\it inefficient}, as it requires an exponential (in $n$ and $t$) amount of computation and communication. 
 The computations in $\GIKRII$ scheme are done over a finite group $(\Gr, +)$. 
  We first explain $t\mbox{-out-of-}n$ {\it replicated secret-sharing} (RSS) \cite{ISN87}. Let $K \defined {n \choose t}$ and $A_1, \ldots, A_K$ denote the set of
 all possible subsets of $\Partyset$ of size $t$. For $k = 1, \ldots, K$, let $\Group_k = \Partyset \setminus A_k$. It is easy to see that in each $\Group_k$, the majority
 of the parties are {\it honest}. 
   To share $s \in \Gr$, the share-generation algorithm of RSS outputs $(v^{(1)}, \ldots, v^{(K)}) \in \Gr^K$, where
  $v^{(1)}, \ldots, v^{(K)}$ are random elements, such that $v^{(1)} + \ldots + v^{(K)} = s$ holds.
   The {\it share} $s_i$ for $P_i$ is defined as $s_i \defined \{v^{(j)} \}_{P_i \in \Group_j}$. 
   Any $t$-sized subset of $(s_1, \ldots, s_n)$ will have at least one ``missing" element from $v^{(1)}, \ldots, v^{(K)}$,
   say $v^{(l)}$, whose probability distribution will be independent of $s$, thus ensuring privacy. On the other hand, any $(t + 1)$-sized subset of $(s_1, \ldots, s_n)$ will have
   all values $v^{(1)}, \ldots, v^{(K)}$ which can be added to reconstruct back $s$, thus ensuring correctness. 
    We say that {\it $s \in \Gr$ is RSS-shared}, if it is secret-shared as per 
   $t\mbox{-out-of-}n$ RSS. That is, if there exist $v^{(1)}, \ldots, v^{(K)}$ where $s = v^{(1)} + \ldots + v^{(K)}$, with all the parties in $\Group_k$ holding 
   $v^{(k)}$. 
   
   Scheme $\GIKRII$ is presented in Fig \ref{fig:GIKRII}. The sharing protocol
   $\GIKRII\mbox{-}\Sh$ verifiably generates a replicated secret-sharing of dealer's secret, maintaining its privacy if $\D$ is {\it honest}. The verifiability ensures that even if
   $\D$ is {\it corrupt}, there exists some value which has been shared as per RSS by $\D$. The reconstruction protocol $\GIKRII\mbox{-}\Rec$ allows the parties
   to reconstruct the RSS-shared value of $\D$. During $\GIKRII\mbox{-}\Sh$, $\D$ generates a vector of values $(v^{(1)}, \ldots, v^{(K)})$ as per
   the share-generation algorithm of RSS and sends $v^{(k)}$ to all the parties in $\Group_k$. If $\D$ is {\it honest}, then this ensures the privacy of $s$. This is because
   if $\Adv$ corrupts the parties in $A_k$, then it will not know $v^{(k)}$.
   To ensure that a potentially {\it corrupt} $\D$ has distributed the same $v_k$ to all the (honest) parties in $\Group_k$, each pair of parties in $\Group_k$ privately exchange
   their respective copies of $v^{(k)}$ and publicly raise a complaint if they find any inconsistency. To resolve any complaint raised for $\Group_k$, $\D$ makes public
   the value $v^{(k)}$ for $\Group_k$, 
   thus ensuring that all the parties in $\Group_k$ have
    the same $v^{(k)}$. Notice that this does not violate privacy, since
     if any inconsistency
    is reported for $\Group_k$, then either $\D$ is {\it corrupt} or $\Group_k$ consists of at least one corrupt party and so adversary already knows $v^{(k)}$. 
    
    The above process will require {\it four} rounds, which can be collapsed
    to three,  based on the idea of pre-exchanging pair-wise random pads.
      During $\GIKRII\mbox{-}\Rec$,  the goal of each $P_i$  is to correctly obtain $v^{(1)}, \ldots, v^{(K)}$. 
      If $P_i \in \Group_k$, then it already has $v^{(k)}$. However, if $P_i \not \in \Group_k$, then
       every party in $\Group_k$ sends $v^{(k)}$ to $P_i$, who 
    applies the majority rule  to filter out the correct $v^{(k)}$.
   \begin{schemesplitbox}{$\GIKRII$}{The $3$-round $\GIKRII$ scheme due to Genarro, Ishai, Kushilevitz and Rabin \cite{GIKR01}.}{fig:GIKRII}
\justify
\centerline{\algoHead{Sharing Phase: Protocol $\GIKRII\mbox{-}\Sh$}}
\begin{myitemize}
  \item {\bf Round I (distributing shares and exchanging random pads)}: Let $K = {n \choose t}$ and $A_1, \ldots, A_K$ denote the set of
    all possible subsets of $\Partyset$ of size $t$ and let $\Group_k = \Partyset \setminus A_k$.
   \begin{myitemize}
    \item[--] $\D$ on having the input $s \in \Gr$, randomly selects $v^{(1)}, \ldots, v^{(K)} \in \Gr$ such that $s = v^{(1)} + \ldots + v^{(K)}$ holds. It then sends
    $v^{(k)}$ to every party $P_i \in \Group_k$, for $k = 1, \ldots, K$.
    \item[--] For $k = 1, \ldots, K$, each $P_i \in \Group_k$ sends a randomly chosen pad $r^{(k)}_{ij} \in \Gr$ to every $P_j \in \Group_k$ where $i < j$.   
   \end{myitemize}
  \item {\bf Round II (pair-wise consistency check within each group)}
     \begin{myitemize}
     \item[--] For $k = 1, \ldots, K$, each pair of parties $P_i, P_j \in \Group_k$ with $i < j$ do the following:
         \begin{myitemize}
         \item[--] $P_i$ broadcasts $a^{(k)}_{ij} = v^{(k)}_i + r^{(k)}_{ij}$, where $v^{(k)}_i$ denotes the version of $v^{(k)}$ received by $P_i$ from $\D$.
         \item[--] $P_j$ broadcasts $a^{(k)}_{ji} = v^{(k)}_j + r'^{(k)}_{ij}$, where $v^{(k)}_j$ denotes the version of $v^{(k)}$ received by $P_j$ from $\D$ and $r'^{(k)}_{ij}$ 
          denotes the pad received by $P_j$ from $P_i$.
         \end{myitemize}
     \end{myitemize}
  \item {\bf Round III (resolving conflicts)}
       \begin{myitemize}
       \item[--] For $k = 1, \ldots, K$, if there exists $P_i, P_j \in \Group_k$ such that $a^{(k)}_{ij} \neq a^{(k)}_{ji}$, then $\D$ broadcasts
        the value $v^{(k)}$. 
             \end{myitemize}     
\item {\bf Output determination --- each party $P_i$ does the following}:
    \begin{myitemize}
    \item[--] If $\exists k \in \{1, \ldots, K \}$ where $P_i \in \Group_k$ such that $\D$ broadcast $v^{(k)}$, then set $v^{(k)}_i$ to $v^{(k)}$.
       \item[--] Output the {\it share} $s_i \defined \{v^{(k)}_i \}_{P_i \in \Group_k}$. \\[.1cm]
    \end{myitemize}        
  \end{myitemize}

\justify
\centerline{\algoHead{Reconstruction Phase: Protocol $\GIKRII\mbox{-}\Rec$}}
Each party $P_i \in \Partyset$ does the following:
\begin{myitemize}
\item[--] $\forall k \in \{1, \ldots, K \}$, such that $P_i \in \Group_k$, send $v^{(k)}_i$ to every party in $\Partyset \setminus \Group_k$.
\item[--] $\forall k \in \{1, \ldots, K \}$ where $P_i \not \in \Group_k$, set $v^{(k)}$ to be the value $v^{(k)}_j$ received from at least $t + 1$ parties $P_j \in \Group_k$.
\item[--] Output $s =\displaystyle \sum_{\Group_k: P_i \in \Group_k} v^{(k)}_i + \displaystyle \sum_{\Group_k: P_i \not \in \Group_k} v^{(k)}$.
\end{myitemize}
  \end{schemesplitbox}
  \subsubsection{\bf The $3$-round $\FGGRS$ Scheme}
  The $\GIKRI$ protocol comes closest in terms of the number of rounds to obtain a round optimal and
  {\it efficient} VSS scheme. Round IV of $\GIKRI\mbox{-}\Sh$ consists of $\D$ making public the polynomials
  of the {\it unhappy} parties. 
  Fitzi et al.~\cite{FGGRS06} observed that the elimination
  of Round IV results in a primitive that satisfies a {\it weaker} commitment property, where the
  reconstructed value may be some predefined default value, when the dealer is {\it corrupt}. 
   This
  primitive is called {\it weak} verifiable secret sharing (WSS) \cite{RB89} and is used as a building block to
  construct a VSS scheme. 
   We next present the required background and the WSS scheme
    of \cite{FGGRS06}.
   \begin{definition}[$(n,t)$-WSS \cite{RB89}]
   \label{def:WSS}
   Let $(\Sh, \Rec)$ be a pair of protocols where $\D \in \Partyset$ has a private input $s \in \SecretSpace$ for
   $\Sh$. Then $(\Sh, \Rec)$ is a {\it perfectly-secure $(n,t)$-WSS scheme}, if the
   following hold.
   \begin{myitemize}
     \item[--] {\bf Privacy and Correctness}: Same as in VSS.
     \item[--] {\bf Weak Commitment}: Even if $\D$ is {\it corrupt}, in any execution of $\Sh$ the
       joint view of the honest parties defines a unique value $s^{\star} \in \SecretSpace$ (which could be
       different from $s$), such that each honest party outputs either $s^{\star}$ or some default value
       $\default$ at the end of $\Rec$, irrespective of $\Adv$.
      \end{myitemize}
  \end{definition}

  The WSS scheme $\FGGRSWSS$ of Fitzi et al.~\cite{FGGRS06} is given in Fig. \ref{fig:FGGRSWSS}.
  Protocol $\FGGRSWSS\mbox{-}{\Sh}$ is the same as $\GIKRI\mbox{-}\Sh$, except that the parties do not execute Round IV. Consequently,
  the parties in $\Unhappy$ will not posses their shares.
   Hence, during $\FGGRSWSS\mbox{-}{\Rec}$, only the {\it happy} parties (who are {\it not} in $\Unhappy$) 
   participate  by broadcasting their respective polynomials. To verify that 
   correct polynomials are broadcasted, the pair-wise consistency of these polynomials is checked.
   The polynomials which are {\it not} found to be pair-wise consistent with ``sufficiently many" polynomials are not considered.
    If the parties are left with at least $n - t$  polynomials,
    then they are used to reconstruct back 
   $\D$'s committed Shamir-sharing polynomial, else the parties output $\bot$.
   The idea is that if the parties are left with $n - t$ polynomials,
   then they lie on the same degree-$(t, t)$ bivariate polynomial as committed by $\D$ to {\it honest} happy parties during the sharing phase,
   as among these polynomials, at least $t + 1$ belong to the honest parties who are happy.
   
      For an {\it honest} $\D$, {\it no} honest party will be in $\Unhappy$ and hence
      all honest parties will have their respective shares of $\D$'s Shamir-sharing polynomial.
      Moreover, if any {\it corrupt} party produces an incorrect polynomial during the reconstruction phase, then it will be ignored
      due to the pair-wise consistency checks. Thus
      $\FGGRSWSS$ achieves the properties of a type-II VSS, for an {\it honest} $\D$. However if $\D$ is {\it corrupt}, 
      then up to $t$ {\it honest} parties may belong to $\Unhappy$.
      Moreover, during the reconstruction phase, even if a single corrupt party produces incorrect polynomials, then the parties reconstruct $\bot$.  
      And this prevents $\FGGRSWSS$ from being a VSS scheme.         
    \begin{schemesplitbox}{$\FGGRSWSS$}{The 3-round $\FGGRSWSS$ scheme due to Fitzi, Garay, Gollakota, Rangan and Srinathan \cite{FGGRS06}.}{fig:FGGRSWSS}
    \justify
    \centerline{\algoHead{Sharing Phase: Protocol $\FGGRSWSS\mbox{-}{\Sh}$}}
    \begin{myitemize}
      \item The parties execute the first $3$ rounds of $\GIKRI\mbox{-}\Sh$. Let $\Unhappy$ be the set of 
      {\it unhappy} parties and let $\WCORE \defined \Partyset \setminus \Unhappy$ be the set of {\it happy} parties.
      If $|\Unhappy| > t$, then discard $\D$. \\[.1cm]
    \end{myitemize}
    \justify
    \centerline{\algoHead{Reconstruction Phase: Protocol $\FGGRSWSS\mbox{-}{\Rec}$}}
    \begin{myitemize}
    \item {\bf Revealing private information}: each $P_i \in \WCORE$ broadcasts $f_i(x)$ and $g_i(y)$.
    \item {\bf Consistency check and output decision --- each party $P_i$ does the following}:
      \begin{myitemize}
      \item[--] Construct a {\it consistency graph} $G$ over the set of parties $\WCORE$
        with an edge between $P_j$ and $P_k$
         if and only if $f_j(\alpha_k) = g_k(\alpha_j)$ and
        $g_j(\alpha_k) = f_k(\alpha_j)$.
      \item[--] Remove $P_j$ from $G$, if it  has degree less than $n-t$ in $G$. Repeat till no more nodes can be removed from $G$. Redefine
      $\WCORE$ to be the set of parties, whose corresponding nodes remain in $G$.
     \item[--] If $|\WCORE| < n - t$  then output a default value $\default$.
      Else interpolate a degree-$t$ polynomial $q(\cdot)$ through
      the points $\{(\alpha_j, f_j(0))\}_{P_j \in \WCORE}$ and output $q(\cdot), s = q(0)$.
      \end{myitemize}
    \end{myitemize}
  \end{schemesplitbox}
  \paragraph{\bf Sharing and Reconstructing Polynomial Using $\FGGRSWSS$}
  One can interpret $\D$'s computation in $\FGGRSWSS\mbox{-}{\Sh}$ as if $\D$ wants to share the degree-$t$ Shamir-sharing
  polynomial $\starF(0, y)$.
  If $\D$ is not discarded, then each (honest) $P_j \in \WCORE$
  receives the share $\starF(0, \alpha_j)$ from $\D$ through its
   degree-$t$ row-polynomial $\starF(x, \alpha_j)$.
     Here
   $\starF(x, y)$ is the degree-$(t, t)$ bivariate polynomial committed by $\D$ to the honest parties in $\WCORE$, which is the same as $F(x, y)$ for an {\it honest} 
   $\D$. If $\D$ is {\it honest}, then adversary learns at most $t$ shares lying on $F(0, y)$
   and hence its view will be independent of $F(0, 0)$.
   Similarly, the computations during $\FGGRSWSS\mbox{-}{\Rec}$ can be interpreted as if
   the parties publicly try to reconstruct a degree-$t$ Shamir-sharing polynomial \starF(0, y), which has been shared by $\D$
   during $\FGGRSWSS\mbox{-}{\Sh}$. 
   If $\D$ is {\it honest}, then the parties robustly reconstruct the shared polynomial. Else, the parties either reconstruct the shared polynomial
   or output $\bot$. 
     Hence we propose the following notations for $\FGGRSWSS$, which
      later simplifies the presentation of $\FGGRS$.
\begin{notation}[\bf Notations for using $\FGGRSWSS$]
\label{notation:FGGRS}
We use the following notations.
\begin{myitemize}
\item We say that party $P_j \in \Partyset$ {\it shares a degree-$t$ polynomial $r(\cdot)$} held by $P_j$, to denote that $P_j$
 plays the role of $\D$ and invokes an instance of $\FGGRSWSS\mbox{-}{\Sh}$ by selecting $r(\cdot)$ as its Shamir-sharing polynomial
   and all the parties participate in this instance. 
\item We say that {\it $P_i$ receives a wss-share $r_{ji}$ from $P_j$}, to denote that 
 in Round I of the $\FGGRSWSS\mbox{-}{\Sh}$ instance invoked by $P_j$, $P_i$ receives a degree-$t$ row-polynomial
 from $P_j$, whose constant term is $r_{ji}$. If $P_j$ is not discarded during the $\FGGRSWSS\mbox{-}{\Sh}$ instance, then
 the wss-shares $r_{ji}$ of all the {\it honest} parties in $\WCORE$ lie on a unique degree-$t$ Shamir-sharing polynomial held by $P_j$.
\item Let $r(\cdot)$ be a degree-$t$ polynomial shared by $P_j$ through an instance of $\FGGRSWSS\mbox{-}{\Sh}$.
 We say that the {\it parties try to reconstruct $P_j$'s shared polynomial}, to denote that the parties execute
  the corresponding instance of $\FGGRSWSS\mbox{-}{\Rec}$, which either outputs $r(\cdot)$ or $\bot$.
\end{myitemize}
\end{notation}
\noindent \paragraph{\bf From WSS to VSS}
   The sharing-phase protocol of $\FGGRS$ (see Fig. \ref{fig:FGGRSVSS}) is the same as the first three rounds of 
   $\GIKRI\mbox{-}\Sh$ with the following twist. The random pads $r_{ji}$ used by $P_j$ to verify the pair-wise consistency 
   of its row-polynomial with the other parties' column-polynomials are ``tied together" by letting these pads lie
   on a random degree-$t$ {\it blinding-polynomial} $r_j(\cdot)$, which is shared by $P_j$ (see Notation \ref{notation:FGGRS}). 
   For pair-wise consistency, $P_j$ makes public the polynomial $A_j(\cdot)$, which is a masked version of its row-polynomial
   and its blinding-polynomial, 
    while every $P_i$
   makes public the supposedly common value on its column-polynomial, blinded with the wss-share of $P_j$'s blinding-polynomial.
   This new way of performing pair-wise consistency checks achieves the same ``goals" as earlier.
    Moreover, privacy is maintained for an {\it honest} $\D$, as
    for every {\it honest} $P_j$, the adversary obtains at most $t$ wss-shares of $P_j$'s blinding-polynomial, which are randomly distributed.
   The set of happy parties for the VSS is identified based on the results of pair-wise consistency and conflict-resolutions, as done in $\GIKRI\mbox{-}\Sh$. Moreover, 
   the parties also ensure that for every happy party $P_j$ for the VSS, there is an overlap of at least $n - t$ between
   the set of happy parties for the VSS and the set of happy parties 
      for $P_j$'s WSS-sharing instance. This  is crucial for ensuring the strong commitment property during the reconstruction phase.

  Let $\starF(x, y)$ be the degree-$(t, t)$ bivariate polynomial, committed by $\D$ during the sharing phase.
   During the reconstruction phase,
   instead of asking the {\it happy} parties to make their row-polynomials public, 
   the parties reconstruct their blinding-polynomials (by executing instances of
   $\FGGRSWSS\mbox{-}{\Rec}$), which are then unmasked from the corresponding $A_j(\cdot)$ polynomials to get back the row-polynomials
   of the happy parties. For the {\it honest} happy parties $P_j$, robust reconstruction of their blinding-polynomials is always guaranteed, thus ensuring
   that their row-polynomials $\starF(x, \alpha_j)$
   are robustly reconstructed. If the reconstruction of $P_j$'s blinding-polynomial fails,
   then $P_j$ is {\it corrupt} and hence can be safely discarded from consideration. 
   However, if $P_j$ is {\it corrupt} and the parties reconstruct a degree-$t$ polynomial during the corresponding $\FGGRSWSS\mbox{-}{\Rec}$ instance,
   then the {\it weak commitment} property of the WSS ensures that reconstructed polynomial is the correct blinding-polynomial.
    Hence unmasking it from 
   $A_j(\cdot)$ will return back the row-polynomial $\starF(x, \alpha_j)$. This is because there are at least
   $n - t$ parties, which belong to the happy set of {\it both} the VSS instance, as well as $P_j$'s WSS instance. Among these
   $n - t$ common happy-parties, at least $t+1$ are {\it honest}. And the wss-shares received by these honest parties $P_k$ from $P_j$ uniquely define
   $P_j$'s blinding-polynomial $r_j(\cdot)$, while evaluations of the column-polynomials of the same honest-parties $P_k$ at $y = \alpha_j$ uniquely determine
   the degree-$t$ row-polynomial $\starF(x, \alpha_j)$. Moreover, during the sharing phase these honest parties $P_k$ 
   collectively ensured that $A_j(\cdot) = \starF(x, \alpha_j) + r_j(\cdot)$ holds,
   as otherwise they do not belong to the happy set of $P_j$'s WSS instance.
  
     \begin{schemesplitbox}{$\FGGRS$}{The 3-round $\FGGRS$ scheme due to Fitzi, Garay, Gollakota, Rangan and Srinathan \cite{FGGRS06}.}{fig:FGGRSVSS}
    \justify
    \centerline{\algoHead{Sharing Phase: Protocol $\FGGRS_{\Sh}$}}
    \begin{myitemize}
    \item {\bf Round I (sending polynomials and exchanging random pads)}:
      \begin{myitemize}
      \item[--] $\D$ picks $F(x, y)$ and distributes $f_i(x), g_i(y)$ lying on $F(x, y)$ to each $P_i$.
      \item[--] Each party $P_i \in \Partyset$ (including $\D$) picks a random degree-$t$ {\it blinding-polynomial} $r_i(\cdot)$ and shares
      it through an instance $\WSSSh_i$  of $\FGGRSWSS\mbox{-}{\Sh}$.
      \end{myitemize}
    \item {\bf Round II (broadcasting common values in a masked fashion) --- each $P_i$ does the following}:
      \begin{myitemize}
      \item[--] Broadcast the degree-$t$ polynomial $A_i(\cdot) \defined f_i(x) + r_i(\cdot)$.
      \item[--] Broadcast $b_{ij} \defined g_i(\alpha_j) + r'_{ji}$, where
        $r'_{ji}$ denotes the wss-share received from $P_j$ during $\WSSSh_j$.
      \item[--] For every $k \in \{1,\ldots,n\}$, concurrently execute Round II of the instance $\WSSSh_k$.
      \end{myitemize}
    \item {\bf Round III}:
      \begin{myitemize}
           \item[--] {\bf (making disputed values public)} --- for all $P_i, P_j$
        where $A_i(\alpha_j) \neq b_{ji}$, dealer $\D$ broadcasts $F(\alpha_j, \alpha_i)$, party 
        $P_i$ broadcasts $f_i(\alpha_j)$ and party  $P_j$ broadcasts $g_j(\alpha_i)$.
      \item[--] For every $k \in \{1,\ldots,n\}$, the parties concurrently execute Round III of the instance $\WSSSh_k$.
      \end{myitemize}
    \item {\bf Local computation at the end of Round III --- each party $P_k$ does the following}:
      \begin{myitemize}
      \item[--] Initialize a set $\Unhappy$ of {\it unhappy} parties.  For every $P_i, P_j$ where $A_i(\alpha_j) \neq b_{ji}$, do the following.
         \begin{myitemize}
         \item[--] Include $P_i \in \Unhappy$ (resp.~$P_j \in \Unhappy$), if $F(\alpha_j, \alpha_i) \neq f_i(\alpha_j)$ (resp.~$F(\alpha_j, \alpha_i) \neq g_j(\alpha_i)$) holds.
         \end{myitemize}
      Let 
      $\VCORE \defined \Partyset \setminus \Unhappy$ be the set of {\it happy} parties.
      \item[--] For $j \in \{1, \ldots, n \}$, let $\WCORE_j$ denote the set of {\it happy} parties during
      the instance $\WSSSh_j$. Remove $P_i$ from $\WCORE_j$, if during Round II, $A_j(\alpha_i) \neq b_{ij}$ holds.
          %
      \item[--] For every $P_j \in \VCORE$, if
        $|\VCORE \intersection \WCORE_j| < n - t$, then remove $P_j$ from $\VCORE$.  Repeat this step till no more parties can be removed from
        $\VCORE$. 
        If $|\VCORE| < n - t$, then discard $\D$. \\[.1cm]
      \end{myitemize}
    \end{myitemize}
    \justify 
    \centerline{\algoHead{Reconstruction Phase: Protocol $\FGGRS_{\Rec}$}}
    \begin{myitemize}
      \item {\bf Reconstructing the blinding-polynomials}:
       \begin{myitemize}
       \item[--] $\forall P_j \in \VCORE$, the parties try to reconstruct $P_j$'s blinding-polynomial by participating in an instance
       $\WSSRec_j$ of $\FGGRSWSS\mbox{-}{\Rec}$. $P_j$ is removed from $\VCORE$ if $\bot$ is the output during $\WSSRec_j$.
       \end{myitemize}
      \item {\bf Output decision --- each party $P_i$ does the following}:
      \begin{myitemize}
        \item[--] For each $P_j \in \VCORE$, compute $f_j(x) = A_j(x) - r_j(\cdot)$, where $r_j(\cdot)$ is 
        reconstructed during $\WSSRec_j$.
         Interpolate a degree-$t$ polynomial $q(\cdot)$ through $\{(\alpha_j, f_j(0))\}_{P_j \in \VCORE}$. Output $s = q(0)$.
      \end{myitemize}
    \end{myitemize}
  \end{schemesplitbox}

  \subsubsection{\bf The $3$-round $\KKK$ Scheme}
  The $\FGGRS$ scheme is a Type-I VSS, because if $\D$ is {\it corrupt}, then {\it only} the
  honest parties in $\VCORE$ get their shares.  
  Moreover, it
  makes use of the broadcast channel during two of the rounds of the sharing phase, which is {\it not} optimal (the optimal is one round). The $\KKK$
  scheme due to Katz et al.~\cite{KKK09} rectifies both these problems.

  The optimal broadcast-channel usage must first be rectified for $\FGGRSWSS\mbox{-}\Sh$. The modified
   construction  $\KKKWSS\mbox{-}\Sh$ (see Fig \ref{fig:KKKWSS}) is based on the observation that during
  Round II of $\FGGRSWSS\mbox{-}{\Sh}$ (which is the same as Round II of $\GIKRI\mbox{-}\Sh$), there is no need to
  {\it publicly} perform the pair-wise consistency checks over {\it masked} values. Instead, 
   parties can first {\it privately} perform the pair-wise consistency checks over {\it unmasked} values 
   and later {\it publicly} announce the results during the third round. However, we also need to add a provision for the dealer to resolve any potential conflicts during the
   third round itself. For this, during the second round, every $P_i, P_j$ privately exchange their supposedly common values and also the random pads.
   Additionally, the pads  are also ``registered" with the dealer. Later during the third round, upon a
  disagreement between $P_i$ and $P_j$, they broadcast their respective common values and their respective random
  pads, else they broadcast their appropriately masked common values. In parallel, dealer either broadcasts the common value in a masked fashion if the pads it received from
  $P_i, P_j$ are same, else it just broadcasts the common value. The secrecy of the common values is
  maintained if $P_i$, $P_j$ and the dealer are honest. On the other hand, if dealer is corrupt
   and if there is a disagreement between honest $P_i, P_j$ , then the dealer can ``take side" 
   with at most one of them during the third round. 
  %
     %

    \begin{protocolsplitbox}{$\KKKWSS\mbox{-}\Sh$}{The $3$-round $\KKKWSS\mbox{-}\Sh$ protocol due to  Katz, Koo and Kumaresan \cite{KKK09}.}{fig:KKKWSS}
    \justify
    \begin{myitemize}
    \item {\bf Round I (sending polynomials and exchanging random pads)}:
      \begin{myitemize}
      \item[--] $\D$ distributes $f_i(x), g_i(y)$ lying on $F(x, y)$ to each $P_i$.
      In parallel,
      each $P_i$ sends a random pad $r_{ij}$ to $P_j$ and additionally 
      the pad-list $\{r_{ij} \}_{P_j \in \Partyset}$ to $\D$.
      \end{myitemize}
    \item {\bf Round II (exchanging common values and confirming pad) --- each $P_i$ does the following}:
      \begin{myitemize}
      \item[--] Send $a_{ij} = f_i(\alpha_j)$ and $b_{ij} = g_i(\alpha_j)$ to $P_j$.
      \item[--] Send $\{r'_{ji} \}_{P_j \in \Partyset}$ to $\D$, where $r'_{ji}$ denotes the pad received from $P_j$ during Round I.
      \end{myitemize}
    \item {\bf Round III (complaint and resolution) --- each party $P_i$ does the following}:
      \begin{myitemize}
      \item[--] For $j \in \{1, \ldots, n \}$, let $a'_{ji}$ and $b'_{ji}$ be the values received from $P_j$ during Round II.
            \begin{myitemize}
              \item[--] If $b'_{ji} \neq f_i(\alpha_j)$, broadcast $(j, \disagreef, f_i(\alpha_j), r_{ij})$,
        		else broadcast $(j, \agreef, f_i(\alpha_j) + r_{ij})$.
	      \item[--] If $a'_{ji} \neq g_i(\alpha_j)$, broadcast $(j, \disagreeg, g_i(\alpha_j), r'_{ji})$,
        		else broadcast $(j, \agreeg, \allowbreak g_i(\alpha_j) + r'_{ji})$.
             \end{myitemize}
         
      \item[--] If $P_i = \D$, then for every ordered pair of parties $(P_j, P_k)$, {\it additionally} do the following.
            \begin{myitemize}
            \item[--]  Let $r_{jk}^{(1)}$ and $r_{jk}^{(2)}$ be the pads received from $P_j$ and $P_k$ respectively
             during Round I and Round II.
            If $r_{jk}^{(1)} \neq r_{jk}^{(2)}$, then broadcast $((j, k),\NEQ, F(\alpha_k, \alpha_j))$, else broadcast $((j, k),\EQ, F(\alpha_k, \alpha_j) + r_{jk}^{(1)})$.
            \end{myitemize}
      \end{myitemize}
    \item {\bf Local computation (identifying unhappy parties) --- each party $P_k$ does the following}:
      \begin{myitemize}
      \item[--] Initialize a set of {\it unhappy} parties $\Unhappy$ to $\emptyset$.  For every $P_i, P_j$ such that
       $P_i$ broadcasts $(j, \disagreef, \allowbreak f_i(\alpha_j), r_{ij})$ and $P_j$ broadcasts $(i, \disagreeg,  g_j(\alpha_j),  r'_{ij})$ where
      $r_{ij} = r'_{ij}$, do the following.
      \begin{myitemize}
      \item[--] Include $P_i$ to $\Unhappy$, if one of the following holds.
              \begin{myitemize}
               \item[--] During Round III, $\D$ broadcasts $((i, j),\NEQ, d_{ij})$, such that $d_{ij} \neq f_i(\alpha_j)$.
                \item[--] During Round III, $\D$ broadcasts $((i, j),\EQ, d_{ij})$, such that $d_{ij} \neq f_i(\alpha_j) + r_{ij}$.                
              \end{myitemize}
         \item[--] Include $P_j$ to $\Unhappy$, if one of the following holds.
              \begin{myitemize}
               \item[--] During Round III, $\D$ broadcasts $((i, j),\NEQ, d_{ij})$, such that $d_{ij} \neq g_j(\alpha_i)$.
                \item[--] During Round III, $\D$ broadcasts $((i, j),\EQ, d_{ij})$, such that $d_{ij} \neq g_j(\alpha_i) + r'_{ij}$.                
              \end{myitemize}              
          \end{myitemize}
          \item[--] If $|\Unhappy| > t$, then discard $\D$. Else let $\WCORE = \Partyset \setminus \Unhappy$ be the set of {\it happy} parties.
          
      \end{myitemize}
    \end{myitemize}
    
   \end{protocolsplitbox}
\noindent \paragraph{\bf From $\KKKWSS$ to $\KKK$}
The notion of sharing and reconstructing degree-$t$ polynomials (as per Notation \ref{notation:FGGRS})
 is applicable even for $\KKKWSS$. 
  Replacing $\FGGRSWSS$ with $\KKKWSS$ in the $\FGGRS$ readily provides a VSS scheme
 with the optimal usage of the broadcast channel. However, the resultant VSS {\it need not} be of Type-II if $\D$ is {\it corrupt}, as the {\it unhappy honest} parties may not get their shares. 
   Hence $\KKK$ deploys an additional trick to get rid of this problem.

  Let $P_i$ be an {\it unhappy} party during $\FGGRS\mbox{-}{\Sh}$ and let 
  $\starF(x, y)$ be $\D$'s committed bivariate polynomial. 
  Note that each happy $P_j$ broadcasts a masking $A_j(\cdot)$ of $\starF(x, \alpha_j)$.
  If $P_i$ is {\it happy} in $\WSSSh_j$, then $P_i$ can compute the point $\starF(\alpha_i, \alpha_j)$ on its supposedly
  column-polynomial $\starF(\alpha_i, y)$ by unmasking the wss-share $r'_{ji}$ computed during $\WSSSh_j$ from $A_j(\alpha_i)$.
  Hence if there is a set $\Support_i$ of at least $t+1$ happy parties $P_j$, who keep $P_i$ happy during $\WSSSh_j$ instances,
  then $P_i$ can compute $t+1$ distinct points on  $\starF(\alpha_i, y)$ and hence
  get $\starF(\alpha_i, y)$. However, it is not clear how to extend the above approach to enable $P_i$ obtain its row-polynomial
  $\starF(x, \alpha_i)$, which is required for $P_i$ to obtain its share of $\D$'s Shamir-sharing polynomial $\starF(0, y)$. The way-out is to let
  $\D$ use a {\it symmetric bivariate polynomial}, 
    which ensures that $\starF(x, \alpha_i) = \starF(\alpha_i, y)$ holds. 
   
   The above idea requires broadcast during the second and third round. To ensure the optimal usage of the
   broadcast channel, the technique used in $\KKKWSS\mbox{-}\Sh$ is deployed, along with postponing the broadcast of masked row-polynomials 
   to the third round. However, this brings additional challenges to filter out the parties from $\WCORE_j$ sets for the individual
   WSS instances.
    For example, a {\it corrupt} $\D$ can distribute pair-wise {\it inconsistent} polynomials  in such a way that the masked polynomial
   $A_j(\cdot)$ broadcast by an {\it honest happy} party $P_j$ is inconsistent with the corresponding masked value broadcast by an {\it honest unhappy} party
   $P_i$, even though $P_i$ belongs to $\WCORE_j$ during $\WSSSh_j$.
   Simply removing $P_i$ from $\WCORE_j$ in this case (as done in $\FGGRS\mbox{-}{\Sh}$) might end up resulting in $P_i$ being removed from the $\WCORE_j$ sets
   of every honest happy party. And this may lead to $\Support_i$ set of size less than $t+1$. To prevent this, 
   apart from pair-wise consistency checks of masked row-polynomials, the parties also carefully consider the results of {\it private} pair-wise consistency checks 
   performed during the second round, whose results are {\it public} during the third round (see Fig. \ref{fig:KKKVSS}).
   We stress that even though each party obtains a single polynomial from $\D$, it is treated both as row as well
   as column-polynomial to perform the pair-wise consistency checks.
   Accordingly, if $P_i$ finds a ``negative" result for the {\it private} pair-wise consistency check with $P_j$, then it broadcasts both $\disagreef$ and $\disagreeg$ messages
   against $P_j$.
    Else it broadcasts just an $\agreeg$ message for $P_j$;
    the $\agreef$ message for $P_j$ is assumed to be {\it implicitly}
   present in the latter case. 
   The reconstruction phase of $\KKK$ scheme is the same as $\BGW$.
   
  \begin{protocolsplitbox}{$\KKK$}{The $3$-round sharing-phase protocol due to Katz, Koo and Kumaresan \cite{KKK09}.}{fig:KKKVSS}
    \justify
    \centerline{\algoHead{Sharing Phase: Protocol $\KKK\mbox{-}{\Sh}$}}
    \begin{myitemize}
    \item {\bf Round I (sending polynomials and exchanging random pads)}:
      \begin{myitemize}
      \item[--] $\D$ embeds its {\it Shamir-sharing polynomial}
        $q(\cdot)$ in a {\it random} degree-$(t, t)$ {\it symmetric} bivariate polynomial
        $F(x, y)$ at $x = 0$
        and sends only the row-polynomial
        $f_i(x) = F(x, \alpha_i)$ to party $P_i$.
      \item[--] Each party $P_i \in \Partyset$ (including $\D$) picks a random degree-$t$ {\it blinding
          polynomial} $r_i(\cdot)$ and shares it through an instance $\WSSSh_i$  of $\KKKWSS\mbox{-}{\Sh}$.
          In addition, $P_i$ sends the polynomial $r_i(\cdot)$ to $\D$.
      \end{myitemize}
    \item {\bf Round II (exchanging common values and confirming pad) --- each $P_i$ does the following}:
      \begin{myitemize}
      \item[--] For $j = 1, \ldots, n$,  send $a_{ij} =  f_i(\alpha_j)$ to $P_j$.
        Send $\{r'_{ji} \}_{j = 1, \ldots, n}$ to $\D$, where $r'_{ji}$ is the wss-share received from $P_j$ during Round I of $\WSSSh_j$.
      \item[--] For $j = 1, \ldots, n$ execute Round II of the instance $\WSSSh_j$.
      \end{myitemize}
    \item {\bf Round III (complaint and resolution) --- each party $P_i$ does the following}:
      \begin{myitemize}
        \item[--] Broadcast the degree-$t$ polynomial $A_i(\cdot) = f_i(x) + r_i(\cdot)$.
         \item[--] For $j \in \{1, \ldots, n \}$, let $a'_{ji}$ be the value received from $P_j$ during Round II.
              \begin{myitemize}
              \item[--] If $a'_{ji} \neq f_i(\alpha_j)$, then broadcast $(j, \disagreef, f_i(\alpha_j), r_{ij})$
               and $(j, \disagreeg, f_i(\alpha_j), r'_{ji})$.
              \item[--] Else broadcast $(j, \agreeg, f_i(\alpha_j) + r'_{ji})$.
             \end{myitemize}
         
      \item[--] If $P_i = \D$, then for every ordered pair of parties $(P_j, P_k)$, {\it additionally} do the following.
            \begin{myitemize}
            \item[--]  Let $r_{jk}^{(1)}$ and $r_{jk}^{(2)}$ be the pads received from $P_j$ and $P_k$ respectively
             during Round I and Round II.
            If $r_{jk}^{(1)} \neq r_{jk}^{(2)}$, then broadcast $((j, k),\NEQ, F(\alpha_k, \alpha_j))$, else broadcast $((j, k),\EQ, F(\alpha_k, \alpha_j) + r_{jk}^{(1)})$.
            \end{myitemize}
      \item[--] For every $j \in \{1,\ldots,n\}$, concurrently execute Round III of $\WSSSh_j$.
      \end{myitemize}
    \item {\bf Local computation at the end of Round III --- each party $P_k$ does the following}:
      \begin{myitemize}
      \item[--] Initialize a set of {\it unhappy} parties $\Unhappy$ to $\emptyset$.  For every $P_i, P_j$ such that 
       $P_i$ broadcasts $(j, \disagreef, \allowbreak f_i(\alpha_j), r_{ij})$ and $P_j$ broadcasts $(i, \disagreeg,  f_j(\alpha_i),  r'_{ij})$ where
      $r_{ij} = r'_{ij}$, do the following.
         \begin{myitemize}
              \item[--] Include $P_i$ to $\Unhappy$, if one of the following holds.
              \begin{myitemize}
               \item[--] During Round III, $\D$ broadcasts $((i, j), \NEQ, d_{ij})$, such that $d_{ij} \neq f_i(\alpha_j)$.
                \item[--] During Round III, $\D$ broadcasts $((i, j), \EQ, d_{ij})$, such that $d_{ij} \neq f_i(\alpha_j) + r_{ij}$.                
              \end{myitemize}
       \item[--] Include $P_j$ to $\Unhappy$, if one of the following holds.
              \begin{myitemize}
               \item[--] During Round III, $\D$ broadcasts $((i, j), \NEQ, d_{ij})$, such that $d_{ij} \neq f_j(\alpha_i)$.
                \item[--] During Round III, $\D$ broadcasts $((i, j), \EQ, d_{ij})$, such that $d_{ij} \neq f_j(\alpha_i) + r'_{ij}$.                
              \end{myitemize}              
          \end{myitemize}
          \item[--] Let $\VCORE = \Partyset \setminus \Unhappy$ be the set of {\it happy} parties
          and for every $P_j \in \VCORE$, let $\WCORE_j$ be the set of happy parties during $\WSSSh_j$.
           Remove $P_j$ from $\VCORE$, if any of
          the following holds:
            \begin{myitemize}
            \item[--] $|\WCORE_j| < n - t$.
            \item[--] $\exists i \in \{1, \ldots, n \}: P_j$ broadcasts $(i, \disagreef, f_j(\alpha_i), r_{ji})$ where
            $A_j(\alpha_i) \neq f_j(\alpha_i) + r_{ji}$.            
            \end{myitemize}
           \item[--] For every $P_j \in \VCORE$, remove $P_i$ from $\WCORE_j$, if any of the following holds.
               \begin{myitemize}
               \item[--] $P_i$ broadcasts  $(j, \agreeg, y)$ such that $A_j(\alpha_i) \neq y$.
                \item[--] $P_j$ broadcasts $(i, \disagreef, f_j(\alpha_i), r_{ji})$ and $P_i$ broadcasts either $(j, \agreeg, \star)$
                or  $(j, \allowbreak \disagreeg, \star, r'_{ji})$, where $r'_{ji} \neq r_{ji}$.
               \end{myitemize}
         \item[--] Remove $P_j$ from $\VCORE$ if $|\VCORE \cap \WCORE_j| < n - t$. Repeat, till no more parties can be removed from $\VCORE$. 
         \item[--] If $|\VCORE| < n - t$, then discard $\D$. 
      \end{myitemize}
    \item {\bf Computing shares --- each party $P_i \in \Partyset$ does the following}:
      \begin{myitemize}
      \item[--] If $P_i \in \VCORE$, then output the {\it share} $ f_i(0)$.
      Else recompute $f_i(x)$ as follows and output the {\it share} $f_i(0)$.
         \begin{myitemize}
         \item[--] Add $P_j$ to $\Support_i$, if $P_j \in \VCORE$ and $P_i \in \WCORE_j$.
         Interpolate  $\{(\alpha_j, A_j(\alpha_i) - r'_{ji}) \}_{P_j \in \Support_i}$ to compute $f_i(x)$. 
         \end{myitemize}
           \end{myitemize}
    \end{myitemize}
    \end{protocolsplitbox}
  \subsubsection{\bf The $3$-round $\AKP$ Scheme}
  \label{sec:AKP}
  In Applebaum, Kachlon and Patra \cite{AKP20}, it is shown that $4$ rounds are necessary and sufficient for securely computing any $n$-party
   degree-$2$ functionality with perfect security and optimal resilience of $t < n/3$. To design their $4$-round MPC protocol, 
   they rely on a $3$-round Type-II perfectly-secure VSS which should ensure that if $\D$ is not discarded at the end of sharing phase, then one of the following holds for every
   {\it honest} party $P_i$ at the end of Round II.
   \begin{myitemize}
   \item $P_i$ holds its {\it tentative Shamir-share} of the underlying secret;
   \item $P_i$ holds at least $t + 1$ {\it tentative shares} of its Shamir-share, which we call as {\it sub-shares}. 
   \end{myitemize}
   Moreover, it is also required that at the end of Round III, either the tentative share or the tentative sub-shares should turn out to be correct 
   (the exact case need not be known to $P_i$ at the end of Round II).
    The $\FGGRS$ {\it does not} satisfy the above requirements, as it is not a Type-II VSS. The $\KKK$ scheme
   also fails because if $P_i \not \in \VCORE$, then it obtains its sub-shares through $\Support_i$
   only at the end of Round III. Hence in
   Applebaum et al.~\cite{AKP20}, a new 3-round VSS scheme
    (see Fig \ref{fig:AKPVSS}) is presented, satisfying the above requirements. The scheme is obtained by tweaking the $\FGGRS$ scheme and 
   by borrowing the idea of
   symmetric bivariate polynomial from the $\KKK$ scheme.  
   The reconstruction phase of the scheme is same as $\BGW\mbox{-}\Rec$.
  \begin{protocolsplitbox}{$\AKP$}{The $3$-round sharing-phase protocol due to Applebaum, Kachlon and Patra \cite{AKP20}.}{fig:AKPVSS}
    \justify
    \centerline{\algoHead{Sharing Phase: Protocol $\AKP\mbox{-}{\Sh}$}}
    \begin{myitemize}
    \item {\bf Round I}:  Same as Round I of $\FGGRS\mbox{-}{\Sh}$, except that
     $\D$ uses a random degree-$(t, t)$ {\it symmetric} bivariate polynomial $F(x, y)$
     and distributes only the row-polynomial $f_i(x) = F(x, \alpha_i)$ to every $P_i$.
    \item {\bf Round II}: Same as Round II of $\FGGRS\mbox{-}{\Sh}$, except that $b_{ij} \defined f_i(\alpha_j) + r'_{ji}$.
    Moreover, every party $P_i$ sets $s_i \defined f_i(0)$ as its {\it tentative Shamir-share}
    and $\{A_j(\alpha_i) - r'_{ji} \}_{P_j \in \Partyset}$ as its {\it tentative sub-shares}.
    \item {\bf Round III}: Same as Round III of $\FGGRS\mbox{-}{\Sh}$, except that 
    for every $P_i, P_j$ where $A_i(\alpha_j) \neq b_{ji}$, party $P_i$ broadcasts 
    $(f_i(\alpha_j), r_{ij})$, party $P_j$ broadcasts $(f_j(\alpha_i), r'_{ij})$ and $\D$ broadcasts $F(\alpha_j, \alpha_i)$.
     \item {\bf Local computation at the end of Round III --- each party $P_k$ does the following}:
      \begin{myitemize}
      \item[--] Compute the sets $\Unhappy$ and $\VCORE$ as in $\FGGRS\mbox{-}{\Sh}$, based on 
      every $P_i, P_j$ for which $A_i(\alpha_j) \neq b_{ji}$.
      \item[--] Remove $P_i$ from $\WCORE_j$, if $A_j(\alpha_i) \neq b_{ij}$ during Round II and $r_{ji} \neq r'_{ji}$ during Round III.
      \item[--] Remove $P_j$ from $\VCORE$ if there exists some $i \in \{1, \ldots, n \}$, such that $P_j$ broadcasts 
      $(f_j(\alpha_i), r_{ji})$ during Round III and $A_j(\alpha_i) \neq f_j(\alpha_i) + r_j(\alpha_i)$.
      \item[--] Remove $P_j$ from $\VCORE$, if $|\VCORE \cap \WCORE_j| < n - t$. Repeat, till no more parties can be removed from $\VCORE$. 
      If $|\VCORE| < n - t$, then discard $\D$.
        \end{myitemize}
      \item {\bf Computing shares --- each party $P_i \in \Partyset$}: compute the shares as in the protocol $\KKK_{\Sh}$. 
    \end{myitemize}
   \end{protocolsplitbox}

%% file: Figures/Bivariate.tex
\begin{figure}[!h]
  \centering
  \hspace*{-0.5cm}
  \begin{tikzpicture}[scale=0.85]
    \draw (0, 0) rectangle (7, 5);
    \draw (0, 2) -- (7, 2);
    \draw (0, 3) -- (7, 3);
    \draw (2.75, 0) -- (2.75, 5);
    \draw (4.25, 0) -- (4.25, 5);
    \draw [fill=yellow] (2.75, 2) rectangle (4.25, 3);
    \draw [rounded corners, fill=ForestGreen!30] (8, 0) rectangle (9, 5);
    \draw [decorate,decoration={brace,amplitude=10pt,mirror,raise=4pt},yshift=0pt]
    (9, 0) -- (9, 5) node [black,midway,xshift=0.8cm] {\footnotesize $[s]_t$};
    \draw (3.5, -4) node {\footnotesize (a)};
    \draw (-1, 2.5) node {\footnotesize $f_j(x)$};
    \draw (-0.3, 2.5) node {\footnotesize $\Rightarrow$};
    \draw (3.5, 6) node {\footnotesize $g_i(y)$};
    \draw (3.5, 5.5) node {\footnotesize $\Downarrow$};
    \draw (8.5, -0.5) node {\footnotesize $\Downarrow$};
    \draw (8.5, -1) node {\footnotesize $q(\cdot)$};
    \draw (1, 4.5) node {\scriptsize $F(\alpha_1, \alpha_1)$};
    \draw (2.25, 4.5) node {\scriptsize $\ldots$};
    \draw (3.5, 4.5) node {\scriptsize $F(\alpha_i, \alpha_1)$};
    \draw (4.75, 4.5) node {\scriptsize $\ldots$};
    \draw (6, 4.5) node {\scriptsize $F(\alpha_n, \alpha_1)$};
    \draw (1, 3.5) node {\scriptsize $\rvdots$};
    \draw (2.25, 3.5) node {\scriptsize $\rvdots$};
    \draw (3.5, 3.5) node {\scriptsize $\rvdots$};
    \draw (4.75, 3.5) node {\scriptsize $\rvdots$};
    \draw (6, 3.5) node {\scriptsize $\rvdots$};
    \draw (1, 2.5) node {\scriptsize $F(\alpha_1, \alpha_j)$};
    \draw (2.25, 2.5) node {\scriptsize $\ldots$};
    \draw (3.5, 2.5) node {\scriptsize $F(\alpha_i, \alpha_j)$};
    \draw (4.75, 2.5) node {\scriptsize $\ldots$};
    \draw (6, 2.5) node {\scriptsize $F(\alpha_n, \alpha_j)$};
    \draw (1, 1.5) node {\scriptsize $\rvdots$};
    \draw (2.25, 1.5) node {\scriptsize $\rvdots$};
    \draw (3.5, 1.5) node {\scriptsize $\rvdots$};
    \draw (4.75, 1.5) node {\scriptsize $\rvdots$};
    \draw (6, 1.5) node {\scriptsize $\rvdots$};
    \draw (1, 0.5) node {\scriptsize $F(\alpha_1, \alpha_n)$};
    \draw (2.25, 0.5) node {\scriptsize $\ldots$};
    \draw (3.5, 0.5) node {\scriptsize $F(\alpha_i, \alpha_n)$};
    \draw (4.75, 0.5) node {\scriptsize $\ldots$};
    \draw (6, 0.5) node {\scriptsize $F(\alpha_n, \alpha_n)$};
    \draw (7.5, 4.5) node {\scriptsize $\Rightarrow$};
    \draw (7.5, 3.5) node {\scriptsize $\rvdots$};
    \draw (7.5, 2.5) node {\scriptsize $\Rightarrow$};
    \draw (7.5, 1.5) node {\scriptsize $\rvdots$};
    \draw (7.5, 0.5) node {\scriptsize $\Rightarrow$};
    \draw (8.5, 4.5) node {\scriptsize $f_1(0)$};
    \draw (8.5, 3.5) node {\scriptsize $\rvdots$};
    \draw (8.5, 2.5) node {\scriptsize $f_j(0)$};
    \draw (8.5, 1.5) node {\scriptsize $\rvdots$};
    \draw (8.5, 0.5) node {\scriptsize $f_n(0)$};
    \draw (11, 0) rectangle (18, 5);
    \draw (11, 2) -- (18, 2);
    \draw (11, 3) -- (18, 3);
    \draw (13.75, 0) -- (13.75, 5);
    \draw (15.25, 0) -- (15.25, 5);
    \draw [fill=yellow] (13.75, 2) rectangle (15.25, 3);
    \draw [rounded corners, fill=ForestGreen!30] (11, -2) rectangle (18, -1);
    \draw [decorate,decoration={brace,amplitude=10pt,mirror,raise=4pt},yshift=0pt]
    (11, -2) -- (18, -2) node [black,midway,yshift=-0.8cm] {\footnotesize $[s]_t$};
    \draw (14.5, -4) node {\footnotesize (b)};
    \draw (19, 2.5) node {\footnotesize $f_j(x)$};
    \draw (18.3, 2.5) node {\footnotesize $\Leftarrow$};
    \draw (14.5, 6) node {\footnotesize $g_i(y)$};
    \draw (14.5, 5.5) node {\footnotesize $\Downarrow$};
    \draw (18.3, -1.5) node {\footnotesize $\Leftarrow$};
    \draw (19, -1.5) node {\footnotesize $q(\cdot)$};
    \draw (12, 4.5) node {\scriptsize $F(\alpha_1, \alpha_1)$};
    \draw (13.25, 4.5) node {\scriptsize $\ldots$};
    \draw (14.5, 4.5) node {\scriptsize $F(\alpha_i, \alpha_1)$};
    \draw (15.75, 4.5) node {\scriptsize $\ldots$};
    \draw (17, 4.5) node {\scriptsize $F(\alpha_n, \alpha_1)$};
    \draw (12, 3.5) node {\scriptsize $\rvdots$};
    \draw (13.25, 3.5) node {\scriptsize $\rvdots$};
    \draw (14.5, 3.5) node {\scriptsize $\rvdots$};
    \draw (15.75, 3.5) node {\scriptsize $\rvdots$};
    \draw (17, 3.5) node {\scriptsize $\rvdots$};
    \draw (12, 2.5) node {\scriptsize $F(\alpha_1, \alpha_j)$};
    \draw (13.25, 2.5) node {\scriptsize $\ldots$};
    \draw (14.5, 2.5) node {\scriptsize $F(\alpha_i, \alpha_j)$};
    \draw (15.75, 2.5) node {\scriptsize $\ldots$};
    \draw (17, 2.5) node {\scriptsize $F(\alpha_n, \alpha_j)$};
    \draw (12, 1.5) node {\scriptsize $\rvdots$};
    \draw (13.25, 1.5) node {\scriptsize $\rvdots$};
    \draw (14.5, 1.5) node {\scriptsize $\rvdots$};
    \draw (15.75, 1.5) node {\scriptsize $\rvdots$};
    \draw (17, 1.5) node {\scriptsize $\rvdots$};
    \draw (12, 0.5) node {\scriptsize $F(\alpha_1, \alpha_n)$};
    \draw (13.25, 0.5) node {\scriptsize $\ldots$};
    \draw (14.5, 0.5) node {\scriptsize $F(\alpha_i, \alpha_n)$};
    \draw (15.75, 0.5) node {\scriptsize $\ldots$};
    \draw (17, 0.5) node {\scriptsize $F(\alpha_n, \alpha_n)$};
    \draw (12, -0.5) node {\scriptsize $\Downarrow$};
    \draw (13.25, -0.5) node {\scriptsize $\Downarrow$};
    \draw (14.5, -0.5) node {\scriptsize $\Downarrow$};
    \draw (15.75, -0.5) node {\scriptsize $\Downarrow$};
    \draw (17, -0.5) node {\scriptsize $\Downarrow$};
    \draw (12, -1.5) node {\scriptsize $g_1(0)$};
    \draw (13.25, -1.5) node {\scriptsize $\ldots$};
    \draw (14.5, -1.5) node {\scriptsize $g_i(0)$};
    \draw (15.75, -1.5) node {\scriptsize $\ldots$};
    \draw (17, -1.5) node {\scriptsize $g_n(0)$};
  \end{tikzpicture}
  \caption{Pictorial depiction of the values on the degree-$(t, t)$
    polynomial $F(x, y)$ distributed by $\D$ and how they constitute $[s]_t$.
    The value highlighted in the yellow color denotes a common value held by every pair of parties
    $(P_i, P_j)$. In the first approach, the Shamir-sharing polynomial $q(\cdot)$ is set as $F(0, y)$ and
    the Shamir-shares are the constant terms of the individual row-polynomials. In the
    second approach, $q(\cdot)$ is set as $F(x, 0)$ and the
    Shamir-shares are the constant terms of the individual column-polynomials.}
\label{fig:Bivariate}
\end{figure}
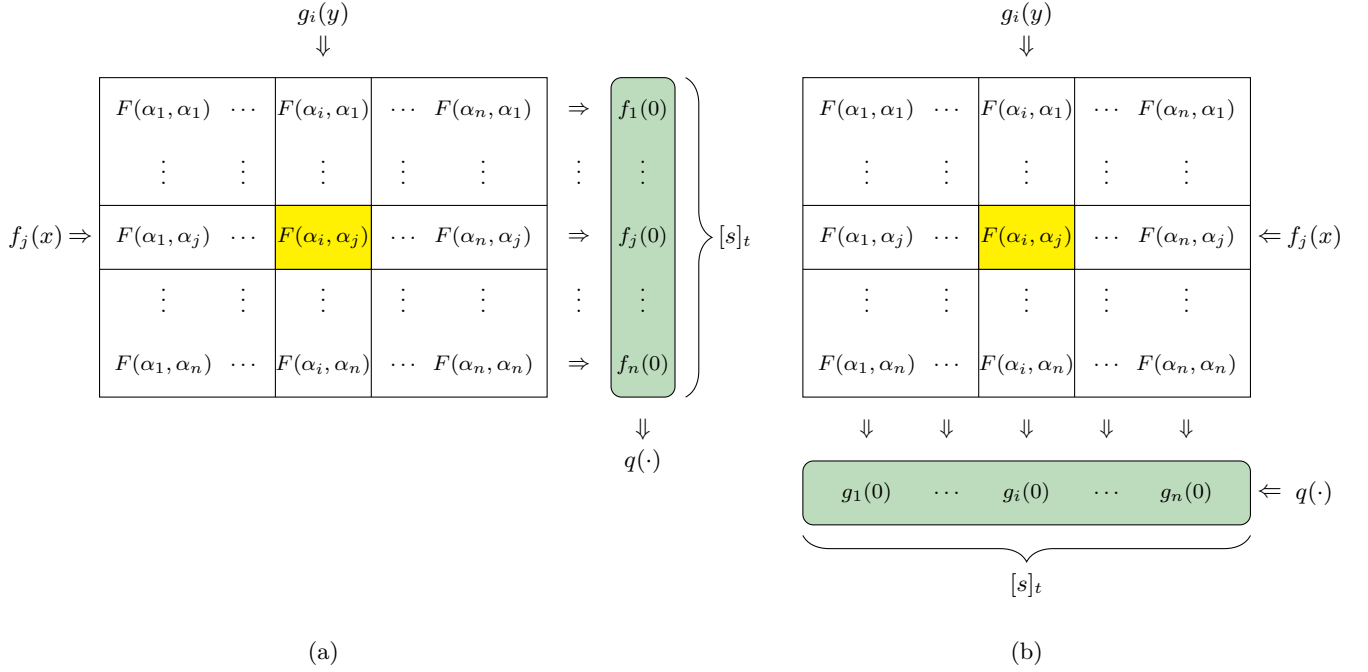

%% file: SThreshold4t.tex
We present a perfectly-secure VSS scheme with $n > 4t$ and a $2$-round sharing phase due to Genarro et al.~\cite{GIKR01}.
   We first present a data structure called 
  $(n, t)$-$\STAR$ (which we often call as just $\STAR$).
  \begin{definition}[$(n, t)$-$\STAR$ \cite{BCG93}]
  \label{def:STAR}
  Let $G$ be an undirected graph over $\Partyset$. Then a pair of subsets of nodes 
   $(\CSet, \DSet)$ where $\CSet \subseteq \DSet \subseteq \Partyset$ is called an $(n, t)$-$\STAR$, if all the following hold.
   \begin{myitemize}
   \item[--] $|\CSet| \geq n - 2t$ and $|\DSet| \geq n - t$.
   \item[--] For every $P_i \in \CSet$ and every $P_j \in \DSet$, the edge $(P_i, P_j)$ is present in $G$.
   \end{myitemize}
  \end{definition} 
 Ben-Or et al.~\cite{BCG93} presents an efficient algorithm for checking the presence of a $\STAR$.
   Whenever the input graph contains a clique of size at least $n - t$, then the algorithm outputs a $\STAR$. 
   
 Protocol $\GIKRIII\mbox{-}{\Sh}$ is based on a simplification of the {\it asynchronous} VSS scheme of Ben-Or et al.~\cite{BCG93}, adapted 
  to the synchronous setting 
 (see Section \ref{sec:BCG}). 
 As done in the earlier protocols, $\D$ distributes row and column-polynomials to the respective parties. The goal is then to verify if the row and
 column-polynomials of all (honest) parties are derived from a single degree-$(t, t)$ bivariate polynomial. However, since $n > 4t$ (compared to $n > 3t$ in the earlier protocols),
 the above verification task is significantly simplified. For simplicity, 
  we first explain an {\it inefficient} verification, followed by the actual {\it efficient} method
  used in the protocol.
  
  Once the parties receive their respective polynomials, they perform the pair-wise consistency checks (based on the idea of using random pads).
  The parties next construct a {\it consistency-graph} $G$ over $\Partyset$,
  where there exists an edge between a pair of parties if no dispute is reported between them. 
  The parties next check for the presence of a {\it clique} of size $n - t$ in $G$, which is bound to exist if $\D$ is {\it honest}.
   If no clique is obtained then clearly $\D$ is {\it corrupt} and hence discarded.
   If a clique $\CSet$ of size $n - t$ is obtained, then the polynomials of the {\it honest} parties in $\CSet$
    are pair-wise consistent and lie on a single degree-$(t, t)$ bivariate polynomial, say $\starF(x, y)$.
    However, there could be up to $t$ (honest) parties {\it outside} $\CSet$
    and the goal is to let  each such ``outsider" $P_i \not \in \CSet$ obtain its degree-$t$ row-polynomial $\starF(x, \alpha_i)$.
   The crucial observation here is that since $n > 4t$, achieving this goal {\it does not} require any additional interaction.
    That is, each $P_i \not \in \CSet$ considers the set of $n - t \geq 3t + 1$ $g_{ji}$ values, received from the parties $P_j \in \CSet$ as part of the pair-wise
   consistency test.  Among these $g_{ji}$ values, at least $2t + 1$ are sent by the {\it honest} parties $P_j \in \CSet$, which
   uniquely define $\starF(x, \alpha_i)$. 
   Since $\starF(x, \alpha_i)$ is a degree-$t$ polynomial and there can be at most $t$ {\it corrupt} parties $P_j \in \CSet$ who may provide incorrect values of $g_{ji}$,
    party $P_i$ can error-correct these values and obtain $\starF(x, \alpha_i)$.
   
   The above method is {\it inefficient}, as finding a maximum-sized clique is an NP-complete problem. 
   Instead, the parties check for the presence of a $\STAR$. If $\D$
   is {\it honest} then the set of honest parties constitute a potential clique of size $n - t$ in $G$ and so a $\STAR$ is always obtained.
   If a $\STAR$ $(\CSet, \DSet)$ is obtained, then there are at least $|\CSet| - t = t + 1$ {\it honest} parties in $\CSet$ holding degree-$t$ row-polynomials 
   and at least
   $|\DSet| - t = 2t + 1$ {\it honest} parties in $\DSet$ holding degree-$t$ column-polynomials, which are pair-wise consistent and hence
    lie on a single degree-$(t, t)$ bivariate polynomial $\starF(x, y)$. 
   Any $P_i \not \in \CSet$ obtains its corresponding row-polynomial $\starF(x, \alpha_i)$ by
   applying the error-correction procedure as discussed above on the $g_{ji}$ values received from $P_j \in \DSet$.
      Protocol $\GIKRIII\mbox{-}{\Sh}$ is presented in Fig \ref{fig:GIKR}; the reconstruction protocol $\GIKRIII\mbox{-}{\Rec}$ is the same as
      $\BGW\mbox{-}\Rec$.
      \begin{protocolsplitbox}{$\GIKRIII\mbox{-}{\Sh}$}{The $2$-round $\GIKRIII\mbox{-}{\Sh}$ protocol with $n > 4t$ due to Genarro, Ishai, Kushilevitz and Rabin \cite{GIKR01}.}{fig:GIKR}
\justify
\begin{myitemize}
\item {\bf Round I}: $\D$ picks $F(x, y)$ and distributes $f_i(x), g_i(y)$ lying on $F(x, y)$ to each $P_i$. 
In parallel, each $P_i$ sends a random mask $r_{ij}$ to each $P_j$.
\item {\bf Round II}: Each $P_i$ broadcasts $a_{ij} = f_i(\alpha_j) + r_{ij}$ and
        $b_{ij} = g_i(\alpha_j) + r'_{ji}$, where $r'_{ji}$ is the pad, received from $P_j$.
\item {\bf Local computation at the end of round II --- each party $P_i$ does the following}:
   \begin{myitemize}
       \item[--] Construct an undirected graph $G$ over $\Partyset$, where the edge $(P_l, P_m)$ is present 
          if $a_{lm} = b_{ml}$ and $a_{ml} = b_{lm}$ holds.  Run the $\STAR$-finding algorithm over $G$.   
        \item[--] If no $\STAR$ is obtained, then discard $\D$. 
         Else let $(\CSet, \DSet)$ be the $\STAR$ obtained in $G$.
                 \begin{myitemize}
                 \item[--] If $P_i \in \CSet$, then output the {\it share} $f_i(0)$.
                 \item[--] Else recompute the row-polynomial $f_i(x)$ as follows and output the {\it share} $f_i(0)$.
                   \begin{myitemize}
                   \item[--] $\forall P_j \in \DSet$, compute $g_{ji} =  b_{ji} - r_{ij}$. Execute $\RSDec(t, t, S_i)$ to get $f_i(x)$, where 
                   $S_i \defined \{g_{ji} \}_{P_j \in \DSet}$. 
                   \end{myitemize}
                 
                 \end{myitemize}
    
    \end{myitemize}
\end{myitemize}
\end{protocolsplitbox}
%

%% file: SThreshold1.tex
Genarro et al.~\cite{GIKR01} presented a 
 VSS scheme (Fig \ref{fig:OneVSS}) with $1$-round sharing-phase, where $n = 5$ and $t = 1$. 
  For simplicity, let $P_1$ be the dealer.
  During the sharing phase, $P_1$ distributes Shamir-shares of its secret.
  Since no additional rounds are available, the parties cannot verify whether $\D$ has distributed consistent shares.
   In the reconstruction phase, the dealer is {\it not allowed} to participate.
   The remaining parties exchange their respective shares and try to
   error-correct one potential incorrect share. If the error-correction is successful then the parties output the constant term of the reconstructed degree-$1$ polynomial, else they output
   $\bot$. Since there can be one corrupt party, there are two possible cases.
    If $P_1$ is {\it honest}, then {\it privacy} is ensured and the shares of $P_2, P_3, P_4$ and $P_5$
   lie on a degree-$1$ polynomial. Hence during the reconstruction phase, even if a potentially corrupt party provides incorrect share,
   it can be error-corrected, thus guaranteeing the {\it correctness} property. 
   
   The case when $P_1$ is {\it corrupt} can be divided into {\it three} sub-cases.
   If $P_1$ has distributed valid Shamir-shares (all lying on degree-$1$ polynomial), then the underlying Shamir-shared value
   will be reconstructed correctly. The second sub-case is when $P_1$ has distributed valid Shamir-shares to {\it exactly three} parties,
   say $P_2, P_3$ and $P_4$, and let their shares  lie on a degree-$1$ polynomial $\starq(\cdot)$. 
   Hence during the reconstruction phase, all honest parties reconstruct $s^{\star} = \starq(0)$ by error-correcting the share provided by 
    $P_5$.
   The remaining sub-case is when the shares of no three parties among
   $\{P_2, P_3, P_4, P_5 \}$ lie on a degree-$1$ polynomial. In this case we define $s^{\star} \defined \bot$, where $\bot \not \in \F$, indicating
    that the sharing dealt by $P_1$ is ``invalid".
    During the reconstruction phase, the error-correction will fail (since the shares of no three parties lie on a degree-$1$ polynomial)
   and hence the honest parties output $\bot$. Thus in all the $3$ sub-cases, {\it strong commitment} property is achieved.
   %
   \begin{schemesplitbox}{$\OneVSS$}{The one round perfectly-secure VSS scheme of Genarro, Ishai, Kushilevitz and Rabin \cite{GIKR01}.}{fig:OneVSS}
	\justify
\centerline{\algoHead{Sharing Phase: Protocol $\OneVSS\mbox{-}{\Sh}$}}
The dealer $P_1$ on having input $s \in \F$, picks a random degree-$1$ polynomial $q(\cdot)$ over $\F$ such that $q(0) = s$. 
    For $i = 2, \ldots, 5$, it sends the {\it share} $s_i = q(\alpha_i)$ to $P_i$. \\[.1cm]
\justify
\centerline{\algoHead{Reconstruction Phase: Protocol $\OneVSS\mbox{-}{\Rec}$}}
Each party $P_i \in \{P_2, P_3, P_4, P_5 \}$ does the following:
 \begin{myitemize}
 \item Send  $s_i$ to every $P_j \in \Partyset \setminus \{P_1 \}$.
   On receiving $s_j$ from $P_j$, include $s_j$ in a list $W_i$. Execute $\RSDec(1, 1, W_i)$.
    \begin{myitemize}
    \item[--] If $\RSDec$ outputs a degree-$1$ polynomial $q(\cdot)$, then output $q(0)$.
    Else output $\bot$.
    \end{myitemize} 
 \end{myitemize}
\end{schemesplitbox}

%% file: APrelim.tex
\section{Preliminaries and Definitions}
\label{sec:APrelim}
In the {\it synchronous} communication setting, 
  each party
 knows in advance how long it has to wait for an expected message, and if the message does not arrive within that time-bound, then
  the sender party is {\it corrupt}.
  Unfortunately, it is impossible to ensure such strict time-outs in real-world networks like the Internet.
  Motivated by this, \cite{BCG93,CanettiThesis} introduced the {\it asynchronous} communication model with {\it eventual message delivery}.
   Apart from a better modelling of real-world networks,  asynchronous protocols have the advantage of running at
   the {\it actual}  speed of  the  underlying network. More specifically, for a synchronous protocol, the participants have to
   pessimistically set the global delay $\Delta$ to a large value to ensure that the messages sent by every party at the beginning of
   a round reach their destination within time $\Delta$. But if the {\it actual} delay $\delta$ is such that
   $\delta << \Delta$, then the protocol fails to take advantage of the faster network and its running time will be proportional to
   $\Delta$.

    In the asynchronous model,
    the messages  can be {\it arbitrarily, but finitely} delayed.
   The only guarantee is that any sent message is {\it eventually} delivered, but probably in a different order.
    The sequence of message delivery 
   is controlled by a {\it scheduler} and to model the worst case scenario, we assume that the scheduler is under the control of $\Adv$.
    Due to the lack of any upper bound on the message delays, no party can wait to receive
  communication from {\it all} its neighbours to avoid an endless wait (as a corrupt neighbour may not send any message).
  As a result,  a party
  can afford to wait for messages from {\it at most} $n - t$ parties (including itself), thus ignoring communication from
  $t$ potentially {\it honest} neighbours. Consequently, all synchronous VSS protocols become insecure when executed
  in an asynchronous environment, as they depend upon the fact that the messages of {\it all} the honest parties
  are considered. Due to the absence of any global clock, the execution of any asynchronous protocol is event-based, rather than round-based.
  
  Informally, an AVSS scheme consists of an {\it asynchronous} sharing and an {\it asynchronous} reconstruction
  phase, providing privacy, correctness and strong commitment guarantees. However, 
  we need these properties to hold {\it eventually}. Additionally, we need termination guarantees. 
  Namely, if $\D$ is {\it honest}, then we require that these phases eventually terminate. 
  On the other hand, if $\D$ is {\it corrupt}, then the termination demands are ``weaker". Namely, we require
  the honest parties to terminate the sharing and reconstruction phase {\it only if} some honest party has terminated the sharing phase. This models the fact
  that a potentially corrupt $\D$ may not invoke the sharing phase in the first 
  place (this is unlike synchronous VSS, where protocols always terminate after  a ``finite" number of communication rounds).
  \begin{definition}[\bf Perfectly-Secure AVSS \cite{CanettiThesis}]
  \label{def:AVSS}
  Let $(\Sh, \Rec)$ be a pair of asynchronous protocols for the $n$ parties, where a designated {\it dealer} $\D \in \Partyset$ has a private input $s$ for $\Sh$ and where each
   (honest) party who completes $\Sh$, subsequently invokes $\Rec$, with its local output of $\Sh$. Then $(\Sh, \Rec)$ constitute a 
   {\it perfectly-secure 
   AVSS scheme}, if all the following holds for every possible $\Adv$.
   \begin{myitemize}
      \item {\bf Termination}:
           \begin{myitemize}
           \item[--] If $\D$ is {\it honest}, then every honest party will eventually complete $\Sh$.
           \item[--] If some honest party has completed $\Sh$, then all honest parties will eventually complete $\Sh$.
           \item[--] If some honest party has completed $\Sh$, then it will eventually complete $\Rec$.
           \end{myitemize}
    \item {\bf Privacy and Correctness}: Same as for synchronous VSS.
    \item {\bf Strong Commitment}: If $\D$ is {\it corrupt} and some honest party completes $\Sh$, 
    then the joint view of the honest parties at the end of $\Sh$ defines a value $s^{\star}$ (possibly different from $s$), such that all honest parties eventually output
    $s^{\star}$ at the end of $\Rec$.
   \end{myitemize}
  \end{definition}    
We can have Type-I and Type-II AVSS schemes. All the existing perfectly-secure AVSS schemes are of Type-II, where
 the secret is {\it always} shared as per the Shamir's secret-sharing scheme. 
 \subsection{Asynchronous Tools}
  \paragraph{\bf Asynchronous Reliable-Broadcast (ACast)}
   An ACast  protocol allows a designated {\it sender}
  $\Sender \in \Partyset$ to identically send a message $m$ to all the parties.
   If
   $\Sender$ is {\it honest}, then all honest parties eventually terminate with output $m$.
   While, if $\Sender$ is {\it corrupt} {\it and} some honest party terminates with output $m^{\star}$, then  eventually every
  other honest party should terminate with output $m^{\star}$.
  Hence the termination guarantees are ``weaker" than {\it synchronous} reliable-broadcast (RB), where
  the protocol {\it always} terminates, irrespective of $\Sender$. 
  Bracha \cite{Bra84} presented a very elegant instantiation of ACast for any $n > 3t$. 
    We use the term $P_i$ {\it broadcasts} $m$ to mean that $P_i$ acts as $\Sender$ and invokes an instance of
    ACast protocol to broadcast $m$. Similarly, the term $P_j$ {\it receives $m$ from the broadcast of} $P_i$ means that $P_j$ completes 
    the instance of ACast protocol where $P_i$ is \Sender, with
 $m$ as output. 
  \paragraph{\bf Online Error-Correction (OEC)}
      Let $s$ be 
   $d$-shared among $\Partyset' \subseteq \Partyset$ such that $d < (|\Partyset'| - 2t)$ holds.
   That is, there exists some degree-$d$ polynomial $q(\cdot)$ with $q(0) = s$ and each $P_i \in \Partyset'$ has a share $q(\alpha_i)$.
       The goal is 
    to make some {\it designated} party, say $P_R$, reconstruct $s$
     (actually OEC allows $P_R$ to  reconstruct $q(\cdot)$). 
   In the {\it synchronous} setting, this is achieved by letting {\it every} party in $\Partyset'$ to send its
 share to $P_R$, who can apply the algorithm $\RSDec$ and error-correct up to $t$ potentially incorrect shares.
  Given that $d < (|\Partyset'| - 2t)$, 
  the reconstruction will be robust. In the {\it asynchronous} setting, achieving the same goal requires a bit of trick. 
    The intuition behind OEC is that $P_R$  keeps waiting till it receives
 $d + t + 1$ shares, all of which lie on a {\it unique} degree-$d$ polynomial, which eventually happens for $P_R$ (even if the {\it corrupt} parties
  in $\Partyset'$ do not send their shares to $P_R$)
  as there are at least $|\Partyset'| - t \geq d + t + 1$ honest parties in $\Partyset'$. 
 This step requires $P_R$ to repeatedly apply $\RSDec$ and try recover $q(\cdot)$, upon asynchronously
  receiving every new share from the parties in $\Partyset'$. Once $P_R$ receives $d + 1 + 1$ shares lying on a degree-$d$ polynomial, say
  $q'(\cdot)$, then $q'(\cdot) = q(\cdot)$. This is because among the $d + t + 1$ shares, at least $d + 1$ are from the {\it honest}
  parties, which uniquely determine $q(\cdot)$.
   We denote the OEC procedure by 
 $\OEC(\Partyset', d)$.

%% file: AThreshold.tex
\section{Perfectly-Secure AVSS Schemes}
\label{sec:AThreshold}
We discuss the various perfectly-secure AVSS schemes \cite{BCG93,PCR15,CHP13}. 
  While the sharing phase requires $n > 4t$, with the exception of \cite{PCR15} the reconstruction phase requires $n > 3t$.
  The necessity of the condition $n > 4t$ follows from \cite{BKR94,ADS20}, where it is shown that in any AVSS scheme designed
  with $n \leq 4t$, there is a {\it non-zero} probability in the termination property. We summarize the AVSS schemes in 
   Table \ref{tab:AVSSSummary}.
    While the degree of sharing $d$ is $t$ for 
    \cite{BCG93,CHP13}, the degree $d$ could be more than $t$ for \cite{PCR15}.

  \begin{table}[!htbp]
    \centering
    \resizebox{\textwidth}{!}{  
    \begin{tabular}{|c|c|c|c|c|c|c|c|}
        \hline
   & \multicolumn{5}{c|}{Sharing Phase} & \multicolumn{2}{c|}{Reconstruction Phase} \\ \hline
 Scheme & $n$ & $d$ & $\param$ & $|\F|$ & Communication Complexity (CC) & $n$ & CC         
  \\ \hline
  $\BCG$ \cite{BCG93} & $n > 4t$ & $t$ & $1$ & $|\F| > n$ & $\Order(n^2 \log{|\F|} + \BC(n^2 \log n))$ & $n > 3t$ & $\Order(n^2 \log{|\F|})$ \\ \hline
  $\PCR$ \cite{PCR15} & $n > 4t$ & $t < d < n - 2t$ & $1$ & $|\F| > n$ & $\Order(n^2 \log{|\F|} + \BC(n^2 \log n))$ & $n > 4t$ & $\Order(n^2 \log{|\F|})$ \\ \hline
  $\CHP$ \cite{CHP13} & $n > 4t$ & $t$ & $\geq n - 3t$ & $|\F| > 2n - 3t$ & $\Order(\param \cdot n^2 \log{|\F|} + \BC(n^2 \log n))$ & $n > 3t$ & 
   $\Order(\param \cdot n^2 \log{|\F|})$ \\ \hline
    \end{tabular}
    }
      \caption{\label{tab:AVSSSummary}Summary of the perfectly-secure AVSS schemes. Here $\param$ denotes the number of values shared
      through a single AVSS instance
      and $\BC$ denotes the communication happening through ACast.}
\end{table}

  \subsection{$\BCG$ Scheme}
  \label{sec:BCG}
  The $\BCG$ scheme is presented in Fig \ref{fig:BCG}. The sharing phase protocol $\BCG\mbox{-}{\Sh}$ is 
  a slightly modified and simplified version of the original protocol \cite{BCG93}, based on the simplifications suggested in \cite{BH07,PCR15}.   
   Protocol $\BCG\mbox{-}{\Sh}$ is similar to $\GIKRIII\mbox{-}{\Sh}$ (see Fig \ref{fig:GIKR}), executed in the
  asynchronous setting. The protocol has four stages, each of which is executed asynchronously. 
  
  During the first stage, $\D$ distributes the row and column-polynomials to the respective parties. During the second stage, 
   the parties perform perform pair-wise consistency checks
  and publicly announce the results, based on which parties build a consistency-graph. 
  Since the results of the consistency checks are broadcasted {\it asynchronously}, the consistency-graph  might be different for different parties (however, the edges which
  are present in the graph of one honest party will be {\it eventually} included in the graph of every other honest party). 
  During the third stage, $\D$ checks for a $\STAR$ $(\CSet, \DSet)$ in its consistency-graph,
  which it then broadcasts as a ``proof" that
  the row-polynomials of the (honest) parties in $\CSet$ and the column-polynomials of the (honest) parties in $\DSet$ lie a single degree-$(t, t)$ bivariate polynomial, 
  which is considered as $\D$'s ``committed" bivariate polynomial.
  A party upon receiving $(\CSet, \DSet)$ from $\D$ {\it accepts} it, when $(\CSet, \DSet)$ constitutes a $\STAR$ in its own consistency-graph. 
   Fo an {\it honest} $\D$, the set of honest parties eventually constitute a clique and hence
   $\D$ eventually finds a $\STAR$, which will be eventually accepted by every honest party. 
  
  Once a $\STAR$ is accepted by $P_i$
   then in the last stage, its goal
  is to compute its share, for which $P_i$ should hold its degree-$t$ row-polynomial, lying on
  $\D$'s committed bivariate polynomial. 
  If $P_i \in \CSet$, then it already has this polynomial.
  Else, $P_i$ waits for the common values on the required row-polynomial from the parties in $\DSet$ and error-corrects the incorrectly received values using
  OEC.  Since $P_i$'s desired row-polynomial has degree-$t$ and since each $P_j$ in $\DSet$ holds a share of this polynomial in the form of a
  common value on its column-polynomial, OEC eventually outputs the desired row-polynomial for $P_i$, as $|\DSet| \geq 3t + 1$ and $\DSet$ contains at most
  $t$ corrupt parties\footnote{For this step, it is necessary that $n > 4t$. Else $|\DSet| \leq 3t$ and
  OEC will fail to let $P_i$ obtain its desired row-polynomial.}.
  
  During reconstruction phase, every party sends its share to every other party. 
   The parties then reconstruct the secret by using OEC on the received shares (this step will work even if $n > 3t$).
  
  \begin{schemesplitbox}{$\BCG$}{The perfectly-secure AVSS scheme of Ben-Or, Canetti and Goldreich \cite{BCG93}.}{fig:BCG}
	\justify
\centerline{\algoHead{Sharing Phase: Protocol $\BCG\mbox{-}{\Sh}$}}
\begin{myitemize}
\item {\bf Stage I : Distributing Polynomials --- the dealer $\D$ does the following}
      \begin{myitemize}
	    \item[--] On having the input $s \in \F$, pick a random degree-$t$ {\it Shamir-sharing polynomial} $q(\cdot)$, such that $q(0) = s$ holds.
	    Then pick a random degree-$(t, t)$ bivariate polynomial $F(x, y)$, such that $F(0, y) = q(\cdot)$ holds.
	     \item[--] For $i = 1, \ldots, n$, send the polynomials $f_i(x) = F(x, \alpha_i)$ and $g_i(y) = F(\alpha_i, y)$.
	    to $P_i$. 
	  \end{myitemize}
\item {\bf Stage II : Pair-wise consistency checks and building consistency-graph --- each party $P_i$}
   \begin{myitemize}
      \item[--] Upon receiving $f_i(x), g_i(y)$ from $\D$, send 
      $f_{ij} = f_i(\alpha_j)$ and $g_{ij} = g_i(\alpha_j)$ to $P_j$, for $j = 1, \ldots, n$.
       \item[--] Upon receiving $f_{ji}, g_{ji}$ from $P_j$, broadcast $(\OK, i, j)$ if $f_{ji} = g_i(\alpha_j)$ and $g_{ji} = f_i(\alpha_j)$ hold.
      \item[--] Construct a graph $G_i$ over $\Partyset$. Add the edge $(P_j, P_k)$ in $G_i$, if $(\OK, j, k)$ and $(\OK, k, j)$ are received from
        the broadcast of $P_j$ and $P_k$ respectively. Keep updating $G_i$, upon receiving new $(\OK, \star, \star)$ messages.
   \end{myitemize}
 \item {\bf Stage III : Finding $\STAR$ in the consistency-graph --- the dealer $\D$ does the following}
    \begin{myitemize}
    \item[--] Let $G_D$ be the consistency-graph built by $\D$. After every update in $G_D$, run the star-finding algorithm to check for the presence of a $\STAR$
     in $G_D$.
     If a $\STAR$ $(\CSet, \DSet)$ is found in $G_D$, then broadcast $(\CSet, \DSet)$.
    \end{myitemize}
\item {\bf Stage IV : share computation --- each party $P_i$ does the following}
    \begin{myitemize}
    \item[--] If $(\CSet, \DSet)$ is received from the broadcast of $\D$, then {\it accept} it if $(\CSet, \DSet)$ is a $\STAR$ in the graph $G_i$.
    \item[--] If $(\CSet, \DSet)$ is accepted, then compute the {\it share} $s_i$ as follows and terminate.
           \begin{myitemize}
           \item[--] If $P_i \in \CSet$, then set $s_i = f_i(0)$, where $f_i(x)$ is the degree-$t$ row-polynomial received from $\D$.
           \item[--] Else initialize $W_i$ to $\emptyset$. Upon receiving $g_{ji}$ from $P_j \in \DSet$, include $g_{ji}$ to $W_i$.
            Keep updating $W_i$ and keep executing $\OEC(W_i, t)$ till a degree-$t$ polynomial $f_i(x)$ is obtained. 
           Then set $s_i = f_i(0)$.  \\[.2cm]
           \end{myitemize}    
    \end{myitemize}
\end{myitemize}
\justify
\centerline{\algoHead{Reconstruction Phase: Protocol $\BCG\mbox{-}{\Rec}$}}
Each party $P_i \in \Partyset$ does the following.
  \begin{myitemize}
  \item[--] Send the share $s_i$ to every party $P_j \in \Partyset$.
  \item[--] Initialize a set $R_i$ to $\emptyset$. Upon receiving $s_j$ from $P_j$, include $s_j$ to $R_i$.
            Keep updating $R_i$ and executing $\OEC(R_i, t)$ till a degree-$t$ polynomial $q(\cdot)$ is obtained. 
            Then output $s = q(0)$ and terminate.            
  \end{myitemize}
\end{schemesplitbox}
For the sake of completeness, we prove the properties of the scheme $\BCG$, as stated in Theorem \ref{thm:BCG}. The proof for the follow-up AVSS 
 schemes also use similar arguments.
\begin{theorem}
\label{thm:BCG}
 $(\BCG\mbox{-}{\Sh}, \BCG\mbox{-}{\Rec})$ constitute a Type-II perfectly-secure AVSS scheme with respect to 
 Shamir's $t\mbox{-out-of-}n$ secret-sharing scheme. 
\end{theorem}
\begin{proof}
Let us first consider an {\it honest} $\D$. The {\it privacy} simply follows from the fact that during $\BCG\mbox{-}{\Sh}$, the adversary learns at most $t$ rows and column-polynomials, lying on
 $F(x, y)$. If $\D$ is {\it honest}, then every pair of {\it honest} parties $P_i, P_j$ {\it eventually} receive their respective polynomials and exchange the common points on their polynomials.
  Since the pair-wise consistency test will pass, $P_i$ and $P_j$ {\it eventually} broadcasts $(\OK, i, j)$ and $(\OK, j, i)$ messages respectively and from the properties of ACast, these messages are 
   {\it eventually} delivered
  to every honest party. Since there are at least $n - t$ honest parties, the honest parties will {\it eventually} constitute a clique in every honest party's consistency graph.
  Consequently, $\D$ {\it eventually} finds a $\STAR$ $(\CSet, \DSet)$ in its consistency graph and broadcasts it, which is {\it eventually} delivered to every honest party.
  Moreover, $(\CSet, \DSet)$ is {\it eventually} accepted by every honest party, as $(\CSet, \DSet)$ {\it eventually} forms a $\STAR$ in every honest party's consistency graph.
  Now consider an arbitrary {\it honest} $P_i$. If $P_i \in \CSet$, then it already has the polynomial $f_i(x)$ and hence it outputs $f_i(0)$ as its share and terminates $\BCG\mbox{-}{\Sh}$.
  On the other hand, even if $P_i \not \in \CSet$, it {\it eventually} receives the common point on  its row polynomial from the honest parties in $\DSet$ as part of the pair-wise
  consistency checks and hence by applying OEC, it {\it eventually}
  computes $f_i(0)$ as its share and terminates $\BCG\mbox{-}{\Sh}$. During $\BCG\mbox{-}{\Rec}$, the share of every honest party is {\it eventually} delivered to every honest party.
  Consequently, by applying OEC, every honest party {\it eventually} reconstructs $s$. This proves the {\it correctness} property.
  
  Next consider a {\it corrupt} $\D$ and let $P_h$ be the {\it first honest} party to terminate $\Sh$. This implies that $P_h$ has received $(\CSet, \DSet)$ from the broadcast of
  $\D$, which constitutes a $\STAR$ in $P_h$'s consistency graph. From the properties of ACast, every honest party will eventually receive and accept $(\CSet, \DSet)$. 
  This implies that the row-polynomials of the {\it honest} parties in $\CSet$ and the column-polynomials of the {\it honest} parties in $\DSet$ are pair-wise consistent
  and lie on a single degree-$(t, t)$ bivariate polynomial, say $F^{\star}(x, y)$. Let $q^{\star}(\cdot) = F^{\star}(0, y)$ and  $s^{\star} = q^{\star}(0)$. Consider an 
  arbitrary {\it honest} $P_i$. If $P_i \in \CSet$, then it already has the row-polynomial $F^{\star}(x, \alpha_i)$ and hence it outputs $F^{\star}(0, \alpha_i) = q^{\star}(\alpha_i)$ as its share
  of $s^{\star}$.  On the other hand, even if $P_i \not \in \CSet$, it {\it eventually} receives the common point on  its row polynomial from the honest parties in $\DSet$ as part of the pair-wise
  consistency checks and hence by applying OEC, it {\it eventually}
  computes $F^{\star}(0, \alpha_i)$ as its share. Moreover, it is easy to see that during $\BCG\mbox{-}{\Rec}$, every honest party eventually outputs $s^{\star}$.
  This proves the {\it strong commitment} property.
\end{proof}

  \subsection{$\PCR$ Scheme}
   Patra, Choudhury and Rangan \cite{PCR15} observed that
  the $\BCG$ scheme can be modified in a non-trivial way to generate a $d$-sharing of 
  $\D$'s input in a verifiable fashion, for {\it any} given $d$ in the range $t \leq d < n - 2t$. The resultant scheme is presented
   in Fig \ref{fig:PCR}. The main motivation for an AVSS with
  the degree of sharing {\it greater} than $t$ is to get efficient MPC protocols (see \cite{DN07,BH07,BH08,PCR15}).
  
   To $d$-share $s$, $\D$ picks a random degree-$d$ {\it Shamir-sharing polynomial}
  $q(\cdot)$ with $q(0) = s$ and embeds it in a random degree-$(d, t)$ bivariate polynomial $F(x, y)$ at $y = 0$.
  Thus the row and column-polynomials have {\it different} degrees.
    This is in {\it contrast} to the earlier schemes  where $q(\cdot)$ has degree-$t$  and
   where $q(\cdot)$ is embedded at $x = 0$ (instead of $y = 0$) in a random 
   degree-$(t, t)$ bivariate polynomial. $\D$ then distributes the row and column-polynomials and
    the parties publicly announce the results of
        pair-wise consistency checks and build
    consistency-graphs.
    $\D$ then proves that it has distributed consistent polynomials to ``sufficiently many" parties,
   derived from a single degree-$(d, t)$ bivariate polynomial, say $\starF(x, y)$ (where $\starF(x, y) = F(x, y)$ for an {\it honest} $\D$).
   This stage is {\it different} from $\BCG\mbox{-}{\Sh}$ and constitutes the core of $\PCR\mbox{-}{\Sh}$. 
   In a more detail, $\D$ publicly proves that it has delivered degree-$d$ row-polynomials lying on $\starF(x, y)$, to at least $n - t = 3t + 1$ 
   parties $\ESet$ and degree-$t$ column-polynomials lying on $\starF(x, y)$, to at least $n - t = 3t + 1$ 
   parties $\FSet$ (the sets $\ESet$ and $\FSet$ need not be the same). 
     Once the existence of $(\ESet, \FSet)$  is confirmed (we will discuss in the sequel how such
   sets are identified), 
   $\D$'s sharing is completed by ensuring that every $P_i$ gets its degree-$d$ column-polynomial
   $g_i(y)$. Party $P_i$ can then output $g_i(0)$ (which is the same as $\starF(\alpha_i, 0)$) as its share
   and the value $s^{\star} = \starF(0, 0)$ will be $d$-shared through $\starF(x, 0)$.
   If $P_i \in \FSet$ then it will already have its $g_i(y)$ polynomial.
   For $P_i \not \in \FSet$, we observe that {\it every} $P_j \in \ESet$ possesses a share on $g_i(y)$ in the form of
   a point on $P_j$'s row-polynomial and which $P_j$ would have sent to $P_i$ as part of pair-wise consistency test. 
   Since there are at least $3t + 1$ such parties $P_j$ in $\ESet$ and since $g_i(y)$ has degree-$t$, 
   party $P_i$ can reconstruct $g_i(y)$ by using the OEC mechanism.
   
      If $\D$ is {\it honest}, then $\Adv$ learns at most
   $t$ rows and column-polynomials lying on $F(x, y)$.
    Now similar to Lemma \ref{lemma:bivariateprivacy}, one can show that the probability distribution of these polynomials
    will be independent of $q(\cdot) = F(0, y)$. That is, for every candidate degree-$d$ polynomial $q(\cdot)$, there exists
   some degree-$(d, t)$ bivariate polynomial, consistent with the row and column-polynomials learnt by $\Adv$. Intuitively this is because 
   $(d+1)(t+1)$ distinct points are required to uniquely determine $F(x, y)$, but 
   $\Adv$ learns at most $t(d + 1) + t$ distinct points on $F(x, y)$ through the polynomials of corrupt parties, 
   leaving $d + 1 - t$ ``degree of freedom" from the view-point of $\Adv$.
   This ensures the privacy of $\D$'s input.
    We next discuss
    how the $(\ESet, \FSet)$ sets are identified.

    Dealer first finds a $\STAR$ $(\CSet, \DSet)$ in its consistency-graph, proving that the row and column-polynomials
    of the honest parties in $\CSet$ and $\DSet$ respectively lie on $\starF(x, y)$.
     This follows from Lemma \ref{lemma:bivariate} and the fact that the
     degree-$d$ row-polynomials and the degree-$t$ column-polynomials of {\it honest} parties in
     $\CSet$ and $\DSet$ respectively are pair-wise consistent, where $\CSet$ and $\DSet$ has 
     $t + 1$ and $n - 2t = d + 1$ {\it honest} parties respectively.
    To find $(\ESet, \FSet)$, the dealer finds {\it additional} ``supportive parties"  whose row and
    column-polynomials also lie on $\starF(x, y)$. The idea is that for an {\it honest} $\D$, the row and column-polynomials of {\it all} honest
    parties lie on $F(x, y)$ and there are $n - t$ honest parties. To hunt for these additional supportive parties, $\D$ follows the following two-stage non-intuitive 
    approach.
    \begin{myitemize}
    \item It first tries to ``expand''  ${\DSet}$ by identifying  additional parties whose degree-$t$ column-polynomials 
  also lie on $\starF(x, y)$. The expanded set ${\FSet}$, includes all the parties having edges with at least 
 $2t+1$ parties from ${\CSet}$. Thus $\DSet \subseteq \FSet$.
 It is easy to see that  the column-polynomial of every $P_j \in \FSet$ lies on $\starF(x, y)$, as it
  will be pair-wise consistent with the row-polynomials of at least $t+1$ {\it honest} parties from
 $\CSet$, all of which lie on $\starF(x, y)$. 
 \item $\D$ then tries to ``expand'' ${\CSet}$ by searching for the parties $P_j$, who have
 an edge with at least $d+t+1$ parties from ${\FSet}$.
 This will guarantee that $P_j$'s degree-$d$ row-polynomial lies on $\starF(x, y)$, as it will be pair-wise consistent with the column-polynomials of at least
 $d + 1$ {\it honest} parties from $\FSet$, all of which lie on $\starF(x, y)$. 
 Such parties $P_j$ are included by $\D$ 
  in a set ${\ESet}$. Notice that the parties in ${\CSet}$ will satisfy the above condition and so $\CSet \subseteq {\ESet}$.
    \end{myitemize}
Once $\D$ finds $(\ESet, \FSet)$, it broadcasts them and then the parties verify whether
 indeed they satisfy the above conditions.  However there is a subtle issue, as 
   an {\it honest} $\D$ may have to wait {\it indefinitely}  for the ``expansion'' of $\DSet$ and
 $\CSet$ sets, beyond their initial cardinalities. For instance, let
  $n = 4t + 1$, $d = 2t$ and let $\CSet$ and $\DSet$ be 
 of size $2t+1$ and $3t+1$ respectively, containing $t$ corrupt parties.
 If the {\it corrupt} parties $P_i$ in $\CSet$  choose to be {\it inconsistent} with the parties  $P_j$ {\it outside}  $\DSet$ (by {\it not} broadcasting the $(\OK, i, j)$ messages), then
   the  {\it honest}  parties $P_j$ {\it outside} $\DSet$ will have edges with only $t+1$ parties from  $\CSet$ and will {\it not} be included in the set $\FSet$.
 So $\FSet$ will remain the same as $\DSet$.
  Similarly, the {\it corrupt} parties in $\FSet$  may choose to be inconsistent with the parties {\it outside}  $\CSet$, due to which
  $\ESet$ will remain the same as $\CSet$.
  To deal with the above, Patra et al.~\cite{PCR15}
   observed that if $\D$ is {\it honest} then eventually the {\it honest} parties 
   form a clique in the consistency-graph. Moreover, if the star-finding algorithm
 is executed on ``this'' instance of the consistency graph,
  then the $\CSet$ component of the obtained
  $\STAR$  will have at least $2t+1$ {\it honest} parties. Now
  if $\CSet$ contains at least $2t+1$ honest parties, then eventually $\DSet$ will expand to $\FSet$, which will contain all $n - t$ honest parties
 and eventually $\CSet$ will expand to $\ESet$ containing $n - t = 3t+1$ parties. 
 This crucial observation is at the heart of $\PCR\mbox{-}{\Sh}$.
 However, it is difficult for $\D$ to identify an instance of its dynamic consistency-graph that 
 contains a clique involving at least $n - t$ honest parties.
  The way-out is to repeatedly run the star-finding algorithm  
  and try the expansion of every instance of $\STAR$ $(\CSet, \DSet)$ obtained in 
  the consistency-graph. 

    \begin{schemesplitbox}{$\PCR$}{The perfectly-secure AVSS scheme of Patra, Choudhury and Rangan \cite{PCR15}.}{fig:PCR}
	\justify
\centerline{\algoHead{Sharing Phase: Protocol $\PCR\mbox{-}{\Sh}$}}
\begin{myitemize}
\item {\bf Stage I : Distributing Polynomials} --- same steps as $\BCG\mbox{-}{\Sh}$ except that the  {\it Shamir-sharing polynomial} $q(\cdot)$ is of degree-$d$,
 the bivariate polynomial $F(x, y)$ is of degree-$(d, t)$ and $F(x, 0) = q(\cdot)$.
\item {\bf Stage II : Pair-wise consistency checks and building consistency-graph} --- same as $\BCG\mbox{-}{\Sh}$.
 \item {\bf Stage III : Finding $(\ESet, \FSet)$ in the consistency-graph --- the dealer $\D$ does the following}
         \begin{myitemize}
	    \item[--] After every update in $G_D$, run the star-finding algorithm to check for the
	              presence of a $\STAR$.  Let there are $\alpha$ number of 
	              distinct $\STAR$s that are found till now in $G_D$, where $\alpha \geq 0$.
                      \begin{myitemize}
			         \item[--] If a new $\STAR$ $(\CSet_{\alpha +1}, {\DSet}_{\alpha +1})$ is found in $G_D$, 
 					then do the following:
					    \begin{myenumerate}
					         \item  Add $P_j$ to a set ${\FSet}_{\alpha+1}$
						             if $P_j$ has an edge with at least $2t+1$ parties from $ {\CSet}_{\alpha + 1}$ in $G_D$.
					        \item  Add $P_j$ to a set ${\ESet}_{\alpha + 1}$
					               if $P_j$ has an edge with at least $d+t+1$ parties from ${\FSet}_{\alpha + 1}$ in $G_D$.
					         \item For $\beta=1, \ldots, \alpha$, update the existing ${\FSet}_{\beta}$ and ${\ESet}_{\beta}$ sets as follows:
					                \begin{myitemize}
						                    \item Add $P_j$ to ${\FSet}_{\beta}$, if $P_j \not \in {\FSet}_{\beta}$
						                          and $P_j$ has an edge with at least $2t+1$ parties from ${\CSet}_{\beta}$ in $G_D$.
						                    \item Add $P_j$ to ${\ESet}_{\beta}$, if $P_j \not \in {\ESet}_{\beta}$
					                          and $P_j$ has an edge with at least $d+t+1$ parties from $ {\FSet}_{\beta}$ in $G_D$.
					                \end{myitemize}
				         \end{myenumerate}
		         \item[--] If no new $\STAR$ is obtained, then update the existing sets ${\FSet}_{\beta}, {\ESet}_{\beta}$ by executing the step $3$ as above.
	         \end{myitemize}
           \item[--]   Let $({\ESet}_{\gamma}, {\FSet}_{\gamma})$ be the first pair among the generated pairs $({\ESet}_{\beta}, {\FSet}_{\beta})$
         such that $|{\ESet}_{\gamma}| \geq 3t+1$ and $|{\FSet}_{\gamma}| \geq 3t+1$. Then
          broadcast $(({\CSet}_{\gamma}, {\DSet}_{\gamma}), ({\ESet}_{\gamma}, {\FSet}_{\gamma}))$.
      \end{myitemize}
 \item {\bf Stage IV : share computation --- each party $P_i$ does the following}
    \begin{myitemize}
    \item[--] If $(({\CSet}_{\gamma}, {\DSet}_{\gamma}), ({\ESet}_{\gamma}, {\FSet}_{\gamma}))$ is received from the broadcast of $\D$, {\it accept} it if 
     all the following hold.
         \begin{myitemize}
         \item[--] $|{\ESet}_{\gamma}| \geq 3t+1$ and $|{\FSet}_{\gamma}| \geq 3t+1$ and 
          $({\CSet}_{\gamma}, {\DSet}_{\gamma})$ is a $\STAR$ in the consistency-graph $G_i$.
         \item[--] Every party $P_j \in {\FSet}_{\gamma}$ has an edge with at least $2t+1$ parties from ${\CSet}_{\gamma}$ in $G_i$.
	\item[--] Every party $P_j \in {\ESet}_{\gamma}$ has an edge with at least $d+t+1$ parties from ${\FSet}_{\gamma}$ in $G_i$.
         \end{myitemize}
    \item[--] If $(({\CSet}_{\gamma}, {\DSet}_{\gamma}), ({\ESet}_{\gamma}, {\FSet}_{\gamma}))$ is accepted, then compute the {\it share} $s_i$ as follows and terminate.
           \begin{myitemize}
           \item[--] If $P_i \in \FSet_{\gamma}$, then set $s_i = g_i(0)$, where $g_i(y)$ is the degree-$t$ column-polynomial received from $\D$.
           \item[--] Else initialize $W_i$ to $\emptyset$. Upon receiving $f_{ji}$ from $P_j \in \ESet$, include $f_{ji}$ to $W_i$.
            Keep updating $W_i$ and keep executing $\OEC(W_i, t)$ till a degree-$t$ polynomial $g_i(y)$ is obtained. 
           Then set $s_i = g_i(0)$.  \\[.2cm]
           \end{myitemize}    
    \end{myitemize}
 \end{myitemize}
 
 \justify
 \centerline{\algoHead{Reconstruction Phase: Protocol $\PCR\mbox{-}{\Rec}$}}
Same steps as $\BCG\mbox{-}{\Rec}$, except that the parties now run $\OEC(\star, d)$ to recover a degree-$d$ polynomial.            
   \end{schemesplitbox}
  \subsection{$\CHP$ Scheme}
  Protocol $\BCG\mbox{-}{\Sh}$ requires a communication of $\Order(n^2 \log{|\F|})$ bits over the pair-wise channels, apart from
  the broadcast of $\Theta(n^2)$ $\OK$ messages and the broadcast of $\STAR$.
    The protocol generates $t$-sharing of a {\it single} secret. If $\D$ wants to $t$-share
    $\param$ secrets,  then it can invoke $\param$ 
    instances of $\BCG\mbox{-}{\Sh}$. This makes the
    {\it broadcast-complexity} (namely the number of bits to be broadcast) proportional to $\param$. Instead  Choudhury, Hirt and Patra \cite{CHP13} proposed a  
   modification of \footnote{It is easy to see that the communication complexity of $\PCR\mbox{-}{\Sh}$ is the same as $\BCG\mbox{-}{\Sh}$.}
    $\PCR\mbox{-}{\Sh}$, which
   allows $\D$ to $t$-share $\param$ secrets for {\it any} given $\param \geq n - 3t$, {\it without} incurring any additional communication complexity.
    The broadcast-complexity of $\CHP\mbox{-}{\Sh}$ will be {\it independent} of $\param$, which is a significant saving.
   This is because each instance of the broadcast in the asynchronous setting
    needs to be emulated by running the costly Bracha's ACast protocol. 
   
   We explain the idea of $\CHP\mbox{-}{\Sh}$ assuming $\param = n - 3t$. If $\param > n - 3t$, then $\D$ can divide its inputs into multiple batches
   of $n - 3t$ and invoke an instance of $\CHP\mbox{-}{\Sh}$ for each  batch. Recall that in 
      $\PCR\mbox{-}{\Sh}$, if $\D$ is {\it honest}, then the adversary's view 
    leaves $d + 1 - t$ ``degree of freedom'' in the degree-$(d, t)$ bivariate polynomial $F(x, y)$, where $t < d < n - 2t$.
   If we consider  the {\it maximum} value $d_{max}$ of $d$ which is $n - 2t - 1$, this implies
    $n - 3t$ degree of freedom. 
   While $\PCR\mbox{-}{\Sh}$ uses this degree of freedom for generating a $d_{max}$-sharing
   of a {\it single} secret by embedding a {\it single} degree-$d_{max}$ 
    sharing-polynomial in $F(x, y)$,
    $\CHP\mbox{-}{\Sh}$ uses it for $t$-sharing of $n - 3t$ values by embedding $n - 3t$ degree-$t$ Shamir-sharing polynomials
     in $F(x, y)$. 
          
      In a more detail, given $s^{(1)}, \ldots, s^{(n - 3t)}$ for $t$-sharing,
    $\D$ picks $n - 3t$ random degree-$t$ Shamir-sharing polynomials $q^{(1)}(\cdot), \ldots, q^{(n - 3t)}(\cdot)$, where
    $q^{(k)}(0) = s^{(k)}$. These polynomials are embedded in a degree-$(d_{max}, t)$ bivariate polynomial,
    which is otherwise a random polynomial, except that $F(\beta_k, y) = q^{(k)}(\cdot)$ holds.
    Here $\beta_1, \ldots, \beta_{n - 3t}$ are distinct, publicly-known non-zero elements from $\F$, different from the
     evaluation-points
    $\alpha_1, \ldots, \alpha_n$ (this requires $|\F| > 2n - 3t$).
    Notice that the embedding  and the degree of the sharing-polynomials are different in
    $\PCR\mbox{-}{\Sh}$ and $\CHP\mbox{-}{\Sh}$. Accordingly, the shares of the parties are different (see Fig \ref{fig:PCR-CHP-Comparison}).
    The {\it shares} of $P_i$ in $\CHP\mbox{-}{\Sh}$ will be $\{F(\beta_k, \alpha_i) \}_{k \in \{1, \ldots,n - 3t \}}$.
    And to compute them, $P_i$ should get its row-polynomial $f_i(x) = F(x, \alpha_i)$, as $P_i$ can then compute its shares by
    evaluating $f_i(x)$ at $x = \beta_1, \ldots, \beta_{n - 3t}$.
    
      \input{Figures/PCR-CHP-Comparison.tex}

    To achieve the above goal, we observe that if $\D$ invokes $\PCR\mbox{-}{\Sh}$ (with the above modifications)
     and if the protocol 
     terminates, then it ensures that
    $\D$ has ``committed" a degree-$(d_{max}, t)$ bivariate polynomial $\starF(x, y)$,
    such that each (honest) party $P_j$ possesses its column-polynomial $g_j(y) = \starF(\alpha_j, y)$.
    We also observe that for each row-polynomial $f_i(x) = \starF(x, \alpha_i)$, every $P_j$ holds a share 
    $\starF(\alpha_j, \alpha_i)$ in the form of $g_j(\alpha_i)$.
    Moreover, the degree of $f_i(x)$ is $d_{max} = n - 2t - 1$.
    Hence, if every party $P_j$ sends its share $g_j(\alpha_i)$ of $f_i(x)$ to $P_i$, then $P_i$ can reconstruct its desired row-polynomial $f_i(x)$
    by applying OEC on the received values. Hence the values
    $\vec S' =  (\starF(\beta_1, 0), \ldots, \starF(\beta_{n - 3t}, 0))$ will be $t$-shared, where  $\vec S' = (s^{(1)}, \ldots, s^{(n - 3t)})$ for an {\it honest} $\D$.      
 %

%% file: Figures/PCR-CHP-Comparison.tex
\begin{figure}[!h]
  \centering
  \begin{tikzpicture}[scale=0.85]
    \draw (0, 0) rectangle (7, 5);
    \draw (0, 2) -- (7, 2);
    \draw (0, 3) -- (7, 3);
    \draw (2.75, 0) -- (2.75, 5);
    \draw (4.25, 0) -- (4.25, 5);
    \draw [fill=yellow] (2.75, 2) rectangle (4.25, 3);
    \draw [rounded corners, fill=ForestGreen!30] (0, -2) rectangle (7, -1);
    \draw [rounded corners, fill=Blond!80] (8, 0) rectangle (11, 5);
    \draw [rounded corners, fill=BabyBlue!30] (12, 0) rectangle (15, 5);
    \draw [color=red, decorate,decoration={brace,amplitude=10pt,mirror,raise=4pt},yshift=0pt]
    (0, -2) -- (7, -2) node [red, midway,yshift=-0.8cm] {\footnotesize
      $[s]_{d_{max}}, s = q(0)$};
    \draw [color=Cobalt!80, decorate,decoration={brace,mirror,raise=4pt}, yshift=0pt]
    (8, 0) -- (11, 0) node [Cobalt!80,midway,yshift=-0.8cm] {\footnotesize $[s^{(1)}]_t$};
    \draw [color=Cobalt!80, decorate,decoration={brace,amplitude=10pt,mirror,raise=4pt},yshift=0pt]
    (12, 0) -- (15, 0) node [Cobalt!80,midway,yshift=-0.8cm] {\footnotesize $[s^{(\param)}]_t$};
    \draw (-1, 2.5) node {\footnotesize $f_i(x)$};
    \draw (-0.3, 2.5) node {\footnotesize $\Rightarrow$};
    \draw (-1, -1.5) node[text=red] {\footnotesize $q(\cdot)$};
    \draw (-0.3, -1.5) node[text=red] {\footnotesize $\Rightarrow$};
    \draw (3.5, 6) node {\footnotesize $g_j(y)$};
    \draw (3.5, 5.5) node {\footnotesize $\Downarrow$};
    \draw (9.5, 6) node[text=Cobalt!80] {\footnotesize $q^{(1)}(\cdot)$};
    \draw (9.5, 5.5) node[text=Cobalt!80] {\footnotesize $\Downarrow$};
    \draw (13.5, 6) node[text=Cobalt!80] {\footnotesize $q^{(\param)}(\cdot)$};
    \draw (13.5, 5.5) node[text=Cobalt!80] {\footnotesize $\Downarrow$};
    \draw (9.5, -2) node[text=Cobalt!80] {\footnotesize $s^{(1)} = q^{(1)}(0)$};
    \draw (13.5, -2) node[text=Cobalt!80] {\footnotesize $s^{(\param)} = q^{(\param)}(0)$};
    \draw (1, 4.5) node {\scriptsize $F(\alpha_1, \alpha_1)$};
    \draw (2.25, 4.5) node {\scriptsize $\ldots$};
    \draw (3.5, 4.5) node {\scriptsize $F(\alpha_j, \alpha_1)$};
    \draw (4.75, 4.5) node {\scriptsize $\ldots$};
    \draw (6, 4.5) node {\scriptsize $F(\alpha_n, \alpha_1)$};
    \draw (1, 3.5) node {\scriptsize $\rvdots$};
    \draw (2.25, 3.5) node {\scriptsize $\rvdots$};
    \draw (3.5, 3.5) node {\scriptsize $\rvdots$};
    \draw (4.75, 3.5) node {\scriptsize $\rvdots$};
    \draw (6, 3.5) node {\scriptsize $\rvdots$};
    \draw (1, 2.5) node {\scriptsize $F(\alpha_1, \alpha_i)$};
    \draw (2.25, 2.5) node {\scriptsize $\ldots$};
    \draw (3.5, 2.5) node {\scriptsize $F(\alpha_j, \alpha_i)$};
    \draw (4.75, 2.5) node {\scriptsize $\ldots$};
    \draw (6, 2.5) node {\scriptsize $F(\alpha_n, \alpha_i)$};
    \draw (1, 1.5) node {\scriptsize $\rvdots$};
    \draw (2.25, 1.5) node {\scriptsize $\rvdots$};
    \draw (3.5, 1.5) node {\scriptsize $\rvdots$};
    \draw (4.75, 1.5) node {\scriptsize $\rvdots$};
    \draw (6, 1.5) node {\scriptsize $\rvdots$};
    \draw (1, 0.5) node {\scriptsize $F(\alpha_1, \alpha_n)$};
    \draw (2.25, 0.5) node {\scriptsize $\ldots$};
    \draw (3.5, 0.5) node {\scriptsize $F(\alpha_j, \alpha_n)$};
    \draw (4.75, 0.5) node {\scriptsize $\ldots$};
    \draw (6, 0.5) node {\scriptsize $F(\alpha_n, \alpha_n)$};
    \draw (1, -0.5) node[text=red] {\scriptsize $\Downarrow$};
    \draw (2.25, -0.5) node[text=red] {\scriptsize $\ldots$};
    \draw (3.5, -0.5) node[text=red] {\scriptsize $\Downarrow$};
    \draw (4.75, -0.5) node[text=red] {\scriptsize $\ldots$};
    \draw (6, -0.5) node[text=red] {\scriptsize $\Downarrow$};
    \draw (1, -1.5) node[text=red] {\scriptsize $g_1(0)$};
    \draw (2.25, -1.5) node[text=red] {\scriptsize $\ldots$};
    \draw (3.5, -1.5) node[text=red] {\scriptsize $g_j(0)$};
    \draw (4.75, -1.5) node[text=red] {\scriptsize $\ldots$};
    \draw (6, -1.5) node[text=red] {\scriptsize $g_n(0)$};
    \draw (7.5, 4.5) node[text=Cobalt!80] {\scriptsize $\Rightarrow$};
    \draw (7.5, 3.5) node[text=Cobalt!80] {\scriptsize $\rvdots$};
    \draw (7.5, 2.5) node[text=Cobalt!80] {\scriptsize $\Rightarrow$};
    \draw (7.5, 1.5) node[text=Cobalt!80] {\scriptsize $\rvdots$};
    \draw (7.5, 0.5) node[text=Cobalt!80] {\scriptsize $\Rightarrow$};
    \draw (9.5, 4.5) node[text=Cobalt!80] {\scriptsize $f_1(\beta_1) = q^{(1)}(\alpha_1)$};
    \draw (9.5, 3.5) node[text=Cobalt!80] {\scriptsize $\rvdots$};
    \draw (9.5, 2.5) node[text=Cobalt!80] {\scriptsize $f_i(\beta_1) = q^{(1)}(\alpha_i)$};
    \draw (9.5, 1.5) node[text=Cobalt!80] {\scriptsize $\rvdots$};
    \draw (9.5, 0.5) node[text=Cobalt!80] {\scriptsize $f_n(\beta_1) = q^{(1)}(\alpha_n)$};
    \draw (11.5, 4.5) node[text=Cobalt!80] {\scriptsize $\ldots$};
    \draw (11.5, 3.5) node[text=Cobalt!80] {\scriptsize $\rvdots$};
    \draw (11.5, 2.5) node[text=Cobalt!80] {\scriptsize $\ldots$};
    \draw (11.5, 1.5) node[text=Cobalt!80] {\scriptsize $\rvdots$};
    \draw (11.5, 0.5) node [text=Cobalt!80] {\scriptsize $\ldots$};
    \draw (13.5, 4.5) node[text=Cobalt!80] {\scriptsize $f_1(\beta_\param) = q^{(\param)}(\alpha_1)$};
    \draw (13.5, 3.5) node[text=Cobalt!80] {\scriptsize $\rvdots$};
    \draw (13.5, 2.5) node[text=Cobalt!80] {\scriptsize $f_i(\beta_\param) = q^{(\param)}(\alpha_i)$};
    \draw (13.5, 1.5) node[text=Cobalt!80] {\scriptsize $\rvdots$};
    \draw (13.5, 0.5) node[text=Cobalt!80] {\scriptsize $f_n(\beta_\param) = q^{(\param)}(\alpha_n)$};
  \end{tikzpicture}
  \caption{Values distributed by $\D$ in $\PCR\mbox{-}{\Sh}$
    and $\CHP\mbox{-}{\Sh}$ on degree-$(d_{max}, t)$ polynomial $F(x, y)$, where
    $d_{max} = n - 2t - 1$. 
     In $\PCR\mbox{-}{\Sh}$, $s$
    is $d_{max}$-shared through $F(x, 0)$ (shown in red color), while
    in $\CHP\mbox{-}{\Sh}$, $s^{(1)}, \ldots, s^{(\param)}$ are $t$-shared through 
    $F(\beta_1, y), \ldots, F(\beta_{\param}, y)$ (shown in  blue color), where
    $\param = n - 3t$.}
  \label{fig:PCR-CHP-Comparison}
\end{figure}
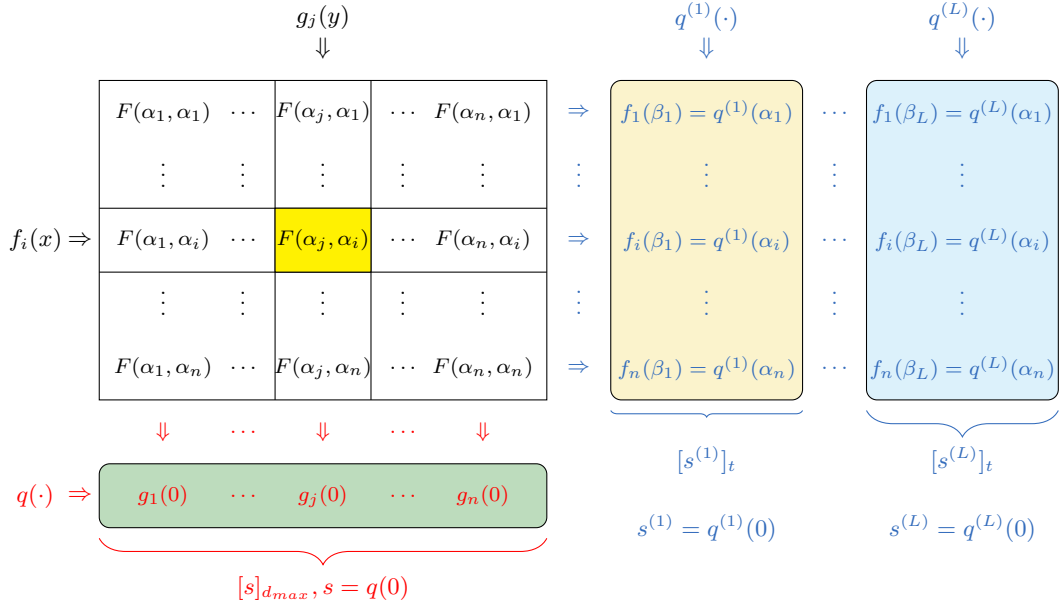

%% file: HPrelim.tex
\section{Preliminaries and Definitions for Hybrid Communication Setting}
\label{sec:HPrelim}
Even though the asynchronous model is practically more relevant compared to the synchronous setting,
 there are some inherent downsides in general with asynchronous protocols. 
 Synchronous protocols offer better {\it resilience}, compared to asynchronous protocols
  (for instance $t < n/4$ for AVSS compared to $t < n/3$ for VSS). Additionally, asynchronous protocols are more complex. 
  An inherent drawback of {\it asynchronous} MPC (AMPC) protocols is the lack of {\it input-provision}
  from {\it all honest} parties; i.e.~inputs of up to $t$ {\it honest} parties {\it may not} be considered for the computation \cite{BH07}.
  To get rid of these drawbacks, several forms of ``synchronization" have been considered 
  in the literature (\cite{DGKN09,BHN10,GPS19,MNR19,AMN0Y20,PR18}).
  One of these is the {\it hybrid} communication setting \cite{BH07,BHN10,PCR15,PR18,CH20,C20}, which is a ``mix" of synchronous and
  asynchronous setting. Namely, the first $R$ rounds are assumed to be synchronous, after which the network behaves {\it asynchronously}.
   There are several practical advantages of the hybrid setting compared to a completely asynchronous setting. 
   For instance, one can guarantee {\it input-provision} in MPC protocols. Moreover, one can design protocols with the 
   {\it same} resilience, as synchronous protocols, {\it without} letting the network to be synchronous for the {\it entire} duration of the protocol.
   Also, it is theoretically interesting to study the trade-off between network synchrony 
   and resilience, round and communication-complexity of various distributed-computing tasks. 
   In the context of VSS, Patra and Ravi \cite{PR18} have shown that
   one can design a Type-II perfectly-secure AVSS satisfying Definition \ref{def:AVSS} with $t < n/3$ in the {\it hybrid} setting (which otherwise would require $t < n/4$ in a
    {\it completely asynchronous} setting), 
   where $R = 1$. 
  Notice that the protocol has optimal resilience as well as it requires the optimal number of synchronous rounds.
  %
  %
%

We next define {\it weak polynomial sharing} (WPS), which is used in \cite{PR18}.
 The primitive allows a dealer $\D$ to distribute shares of a degree-$t$ polynomial held by $\D$. If $\D$ is {\it honest}, then every honest party
  should eventually terminate with its share. Moreover, even if $\D$ is {\it corrupt} and some honest party terminates $\Sh$, then it is ensured that
  $\D$ has distributed shares of some degree-$t$ polynomial to at least $t + 1$ honest parties.
  Property-wise, WPS is a {\it weaker} primitive than WSS, as it has only a sharing phase and {\it may not} allow even a ``weak reconstruction" of $\D$'s shared polynomial.  
\begin{definition}[\bf Weak Polynomial Sharing (WPS) \cite{PR18}]
  \label{def:AWPS}
  Let $\Sh$ be protocol for the $n$ parties in the hybrid setting, where a designated {\it dealer}
  $\D \in \Partyset$ has a degree-$t$ polynomial $f(x)$ 
  over $\F$ as input for $\Sh$. Then $\Sh$ constitutes a {\it perfectly-secure WPS}, if all
  the following hold.
   \begin{myitemize}
   \item {\bf Termination and Privacy}: Same as AVSS.
    \item {\bf Correctness}: If some honest party terminates $\Sh$, then there exists a degree-$t$ {\it weakly
       committed polynomial} $\starf(x)$ over $\F$ such that.
     \begin{myitemize}
     \item[--] If $\D$ is {\it honest}, then $\starf(x) = f(x)$ and each honest $P_i$ outputs
       $f(\alpha_i)$ at the end of $\Sh$.
     \item[--] If $\D$ is {\it corrupt}, then every honest $P_i$ outputs either $\starf(\alpha_i)$
       or some default value $\default$, with at least $t+1$ honest parties outputting $\starf(\alpha_i)$.
     \end{myitemize}
   \end{myitemize}
  \end{definition}    

%% file: HThreshold.tex
\section{Hybrid AVSS Protocol with $t < n/3$}
\label{sec:HThreshold}
In the AVSS schemes, the parties not receiving their shares from $\D$, deploy OEC
 to recompute their shares from the sub-shares received from the parties, who have received their shares from $\D$. 
 This inherently requires $n > 4t$. On contrary, the hybrid 
  AVSS of Patra et al.~\cite{PR18} is designed with $n > 3t$. The presence of a synchronous round at the beginning 
    simplifies certain aspects of verifiability and completely avoids the need for OEC. We first start with the WPS construction of \cite{PR18}, which 
    is similar to the $\KKKWSS\mbox{-}{\Sh}$ protocol.
  The dealer embeds its degree-$t$ polynomial in a random {\it symmetric} degree-$(t,t)$ bivariate polynomial and
  distributes its row-polynomials. In parallel, the parties pair-wise exchange random pads. Since the pads are exchanged during the {\it synchronous} round,
  each $P_i$ receives the pad selected for it by every other party, at  the end of the synchronous round. 
  The sent and received pads are also ``registered" with $\D$ for the comparison purpose.
  Based on this, $\D$ sends to $P_i$ its list of conflicting-parties $\Bad_i$, 
   who did not concur on the pads.
   Based on $\Bad_i$, party $P_i$ broadcasts its common values in a masked fashion for the parties who are not in $\Bad_i$
   and in an unmasked fashion for the parties who are in $\Bad_i$.
   The dealer then checks whether $P_i$'s public values are consistent with the bivariate polynomial and the pads registered with $\D$
   and accordingly includes $P_i$ to a set $\WCORE$. Once $\WCORE$ achieves the size of $2t+1$ (which {\it eventually} happens for an {\it honest} $\D$), 
   $\D$ broadcasts $\WCORE$. It is then publicly verified that no pair of parties in $\WCORE$  publicly conflicts over
    their supposedly common values that are either in padded or in clear form.
    
    If $\WCORE$ is successfully verified, then the row-polynomials of the (honest) parties in $\WCORE$ lie on a single degree-$(t, t)$
    symmetric bivariate polynomial $\starF(x, y)$. While every $P_i \in \WCORE$ can output the constant term of its row-polynomial
    as its share, any $P_i \not \in \WCORE$ tries to compute its row-polynomial by interpolating the common values with the parties in $\WCORE$. If
     $P_i$ has a conflict with any party in $\WCORE$, then the common value is
	publicly available. Else, it subtracts the pad it sent to that party in the synchronous round from 
      the padded value available publicly. If the interpolation does not give a degree-$t$ polynomial (which can happen only for a  {\it corrupt} $\D$),
      then $P_i$ outputs $\bot$.
\begin{protocolsplitbox}{$\AWPS$}{WPS protocol of Patra and Ravi  \cite{PR18}.}
  {fig:AWPS}
  \centerline{\algoHead{Synchronous Phase}}
  \begin{myitemize}
  \item {\bf Sending polynomials and exchanging random pads}:
    \begin{myitemize}
    \item[--] $\D$ with input $f(\cdot)$ chooses a random symmetric degree-$(t, t)$ bivariate
      $F(x, y)$ such that $F(0, y) = f(\cdot)$ and sends $f_i(x) = F(x, \alpha_i)$ to each party
      $P_i \in \Partyset$
    \item[--] Each party $P_i \in \Partyset$ picks a random pad $r_{ij}$ for every $P_j \in \Partyset$
      and sends $r_{ij}$ to $P_j$.
    \item[--] Each $P_i$ sends $\{r_{ij} \}_{P_j \in \Partyset}$ to $\D$. Let  $\{r^{(1)}_{ij} \}_{P_j \in \Partyset}$ be the list of {\it sent-pads} received by $\D$ from $P_i$. \\[0.05cm]
    \end{myitemize}
\end{myitemize}    
    \justify
    \centerline{\algoHead{Asynchronous Phase}}
\begin{myitemize}    
  \item {\bf Verifying masks}:
    \begin{myitemize}
    \item[--] Each $P_i$ sends the pads $\{r'_{ji}\}_{P_j \in \Partyset}$ received from various parties to $\D$. 
    Let $\{r_{ji}^{(2)}\}_{P_j \in \Partyset}$ be the list of {\it received-pads} which $\D$ receives from $P_i$.
    \item[--] Upon receiving $\{r_{ji}^{(2)}\}_{P_j \in \Partyset}$ from $P_i$, 
    dealer $\D$ sends to $P_i$ a set $\Bad_i = \{P_j :  r_{ji}^{(2)} \neq r_{ji}^{(1)}\}$.
    \end{myitemize}
  \item {\bf Broadcasting masked/unmasked common values --- each party $P_i$}: Broadcasts $(\mathcal{A}_i, \mathcal{B}_i, \Bad_i)$, where:
    \begin{myitemize}
    \item[--] $\mathcal{A}_i = \{a_{ij} = f_i(\alpha_j) + r_{ij}\}_{P_j \in \Partyset}$.
    \item[--] $\mathcal{B}_i = \{b_{ij}\}_{P_j \in \Partyset}$, where $b_{ij} = f_i(\alpha_j)$ if $P_j \in \Bad_i$
      and $b_{ij} = f_i(\alpha_j) + r'_{ji}$ otherwise.
    \end{myitemize}
  \item {\bf Constructing and broadcasting $\WCORE$ --- Dealer $\D$ does the following}:
    \begin{myitemize}
    \item[--] Upon receiving $(\mathcal{A}_i, \mathcal{B}_i, \Bad_i)$ from $P_i$, mark $P_i$ as {\it correct} and include in $\WCORE$, if all the following holds.
          \begin{myitemize}
      \item[--] $a_{ij} - r^{(1)}_{ij} = F(\alpha_j, \alpha_i)$
      \item[--] $b_{ij} = F(\alpha_j, \alpha_i)$ for all $P_j \in \Bad_i$ and
        $b_{ij} - r^{(2)}_{ji} = F(\alpha_j, \alpha_i)$ otherwise
      \item[--] $\Bad_i$ is the same set sent by $\D$ to $P_i$.
      \end{myitemize}
    \item[--] Wait until $|\WCORE| \ge 2t+1$ and then broadcast $\WCORE$.
     \end{myitemize} 
   \item {\bf Verifying $\WCORE$ --- each party $P_i$}:  {\it Accept} $\WCORE$ received from the broadcast of $\D$ if all the following hold.
         \begin{myitemize}
         \item[--] $|\WCORE| \ge 2t+1$ and $(\mathcal{A}_j, \mathcal{B}_j, \Bad_j)$ is received from the broadcast of each $P_j \in \WCORE$.
         \item[--] Every $P_j, P_k \in \WCORE$ are {\it pair-wise} consistent, as per the following conditions.
            \begin{myitemize}
	      \item[--] if $P_j \in \Bad_k$ and $P_k \in \Bad_j$ then $b_{jk} = b_{kj}$
	      \item[--] if $P_j \in \Bad_k$ and $P_k \not\in \Bad_j$ then $a_{kj} = b_{jk}$
	      \item[--] if $P_j \not\in \Bad_k$ and $P_k \in \Bad_j$ then $a_{jk} = b_{kj}$
	      \item[--] Else $a_{jk} = b_{kj}$ and $a_{kj} = b_{jk}$
	      \end{myitemize}
         \end{myitemize}
  \item {\bf Output stage --- each party $P_i$}: If $\WCORE$ is accepted, then terminate with output $s_i$, computed as follows.
    \begin{myitemize}
         \item[--] If $P_i \in \WCORE$ then set $s_i = f_i(0)$.
      \item[--] Else interpolate the  points
        $\{(\alpha_j, s_{ij})\}_{P_j \in \WCORE}$ where $s_{ij} = b_{ji}$ if $P_i \in \Bad_j$ and
        $s_{ij} = b_{ji} - r_{ij}$ otherwise. If the interpolation outputs a degree-$t$ polynomial $f_i(x)$ then 
        set $s_i =  f_i(0)$, else set $s_i = \bot$.
    \end{myitemize}
  \end{myitemize}
\end{protocolsplitbox}
\noindent \paragraph{\bf From WPS to VSS}
$\AWPS$ fails to serve as a VSS because if $\D$ is {\it corrupt}, then the parties outside $\WCORE$ may output $\bot$.
 Protocol $\PR\mbox{-}{\Sh}$ fixes this shortcoming. The protocol has two ``layers" of communication.
 The  first layer is similar to $\AWPS$ and identifies
  a set $\VCORE$ of $2t + 1$ parties whose row-polynomials lie on a single
 degree-$(t, t)$ symmetric bivariate polynomial $\starF(x, y)$. The second layer 
  enables the parties {\it outside} $\VCORE$ to obtain their row-polynomials lying on $\starF(x, y)$. In a more detail, every
 $P_j$ picks a random {\it blinding-polynomial} $r_j(\cdot)$ and shares it by invoking an instance $\AWPS_j$ of $\AWPS$. Additionally, 
 it makes public the polynomial $r_j(\cdot) + f_j(x)$. The idea is that if later $P_j$ is a part of $\VCORE$, then any
  $P_i \not \in \VCORE$ can compute the point $f_j(\alpha_i)$ (which is the same as $f_i(\alpha_j)$) on $P_i$'s row-polynomial, if $P_i$ obtains the output
  $r_j(\alpha_i)$ during $\AWPS_j$.  While an {\it honest} $P_j \in \VCORE$ makes public the correct $r_j(\cdot) + f_j(x)$, care has to be taken to ensure that even a 
  {\it corrupt} $P_j \in \VCORE$ has made public the correct polynomial. This is done as follows.
  First, each $P_k$ participates {\it conditionally} during $\AWPS_j$ only if
   the blinded polynomial of $P_k$ is consistent with respect to its received point on $r_k(\cdot)$ during $\AWPS_k$
   and the supposedly common value $f_k(j)$. Second, $P_j$ is included in $\VCORE$ only when 
   during $\AWPS_j$ the generated $\WCORE$ set $\WCORE_j$ is {\it accepted} and which has an overlap of 
   $2t + 1$ with $\VCORE$. 
   
  For a {\it corrupt} $P_j  \in \VCORE$,  an {\it honest}
   $P_i \not \in \VCORE$ may end up obtaining $\bot$ during $\AWPS_j$. However there will be
   at least $t + 1$ {\it honest} $P_j \in \VCORE$, corresponding to whom $P_i$ eventually obtains $r_j(\alpha_i)$ during $\AWPS_j$, using which
   $P_i$ obtains $t + 1$ points on $f_i(x)$, which are sufficient to compute $f_i(x)$.
\begin{protocolsplitbox}{$\PR\mbox{-}{\Sh}$}{The hybrid AVSS protocol of Patra and Ravi \cite{PR18}.}{fig:PR}
  \justify
  \centerline{\algoHead{Synchronous Phase}}
    $\D$ and parties execute the same steps as in the synchronous phase of $\AWPS$.
    Additionally, each $P_i$ picks a random degree-$t$ {\it blinding-polynomial} $r_i(\cdot)$ and as a dealer
    invokes an instance $\AWPS_i$ of $\AWPS$ to share $r_i(\cdot)$. Moreover, $P_i$ also participates in the synchronous
      phase of the instance $\AWPS_j$ for every $P_j \in \Partyset$. \\[.1cm]
    \centerline{\algoHead{Asynchronous Phase}}
   \begin{myitemize}
	  \item {\bf Verifying masks}: Parties and $\D$ execute the same steps as in $\AWPS$.
	      \item {\bf Broadcasting values --- each party $P_i$}: Broadcast $(\mathcal{A}_i, \mathcal{B}_i, \Bad_i, d_i(x))$, where
	      $\mathcal{A}_i, \mathcal{B}_i, \Bad_i$ are same as in $\AWPS$ and $d_i(x) = r_i(\cdot) + f_i(x)$.
	  \item{\bf Participating in $\AWPS$ instances --- each party $P_i$}: For $j = 1, \ldots, n$, participate in $\AWPS_j$
	             if a  degree-$t$ polynomial $d_j(x)$ is received from the broadcast of $P_j$ and
	             if $d_j(\alpha_i) = r_j(\alpha_i) + f_i(\alpha_j)$ holds.
	    \item {\bf Computing and Broadcasting $\VCORE$ --- the $\D$}: compute and broadcast $(\VCORE, \{\WCORE_i \}_{P_i \in \VCORE})$, such that:
	       \begin{myitemize}
	       \item[--] Every $P_i \in \VCORE$ is marked as {\it correct} (by satisfying the same conditions as in $\AWPS$).
	       \item[--] For every $P_i \in \VCORE$, the set $\WCORE_i$ is {\it accepted} during the instance $\AWPS_i$.
	       
	         \item[--] $|\VCORE| \geq 2t+ 1$ and $|\VCORE \cap \WCORE_i| \geq 2t+1$ for every $P_i \in \VCORE$.	       
	       \end{myitemize}
	        \item {\bf Verifying $\VCORE$ --- each party $P_i$}: {\it Accept} $(\VCORE, \{\WCORE_i \}_{P_i \in \VCORE})$ received
	        from the broadcast of $\D$, if:
	    		  \begin{myitemize}
		      \item[--] Every $P_j, P_k \in \VCORE$ are {\it pair-wise consistent} (using the same criteria as in $\AWPS$).
		      \item[--] For all $P_j \in \VCORE$, the set $\WCORE_j$ is {\it accepted} during the instance $\AWPS_j$.
		      \item[--] $|\VCORE| \ge 2t+1$ and for every $P_j \in \VCORE$, the condition $|\VCORE \cap \WCORE_j| \ge 2t+1$ holds.
		      \end{myitemize}
  \item {\bf Output stage --- each party $P_i$}: If $(\VCORE, \{\WCORE_i \}_{P_i \in \VCORE})$ is accepted then 
   terminate with output $s_i$, where:
         \begin{myitemize}
      \item If $P_i \in \VCORE$ then set $s_i = f_i(0)$.
      \item Else compute the output $r_{ji}$ in $\AWPS_j$ for every $P_j \in \VCORE$,
      interpolate degree-$t$ polynomial $f_i(x)$ through the points $\{(\alpha_i, s_{ij} = d_j(\alpha_i) - r_{ji}) \}_{P_j \in \VCORE \wedge r_{ji} \neq \bot}$
      and set $s_i = f_i(0)$.
      \end{myitemize}
  \end{myitemize}
\end{protocolsplitbox}